\documentclass[reqno, 11pt]{amsart}
\usepackage{amsaddr}
\usepackage{amssymb}
\usepackage[nesting]{hyperref}
\usepackage[pdftex]{graphicx}
\usepackage{listings}
\usepackage{multirow}
\usepackage{placeins}
\usepackage{color}
\usepackage{subfigure}
\usepackage{lscape}
\usepackage{dsfont}
\usepackage{amsmath}
\usepackage{amsfonts}
\usepackage{tikz}
\usepackage{verbatim}
\usepackage{accents}
\usepackage{natbib}  

\newcommand{\BlackBoxes}{\global\overfullrule5pt}

\BlackBoxes

\newcommand{\R}{\mathbb{R}} 
\newcommand{\Q}{\mathbb{Q}} 
\newcommand{\N}{\mathbb{N}} 

\newcommand{\PP}{\mathbb{P}}
\newcommand{\EE}{\mathbb{E}}

\definecolor{darkgreen}{rgb}{0.01, 0.74, 0.25}

\newtheorem{theorem}{Theorem}

\newtheorem{lemma}[theorem]{Lemma}

\theoremstyle{definition}
\newtheorem{example}[theorem]{Example}
\newtheorem{remark}[theorem]{Remark}

\numberwithin{equation}{section} \numberwithin{theorem}{section}
\def\0{\kern0pt\-\nobreak\hskip0pt\relax}

\makeatletter
\AtBeginDocument{ \def\@serieslogo{ \vbox to\headheight{ \parindent\z@ \fontsize{6}{7\p@}\selectfont
\vss}}}

\def\makeoverbar#1#2#3#4#5#6#7{ \setbox0=\hbox{$\m@th#2\mkern#5mu{{}#3{}}\mkern#6mu$} \setbox1=\null \dimen@=#4\fontdimen8#13 \dimen@=3.5\dimen@
\advance\dimen@ by \ht0 \dimen@=-#7\dimen@ \advance\dimen@ by \wd0
\ht1=\ht0 \dp1=\dp0 \wd1=\dimen@
\dimen@=\fontdimen8#13 \fontdimen8#13=#4\fontdimen8#13
\rlap{\hbox to \wd0{$\m@th\hss#2{\overline{\box1}}\mkern#5mu$}}
\fontdimen8#13=\dimen@}
\def\mylabel#1#2{{\def\@currentlabel{#2}\label{#1}}}
\makeatother
\begin{document}
\title[  ]{Optimal investment in ambiguous financial markets with learning}

\author[N. \smash{B\"auerle}]{Nicole B\"auerle$^*$}
\address[N. B\"auerle]{Institute of Stochastics,
Karlsruhe Institute of Technology (KIT), D-76128 Karlsruhe, Germany}

\email{\href{mailto:nicole.baeuerle@kit.edu}{nicole.baeuerle@kit.edu}}

\author[A. \smash{Mahayni}]{Antje Mahayni}
\address[A. Mahayni]{ Mercator School of Management, University of Duisburg–Essen, Lotharstr. 65, 47057, Duisburg, Germany}

\email{\href{amahayni@uni-due.de} {amahayni@uni-due.de}}

\thanks{${}^*$ Corresponding author}

\begin{abstract}
We consider the classical multi-asset Merton investment problem under drift uncertainty, i.e. the asset price dynamics are given by geometric Brownian motions with constant but unknown drift coefficients. The investor assumes a prior drift distribution  and is able to learn by observing the asset prize realizations during the investment horizon.
While the solution of an expected utility maximizing investor with constant relative risk aversion (CRRA) is well known, we consider the optimization problem under risk and ambiguity preferences by means of the KMM (\cite{klibanoff2005smooth})  approach. Here, the investor maximizes a double certainty equivalent. The inner certainty equivalent is  for given drift coefficient, the outer  is based on a drift
distribution.
Assuming also a CRRA type ambiguity function, it turns out that the optimal strategy can be stated in terms of the solution without ambiguity preferences but an adjusted drift distribution. To the best of our knowledge an explicit solution method in this setting is new. We rely on some duality theorems to prove our statements.\\
Based on our theoretical results, we are able to shed light on the impact of the prior drift distribution as well as the consequences of ambiguity preferences via the transfer to an adjusted  drift distribution, i.e. we are able to explain the interaction of risk and ambiguity preferences.
We compare our results with the ones in a pre-commitment setup where the investor is restricted to deterministic strategies. It turns out that (under risk and ambiguity aversion)
an infinite investment
horizon implies in both cases a maximin decision rule, i.e. the investor follows the  worst (best) Merton fraction (over all realizations of it) if she is more (less)  risk averse than a log-investor. We illustrate our findings with an extensive numerical study.   
\end{abstract}
\maketitle

\makeatletter \providecommand\@dotsep{5} \makeatother

\vspace{0.5cm}
\begin{minipage}{14cm}
{\small
\begin{description}
\item[\rm \textsc{ Key words}]
{\small portfolio optimization, learning, smooth ambiguity, duality theory, Bayesian investment problem}
\item[\rm \textsc{JEL classifications}] {\small C61, G11, D81}
\end{description}
}
\end{minipage}

\section{Introduction}
We investigate the effects of model ambiguity preferences on optimal investment decisions in a multi asset Black Scholes market. Since the seminal paper by  \cite{ellsberg1961risk}, we know that decision makers may have a non-neutral attitude towards model ambiguity. As a result, preferences are decomposed into risk preferences (based on known probabilities) and preferences concerning the degree of uncertainty about the (unknown) model parameters and are evaluated separately. This is in particular relevant for portfolio optimization problems. A recent literature suggests that model ambiguity is at least as prominent as risk in making investment decisions, see \cite{chen2002ambiguity}.

There are different ways to incorporate model ambiguity in decision making, like for example summarized in \cite{guidolin2013ambiguity}. In our setting, model ambiguity refers to the drift uncertainty in the dynamics of asset prices \footnote{It is well-known that the drift of stock prices is notoriously difficult to estimate, \cite{gennotte1986optimal}} and we apply the smooth ambiguity approach of \cite{klibanoff2005smooth} to deal with it. The risk in asset prices itself is evaluated by a utility function applied to the terminal wealth. Thus, the expected utility is itself a random variable (determined by the prior distribution of the drift parameters) which is evaluated by a second utility function (ambiguity function) capturing the model ambiguity. This approach allows for a separation of risk and ambiguity.  As a result we end up with a stochastic optimization problem over a nested expectation which leads to non-linear expectations and the fact that we cannot solve the problem with a standard HJB approach. As in \cite{balter2021time}  we assume that both the risk aversion and ambiguity aversion of the investor are described by (CRRA) power functions. While \cite{balter2021time} consider pre-commitment strategies, we take into account for the possibility that the investor is able to gradually learn about the drift by observing the asset price movements.

First of all, we contribute to the literature by analytically solving the portfolio optimization problem under drift uncertainty and learning while taking into account for both,  risk and ambiguity preferences. To the best of our knowledge this has not yet been achieved before in our setting. We exploit the fact that the norm-like functions which appear in the concatenations of the certainty equivalents for risk and ambiguity allow for a specific dual representation. On the one hand, such a result should have been expected from previous research about the representation of smooth ambiguity, see \cite{iwaki2014dual}, on the other hand it shows that smooth ambiguity is nothing else than a special kind of robust control. Indeed, we can think of the optimization problem as a classical Bayesian problem with adjusted prior probability distribution for the drift where the adjustment is computed in a second optimization problem. The attitude towards ambiguity depends on the relation of the levels of risk and ambiguity aversion (\cite{balter2021time}, \cite{iwaki2014dual}). Our quasi closed form solutions allow an in depth analysis of the interaction of risk and ambiguity preferences. Among others it turns out that with relatively increasing ambiguity aversion, the prior distribution is smoothly shifted from 'good' to 'bad' drift scenarios, i.e. an ambiguity averse decision maker is more pessimistic.

Second, we are able to determine the long-time behavior of an ambiguity neutral/averse Bayes investor in the multi asset case. This is much more challenging than in the single-asset case. There it has been shown in  \cite{bauerle2017extremal} that the worst/best drift is crucial. In the vector-valued case it is not clear which drift scenarios are worst and which best cases. It turns out that these extreme scenarios are determined by the Euclidean norm of the possible drift vectors. An investor who is more (less) risk averse than a log-investor \footnote{investor with logarithmic ultility}   then tends to a maximin (maximax) decision rule (which does not depend on the probability distribution of the drift). More precisely, an infinite investment horizon implies that the  more (less) risk averse investor initially acts as someone who knows that the drift belongs to the worst (best) case scenario.
Relying on this result we are able to compare the optimal strategy under learning with pre-commitment strategies, i.e. strategies where the investor is restricted to strategies which are deterministic functions of time.  We are able to explain why (compared to a pre-commitment strategy) the value of learning is rather low, even in the case of high investment horizons. For short investment horizons (or remaining investment horizons), the investor is not able to learn {\it{much}} about the drift. For long investment horizons, both investors are initially guided by the worst (best) scenario. Finally this observation also carries over to the investor with model ambiguity since she essentially behaves like a Bayesian investor with adjusted prior which does not play a role when a large time horizon is present. This is maybe expected since a large time horizon allows for perfectly learning the model.
\\

Related literature:\\
{\em Bayes optimization problems and their sensitivity:} We treat the fact that the asset drifts or market prices of risk are unknown as a Bayesian problem where we have a prior distribution (knowledge) about the values of the parameters which are here the drift or equivalently the market price of risk. The observations of the asset prices can then be used to update the belief which is also known as learning. This is done with the help of a filter. This filter then becomes part of the state process of the optimization problem. Techniques like these are well-known in finance, see e.g.
\cite{lakner1995utility} \cite{brennan1998role}, \cite{karatzas2001bayesian}, \cite{honda2003optimal}, \cite{rieder2005portfolio}, \cite{bjork2010optimal} for investment problems in Black Scholes markets.
The sensitivity of the optimal investment strategy in a single-asset Bayesian Black Scholes model w.r.t.\ model parameters is an interesting topic and investigated among others in  \cite{rieder2005portfolio}, \cite{longo2016learning}, \cite{bauerle2017extremal}.

{\em Duality and ambiguity:} One early approach to deal with ambiguity is to consider robust approaches, e.g.\ \cite{gilboa2004maxmin,hansen2001robust} and in continuous time control problems with ambiguous interest rates \cite{lin2021optimal} and in semimartingale markets \cite{schied2007optimal}, \cite{schied2009robust}. The latter survey paper discusses different robust formulations and their connection to risk measures.   Robust approaches care about worst-case scenarios and have sometimes been criticised for being too pessimistic. Thus, \cite{klibanoff2005smooth} introduced smooth ambiguity (KMM) which weights possible scenarios in a smooth way. Using Yaari's duality theory (\cite{yaari1987dual}), \cite{iwaki2014dual} already showed that there is a dual approach to smooth ambiguity which connects the 'smoothing function' to some distorted probabilities. The dual representation of entropic risk measures is used  in 
 \cite{bauerle2019markov} to tackle discrete-time decision problems with exponential utility for the KMM ambiguity applied to an unknown parameter. In  \cite{skiadas2003robust,skiadas2013smooth} the author relates a robust control to a recursive utility approach.

{\em Solving smooth ambiguity:} Smooth ambiguity models are by definition more complex than standard decision models under uncertainty. There are not many explicitly solved cases and approaches to tackle the problem. A static two-asset problem is considered and solved directly in \cite{gollier2011portfolio}.  In \cite{balter2021time}, the authors treat a Black-Scholes market with one risky asset and restrict to deterministic strategies which excludes learning. In  \cite{guan2022equilibrium} the authors look for equilibrium strategies in a smooth ambiguity problem with investment in a single-asset Black Scholes market and reinsurance and the mean-variance criterion. The recent paper \cite{glx} also considers equilibrium portfolio strategies for smooth ambiguity preferences.

{\em Learning and ambiguity:} The effect of learning under ambiguity has e.g.\ been investigated in \cite{epstein2007learning} among others with the help of dynamic variants of the Ellsberg problem. In  \cite{ju2012ambiguity} a generalized recursive smooth ambiguity and the effects on learning are considered. Both references deal with problems in discrete time.
\cite{suzuki2018continuous} generalizes the results in \cite{ju2012ambiguity} by taking a limit to continuous time. Whereas in \cite{baillon2018effect}  the effect of learning information on people’s attitudes toward ambiguity is investigated. \cite{miao2009ambiguity} considers optimal consumption and investment in a similar financial market with incomplete information, however apply a recursive multiple priors approach which immediately yields some kind of worst case problem over densities. Numerical results are not provided in the paper.

The outline of the paper is as follows.
First, Section \ref{sec_theory} states the optimization problem and its solution. We start with a multi-asset Black Scholes model where we set the interest rate to zero for simplicity. Then we review the classical Bayesian case and  explain how the problem with ambiguity can be solved analytically. A main tool is Sion's minimax theorem. Proofs are deferred to the appendix.  
Sec. \ref{sec:sensitivity} first analyzes the optimal strategy in the Bayesian model, since the optimal strategy in the ambiguous case boils down to this setting with adjusted prior. We discuss the behavior of the optimal investment strategy for short and long time horizon. The latter one being quite difficult to analyse. A particular focus is on the special case with two-point prior where we can represent the optimal investment strategy as a convex combination of Merton fractions which are optimal in the setting with complete information given by the two drift settings. We also consider the case with pre-commitment i.e. when the investment strategy has to be chosen at time $0$ and compare the two cases. Finally we simplify in the two-prior scenario the problem with model ambiguity. In Section \ref{sec_numerics} we present an extensive numerical study which sheds further light on our theoretical statements and gives some intuitive explanations. We restrict here to the single-asset case and two-point prior. The appendix contains proofs and further parameter constellations not discussed in the main sections.

\section{Optimization problem and solution}\label{sec_theory}
Let $(\Omega, \mathcal{F}, (\mathcal{F}_t), \mathbb{P})$ be a filtered probability space and $T>0$ be a finite time horizon. The underlying financial market consists of $d$ stocks and one riskless bond, each defined on the previously mentioned probability space. The price process $S=(S_1(t),\ldots,S_d(t))_{t\in [0,T]}$ of the $d$ stocks will for $i=1,\ldots ,d$ be given by
\begin{equation}
d S_i(t) = S_i(t) \left[ \mu_i dt + \sum_{j=1}^d \sigma_{ij} dW_j(t)\right] = S_i(t) \left[  \sum_{j=1}^d \sigma_{ij} dY_j(t)\right],
\end{equation}
where $W=(W_1(t),\ldots,W_d(t))^\top_{t\in [0,T]}$ is a d-dimensional Brownian motion, $\mu_i \in \R, \sigma_{ij} \in \R_+, i,j=1,\ldots,d$ and $\sigma=(\sigma_{ij})$ is regular. We further set $$ Y(t):= W(t)+\Theta t,\quad \Theta^\top := \sigma^{-1} \mu, \quad \mu:=(\mu_1,\ldots,\mu_d),  $$
where $\Theta$ denotes the market price per unit of risk.
The price process  of the riskless bond is for simplicity assumed to be identical to $1$.

We further assume that $\mu$   is not known and thus a random variable. This implies that the market price of risk $\Theta$ is also not known to the investor.  However, she has a prior knowledge about $\Theta$ in form of a prior distribution $\PP$ on $\R^d$. For the numerical part we assume that the random variable $\Theta$  may take only one of the values $\vartheta_1,\ldots ,\vartheta_m$ with $\PP(\Theta=\vartheta_i)=p_i$.

In a next step we introduce a suitable set of trading strategies. Since the riskless bond is equal to 1 and we only consider self-financing strategies we can express the wealth process with the help of the investment in risky assets only. By $\pi = (\pi_1,\ldots,\pi_d)$ we denote a $d$-dimensional stochastic process representing the trading strategy of some investor, where $\pi_k(t)$ describes the amount   invested  in the $k$-th stock  at time $t\in [0,T]$.  We denote by
$$ \mathcal{F}^Y(t) := \sigma(Y(s), 0\le s\le t), \quad \mathcal{F}^Y:= (\mathcal{F}^Y(t))$$
the filtration generated by $Y$ which is equivalent to the filtration generated by $S$. Strategies $\pi$ should be  $\mathcal{F}^Y$-progressively measurable. This means that the agent is able to observe the stock prices and updates the belief about the market price of risk from this observation. In other words, the agent is able to {\em learn} the right market price of risk.
The associated wealth process denoted by $(X^\pi_t)_{t\in [0,T]}$ is given by
\begin{equation}
d X^\pi_t = \sum_{k=1}^d \pi_k(t) \frac{\mathrm{d}S_k(t)}{S_k(t)} =\pi(t) \sigma dY(t)\label{eq sde price process}
\end{equation}
with initial  capital $x_0 \in \R$. In what follows let
$$ u(x) = \frac1\alpha x^\alpha, \quad \alpha<1, \alpha\neq 0$$
be a CRRA utility function. The absolute and relative risk aversion is here given by
\begin{align*}
R_A(x):=-\frac{u''(x)}{u'(x)}= \frac{1-\alpha}{x} \text{ and } R_R(x):= x R_A(x)=1-\alpha.
\end{align*}
It is well-known that the limiting case $\alpha\to 0$ corresponds to the logarithmic utility. 

\subsection{The classical Bayesian case}\label{sec:classicBayes}
The investor aims to maximize her expected utility of terminal wealth. First we assume that the investor is ambiguity-neutral w.r.t.\ the unknown parameter and consider
\begin{equation}\label{eq:Bayesproblem}
V(x_0)= \sup_\pi \int \EE_\vartheta [u(X_T^\pi)] \PP(d\vartheta)
\end{equation}
where the supremum is taken over all $\mathcal{F}^Y$-adapted strategies $\pi$ for which the stochastic integral and the expectations are defined and $X_T^\pi\ge 0$. We denote this set by $\mathcal{A}.$ $\EE_\vartheta$ is the conditional expectation, given $\Theta=\vartheta.$
This problem is the well-known Bayesian adaptive portfolio problem. We summarize its solution in the following theorem (\cite{karatzas2001bayesian,rieder2005portfolio}) (where $\|\cdot\|$ is the usual Euclidean norm):

\begin{theorem}\label{theo:Bayes}
The maximal expected utility attained in \eqref{eq:Bayesproblem} is given by
\begin{equation}
V(x_0) = \frac{x_0^\alpha}{\alpha} \left( \int_{\R^d} F(T,z)^\gamma \varphi_T(z) dz\right)^{1/\gamma}, \quad x_0>0
\end{equation}\label{eq:FL}
where $\gamma = 1/(1-\alpha)$, $\varphi_T$ is the density of the $d$-dimensional normal distribution $\mathcal{N}(0,TI)$ ($I$ being the identity matrix) and for $t\ge 0, z\in\R^d, \vartheta\in\R$ we define
\begin{equation}
F(t,z):= \int L_t(\vartheta,z) \PP(d\vartheta), \quad  L_t(\vartheta,z) := \left\{ \begin{array}{ll}
       \exp{\big(z\cdot \vartheta-\frac12 \|\vartheta\|^2 t\big)}  &  t>0\\
     1    & t=0.
   \end{array}\right.\;  
\end{equation}
The optimal fractions invested in the stocks are for $t\ge 0$ given by
\begin{equation}\label{eq:Bayesfraction}
\frac{\pi^*(t)}{X^*(t)}
=\gamma (\sigma^\top)^{-1}  \frac{ \int_{\R^d} \nabla F(T,z+Y(t))(F(T,z+Y(t))^{\gamma-1}\varphi_{T-t}(z) dz}{\int_{\R^d} (F(T,z+Y(t))^{\gamma}\varphi_{T-t}(z) dz}
\end{equation}
where $X^*$ is the wealth process under the optimal strategy $\pi^*.$
\end{theorem}

\begin{remark}\label{rem:Bayes_sol}
\begin{itemize}
    \item[a)] Recall in particular that in case the market price of risk is known and is equal to $\vartheta,$ the optimal fractions which have to be invested do not depend on time and wealth and are given by
$$ \frac{\pi^*(t)}{X^*(t)}=\kappa^{\text{Mer}}(\gamma,\vartheta):=  \gamma (\sigma^\top)^{-1} \vartheta.$$
This is a special case of our model when the prior distribution is concentrated on $\vartheta.$ It is often named  {\em Merton fractions} of a CRRA investor with a level of relative risk aversion $\frac{1}{\gamma}=1-\alpha$.
 Indeed this fraction can  be recovered from  \eqref{eq:Bayesfraction} as follows: Since  the prior is concentrated on $\vartheta$ we obtain that
$F(t,z)=  L_t(\vartheta,z) $ and thus
\begin{equation}
\frac{\pi^*(t)}{X^*(t)}
=\gamma (\sigma^\top)^{-1}  \frac{ \int_{\R^d}  \vartheta F(T,z+Y(t))(F(T,z+Y(t))^{\gamma-1}\varphi_{T-t}(z) dz}{\int_{\R^d} (F(T,z+Y(t))^{\gamma}\varphi_{T-t}(z) dz }= \gamma (\sigma^\top)^{-1} \vartheta.
\end{equation}
    \item[b)]  In case $\alpha\to 0$ which can be interpreted as the logarithmic utility we obtain that the optimal fractions invested in the stocks are for $t\ge 0$ given by
    \begin{equation}\label{eq:Bayesfraction_log}
\frac{\pi^*(t)}{X^*(t)}
= (\sigma^\top)^{-1}  \hat \Theta_t =  (\sigma^\top)^{-1} \EE[\Theta|\mathcal{F}^Y(t)] =  (\sigma^\top)^{-1}\frac{  \nabla F(t,Y(t))}{F(t,Y(t))}.
\end{equation}
  This is sometimes called {\em certainty equivalence principle} since the unknown market price of risk in the Merton fractions is simply replaced by its conditional expectation, given the information so far. The fraction here does not depend on the time horizon of investment. Moreover for discrete distribution $\PP$ on $\vartheta_1,\ldots ,\vartheta_m$ with $\PP(\Theta=\vartheta_i)=p_i$, it holds that
  $$ \PP(\Theta = \vartheta_i| \mathcal{F}^Y(t)) = \frac{  p_i L_i(\vartheta_i,Y(t))}{F(t,Y(t))}, \quad i=1,\ldots,m$$
  is the conditional distribution of $\Theta$, \cite{karatzas2001bayesian,rieder2005portfolio}. Also note that the expectation of $\hat\Theta_t$ remains constant over time and is thus equal to the expectation of the prior, since the process is a martingale by construction.    
\end{itemize}
\end{remark}

\subsection{The case with model ambiguity concerns}
Now we are interested in an investor who takes model ambiguity into account, i.e. instead of problem \eqref{eq:Bayesproblem} we consider for a second utility function $$v(x)=\frac1\lambda\; x^\lambda, \lambda <1,\lambda \neq 0$$ the problem (see e.g.\ \cite{balter2021time})
\begin{eqnarray}\nonumber
 && \sup_{\pi\in\mathcal{A}}  v^{-1} \int v \circ u^{-1} \EE_\vartheta [u(X_T^\pi)] \PP(d\vartheta)\\ \label{eq:Aproblem}
 &=& \sup_{\pi\in\mathcal{A}}  \left(  \int \left( \EE_\vartheta [(X_T^\pi)^\alpha]\right)^{\lambda/\alpha} \PP(d\vartheta)\right)^{1/\lambda}
\end{eqnarray}
This means that model ambiguity, represented by an uncertain market price of risk, is evaluated with a second utility function $v$ which is here of the same form but with possibly different parameter.
In  case $\alpha >0$ problem \eqref{eq:Aproblem} is equivalent to
\begin{equation}\label{eq:Aproblem1}
 \sup_{\pi\in\mathcal{A}}  \left(  \EE \left[\big( \EE_\Theta [(X_T^\pi)^\alpha]\big)^{\lambda/\alpha} \right]\right)^{\alpha/\lambda},
\end{equation}
in case $\alpha<0$ it is equivalent to 
\begin{equation}\label{eq:Aproblem2}
 \inf_{\pi\in\mathcal{A}}  \left(  \EE \left[\big( \EE_\Theta [(X_T^\pi)^\alpha]\big)^{\lambda/\alpha} \right]\right)^{\alpha/\lambda},
\end{equation}
In case $\alpha=\lambda$ the problem reduces to the Bayesian problem discussed previously. Thus, if model risk and the market risk is evaluated with the same parameter we are back in the setting of Section \ref{sec:classicBayes}. In what follows we restrict the discussion to the case $\alpha,\lambda\in(0,1).$ The cases where at least one of the parameters is negative are similar and discussed in the appendix. 

\subsubsection{The ambiguity loving case}
Let us now assume that $\lambda> \alpha>0$ and define $\mathbf{p} := \lambda/\alpha>1$. The economic interpretation is that the agent is less concerned about model ambiguity than about the risk in the stock market itself. In this case by using the $L^\mathbf{p}$ norm $\|\cdot\|_\mathbf{p}$ we can write  problem \eqref{eq:Aproblem1} as
\begin{equation}\label{eq:Aproblem3}
 \sup_{\pi\in\mathcal{A}}  \left\| \EE_\Theta [(X_T^\pi)^\alpha]\right\|_\mathbf{p}
\end{equation}
where the norm is w.r.t. $\Theta$. It is well-known that the $L^\mathbf{p}$ norm has the following dual representation for a r.v. $X\ge 0$, where $1/\mathbf{p}+1/\mathbf{q}=1$ (see e.g.\ \cite{rudin1991functional}):
\begin{lemma}\label{lem:duality_1}
If $\mathbf{p} := \lambda/\alpha>1$ we obtain for non-negative $X\in L^\mathbf{p}$
    \begin{equation}\label{eq:dual_case1}
\| X\|_\mathbf{p} = \sup\left\{ \int X d \Q : \left\| \frac{d\Q}{d\PP}\right\|_\mathbf{q} \le 1 \right\}.
\end{equation}
where on the right-hand side of \eqref{eq:dual_case1} the supremum is taken over all measures $\Q$ (not necessarily probability measures) which are absolutely continuous w.r.t.\ $\PP$ and satisfy the constraint. Moreover, an optimal measure $\Q^*$ exists. 
\end{lemma}

For a random variable $X$ with values $\{x_1,\ldots,x_m\}$ and corresponding probabilities $ p_1,\ldots, p_d$ we can thus write
\begin{equation}\label{eq:dual:discrete}
\left( \sum_{i=1}^m x_i^\mathbf{p} p_i\right)^{1/\mathbf{p}} = \sup\left\{ \sum_{i=1}^m x_i q_i  : \sum_{i=1}^m \left(\frac{q_i}{p_i}\right)^\mathbf{q} p_i \le 1, q_i\ge 0 \right\}.
\end{equation}
In what follows define the set of measures $\mathfrak{Q}$
as the set of measures which satisfy the constraints in \eqref{eq:dual_case1}.
This gives immediately rise to the following solution algorithm for our problem:

\begin{theorem}\label{theo:main0}
In the model of this subsection we have
\begin{eqnarray}\nonumber
 \sup_{\pi\in\mathcal{A}} \left\| \EE_\Theta [(X_T^\pi)^\alpha]\right\|_\mathbf{p} &=&
 \sup_{\pi\in\mathcal{A}} \sup_{\Q \in \mathfrak{Q}} \int \EE_\vartheta [(X_T^\pi)^\alpha] \Q(d\vartheta) =  \sup_{\Q \in \mathfrak{Q}} \sup_{\pi\in\mathcal{A}}   \int \EE_\vartheta [(X_T^\pi)^\alpha]\Q(d\vartheta)\\ \label{eq:Aproblem4}
 &=& \int \EE_\vartheta [(X_T^{\pi^*})^\alpha]\Q^*(d\vartheta).
\end{eqnarray}
After normalizing $\Q$, the inner optimization problem is however, exactly the Bayesian portfolio problem of the previous section with distribution $\tilde\Q:= \Q/ \Q(\R)$ for the unknown parameter.  So solving \eqref{eq:Aproblem1} boils down to solving the classical Bayesian portfolio problem first with value given in Theorem \ref{theo:Bayes} and then in a second step  finding the optimal prior distribution implied by $\Q^*$ which is obtained from the outer optimization problem. The optimal strategy $\pi^*$ is then the one in Theorem \ref{theo:Bayes} with $\PP$ replaced by $\Q^*.$
\end{theorem}

\subsubsection{The ambiguity averse case}
Let us now assume that $\alpha >\lambda> 0$, i.e. the agent is more concerned about model ambiguity than about the risk in the financial market. This case is slightly more complicated.  Define again $\mathbf{p} := \lambda/\alpha<1$ and $\mathbf{q}$ by $1/\mathbf{p}+1/\mathbf{q}=1$.   Note that $\mathbf{q}<0$. We obtain (see Appendix \ref{app:duality}):

\begin{lemma}\label{lem:duality2}
If $\mathbf{p} := \lambda/\alpha<1$   we obtain for non-negative $X\in L^1 $ 
\begin{equation}\label{eq:dual_caseb}
\left(\int  X^\mathbf{p} d\PP\right)^{1/\mathbf{p}} = \inf\left\{\int X d \Q : \Big(\int \Big(\frac{d\Q}{d\PP}\Big)^\mathbf{q} d\PP\Big)^{1/\mathbf{q}} \ge 1  \right\}.
\end{equation}
where on the right-hand side of \eqref{eq:dual_caseb} the infimum is taken over all measures $\Q$ (not necessarily probability measures) which are absolutely continuous w.r.t.\ $\PP$. Moreover, an optimal measure $\Q^*$ exists. 
\end{lemma}
In what follows we denote by $\mathfrak{Q}'$  the set of all measures which satisfy the constraints in \eqref{eq:dual_caseb}
Hence we can write our optimization problem as
\begin{equation}\label{eq:Aproblem5}
 \sup_{\pi\in\mathcal{A}}  \left(  \EE \left( \EE_\Theta [(X_T^\pi)^\alpha]\right)^{\mathbf{p}} \right)^{\frac{1}{\mathbf{p}}} =  \sup_{\pi\in\mathcal{A}} \inf_{\Q \in \mathfrak{Q}'} \int \EE_\vartheta [(X_T^\pi)^\alpha] \Q(d\vartheta).
\end{equation}
We have to show next that we can interchange the infimum and supremum in order to solve the problem as in the previous case. We proceed as in \cite{bauerle2019markov} and use the following minimax theorem (see \cite{sion1958general}):

\begin{theorem}\label{theo:minimax}
Let $M$  be  any space and $O$ be a compact space,  $h$ a function on $M\times O$ that is concave-convexlike. If  $h(x,y)$ is lower semi-continuous in $y$ for all $x\in M$ then $$\sup_x\inf_y h(x,y)= \inf_y \sup_x h(x,y).$$
\end{theorem}

\begin{theorem}\label{theo:main2}
In the model of this subsection where $\PP$ has compact support we have
\begin{equation}
    \sup_{\pi\in\mathcal{A}} \inf_{\Q \in \mathfrak{Q}'} \int \EE_\vartheta [(X_T^\pi)^\alpha] \Q(d\vartheta) =  \inf_{\Q \in \mathfrak{Q}'}  \sup_{\pi\in\mathcal{A}} \int \EE_\vartheta [(X_T^\pi)^\alpha] \Q(d\vartheta)=\int \EE_\vartheta [(X_T^{\pi^*})^\alpha]\Q^*(d\vartheta).
\end{equation}
The solution procedure is then as in Theorem \ref{theo:main0} by first solving the inner and then the outer problem.
\end{theorem}

\begin{proof}
We apply Theorem \ref{theo:minimax} to the function $L :  \mathcal{A}\times \mathfrak{Q}' \to \R$ defined by
$$L(\pi,\Q)=  \int \EE_\vartheta [(X_T^\pi)^\alpha]\Q(d\vartheta).$$
where $\mathcal{A}$ is the set of admissible strategies. The set $\mathfrak{Q}'$ is not compact, however since we want to minimize over this set we can first restrict to a set $\mathfrak{\tilde Q}$ where the elements $\Q$ have a mass bounded by a constant $K>0$. The proof of Lemma \ref{lem:duality2} shows that an optimal $\Q^*$ with finite mass exists.  Next the mapping $\Q\mapsto L(\Q,\pi)$ is linear and continuous w.r.t.\ weak convergence. Finally, since $X_T^\pi$ is linear in $\pi$ and $x^\alpha$ is concave, the mapping $\pi\mapsto L(\Q,\pi)$ is concave. Hence the assumptions are satisfied and  Theorem \ref{theo:minimax} implies that we are able to interchange  the infimum and supremum.
\end{proof}

\begin{remark}
For simplicity we assume that our financial market is complete. However the approach to solve the smooth ambiguity problem by making use of a dual representation of the target function is independent of the precise formulation of the market. For example we could also solve  smooth optimization problems in incomplete markets like in \cite{karatzas1991martingale} in the same way with the drawback that the solution of the inner optimization problem is more complicated and does not allow for such a nice formula as in Theorem \ref{theo:Bayes}.    
\end{remark}

\begin{remark}
    When we choose as utility function for our investor $u(x)=-e^{-\alpha x}, \alpha>0$ and for the ambiguity $v(x)=-e^{-\lambda x}, \lambda>0$ then a quick calculation shows that $v\circ u^{-1}=- (-x)^{\lambda/\alpha}$ and we can apply the same solution method in this case. However, the idea of the duality approach is more general. For example suppose that ambiguity aversion is smoothly measured via the negative of the entropic risk measure $- \frac{1}{\alpha} \ln ( \EE [e^{-\alpha  X}]), \alpha >0$ and our decision maker is risk neutral. Then we would like to maximize
    $$-\frac{1}{\alpha} \ln \int  e^{-\alpha \EE_\vartheta [X_T^\pi] }\Q(d\vartheta) = \inf_\Q \left\{ \int \EE_\vartheta [X_T^\pi]  \Q(d\vartheta) +\frac1\alpha I(\Q || \PP) \right\}  $$
    where the infimum is taken over all probability measures $\Q$ which are absolutely continuous w.r.t.\ $\PP$ and $I$ is the usual relative entropy. This representation gives again rise to a decomposition of the problem in an inner and outer problem, \cite{bauerle2019markov}. Representations like this are well-known for convex risk measures, \cite{schied2007optimal}. 
\end{remark}

\section{Sensitivity and comparison results}\label{sec:sensitivity}
In the previous section we have seen that solving the optimization problem with smooth ambiguity boils down to solving the classical Bayesian investment problem with the 'right' prior distribution for the unknown market price of risk. In particular, as far as sensitivity results are concerned, the model behaves as the classical Bayesian problem. Thus, we first prove some sensitivity results for the classical case, i.e. for the optimal fractions in Theorem \ref{theo:Bayes} which we denote by 
\begin{equation*}
    \kappa(t,T,Y(t)) := \frac{\pi^*(t)}{X^*(t)}.
\end{equation*}

\subsection{Limiting results for optimal investment fraction - General discrete prior}
In what follows, suppose that the prior distribution $\PP$ is concentrated on $\vartheta_1,\ldots ,\vartheta_m$ with $\PP(\Theta=\vartheta_k)=p_k, k=1,\ldots,m$. The
the $i$-the row of $(\sigma^\top)^{-1}$ will be denoted by $(\sigma^\top)^{-1}_i$ and $ \|\cdot\|$ is the usual Euclidean norm. When we assume that $p_k>0$ for all $k$, we obtain the following results:

\begin{theorem}\label{theo:sensitivity}
In the Bayesian model  it holds for the optimal investment fractions:
\begin{itemize}
    \item[a)] for all $y\in \R:$
    $$\lim_{T\to 0} \kappa(0,T,y) =\gamma  (\sigma^\top)^{-1} \sum_{k=1}^m  \vartheta_k p_k.$$
     \item[b)] when $\|\vartheta_1\| <\ldots < \|\vartheta_m\|$ and $\alpha>0$  then for all $t\ge 0$, $y\in \R:$
    $$ \lim_{T\to\infty}  \kappa(t,T,y)=\gamma  (\sigma^\top)^{-1} \cdot \vartheta_m = \kappa^{\text{Mer}}(\gamma,\vartheta_m).   $$
     \item[c)] when $\|\vartheta_1\| <\ldots < \|\vartheta_m\|$ and $\alpha<0$  then for all $t\ge 0$, $y\in \R:$
      $$ \lim_{T\to\infty}  \kappa(t,T,y)= \gamma (\sigma^\top)^{-1} \cdot \vartheta_1 =\kappa^{\text{Mer}}(\gamma,\vartheta_1).   $$
   \item[d)]  for all $t\ge 0$, $y\in \R$ and $i=1,\ldots,d:$
       $$\gamma \min_k\{ (\sigma^\top)^{-1}_i \cdot \vartheta_k\} \le   \kappa_i(t,T,y) \le  \gamma \max_k \{(\sigma^\top)^{-1}_i \cdot \vartheta_k\}.$$
    \item[e)]  for all $T>0$:
       $$ \lim_{\alpha \downarrow 0} \kappa(0,T,0) =(\sigma^\top)^{-1} \sum_{k=1}^m  \vartheta_k p_k.$$     
\end{itemize}
\end{theorem}

A proof of this theorem can be found in  Appendix \ref{app:sensi}. Recall again that the limiting case $\alpha\to 0$ corresponds to the logarithmic utility (cp.\ Rem.\ \ref{rem:Bayes_sol} b)) and  the unknown market price of risk $\Theta$ in the  Merton fraction is simply replaced by its expectation e). 
The same is true, when the time horizon is short, see a). For a large time horizon,  b), c) investors tend to extreme Merton fractions. In general the optimal investment fractions can be bounded by the entries in the Merton fractions, d). 

\begin{remark}
  The statements are substantially easier to formulate and prove in the case when we have only one stock. In this case part a), d) and e)  are special cases of Theorem 9 in \cite{rieder2005portfolio}, part b) and c) are Theorem 3.1/3.3 in \cite{bauerle2017extremal}. 
  Further sensitivity results in the one risky asset case can be found in \cite{longo2016learning}.
\end{remark}

\subsection{Representation for optimal investment fraction - Two-point prior}

Assume now further that $\mu=(\mu_1,\ldots,\mu_d)$ can take only two  possible (arbitrary) values $\bar\mu=(\bar\mu_1,\ldots ,\bar\mu_d)$ and $\underline{\mu}=(\underline{\mu}_1,\ldots,\underline{\mu}_d)$.  We assume that
\begin{equation}\label{eq:apriori_distributionmulti}
     \PP(\mu=\bar\mu)=p \text{ and } \PP(\mu=\underline\mu)=1-p, \quad p\in(0,1).
\end{equation}The optimization problem \eqref{eq:Bayesproblem} boils down to
\begin{equation}\label{eq:Bayesproblem_examplem}
V(x_0)= \sup_{\pi\in\mathcal{A}} \; p \EE_{\bar\mu} [u(X_T^\pi)] + (1-p) \EE_{\underline\mu} [u(X_T^\pi)].
\end{equation}
where $u(x)=\frac1\alpha x^{\alpha}.$ In this case the function $F$ appearing in \eqref{eq:FL} is given by  $ F(t,z)  = p L_t(\bar\vartheta,z)+(1-p) L_t(\underline\vartheta,z)$. Let us further denote
\begin{equation*}
    \hat F(t,T,Y(t)) :=\frac{ \int_{\R}
L_T(\underline\vartheta,z+Y(t))
 (F(T,z+Y(t)))^{\gamma-1}\varphi_{T-t}(z) dz}
 {\int_{\R} (F(T,z+Y(t))^{\gamma}\varphi_{T-t}(z) dz}.
\end{equation*}
Then we can express the solution of \eqref{eq:Bayesproblem_examplem} more explicitly.

\begin{lemma}\label{lem_opt_invmulti}
In the case of  a priori distribution of $\mu$  given by \eqref{eq:apriori_distributionmulti}, the optimal investment fractions  $ \kappa(t,T,Y(t)):=\frac{\pi^*(t)}{X^*(t)}$  of the Bayesian problem \eqref{eq:Bayesproblem_examplem} can be written as follows.

    \begin{align*}
\kappa(t,T,Y(t))
&=   \gamma (\sigma^\top)^{-1} \bar\vartheta\; \alpha(t,T,Y(t)) +   \gamma (\sigma^\top)^{-1} \underline\vartheta\; (1- \alpha(t,T,Y(t)))\\
&= \kappa^{\text{Mer}}(\gamma,\bar \vartheta)\; \alpha(t,T,Y(t)) +  \kappa^{\text{Mer}}(\gamma,\underline\vartheta)\; (1- \alpha(t,T,Y(t)))
\end{align*}
with $\alpha(t,T,Y(t))=1-(1-p) \hat F(t,T,Y(t))\in[0,1].$
\end{lemma}

\begin{proof}
From Theorem \ref{theo:Bayes} we know that the optimal fraction is given by:
\begin{equation}
\kappa(t,T,Y(t))
=\gamma (\sigma^\top)^{-1}  \frac{ \int_{\R^d} \nabla F(T,z+Y(t))(F(T,z+Y(t))^{\gamma-1}\varphi_{T-t}(z) dz}{\int_{\R^d} (F(T,z+Y(t))^{\gamma}\varphi_{T-t}(z) dz}
\end{equation}
We obtain that
\begin{align*}
  \nabla F(t,z) & =  p \nabla L_t(\bar\vartheta,z)+(1-p) \nabla L_t(\underline\vartheta,z)\\
  & = p L_t(\bar\vartheta,z) \bar\vartheta+(1-p) L_t(\underline\vartheta,z) \underline\vartheta\\
    & = \bar\vartheta F(t,z) -(1-p) L_t(\underline\vartheta,z)(\bar\vartheta- \underline\vartheta)  
\end{align*}
which implies
    \begin{align*}
\kappa(t,T,Y(t))
& =  \gamma (\sigma^\top)^{-1} \bar\vartheta-
(1-p)\gamma (\sigma^\top)^{-1}  (\bar\vartheta-\underline\vartheta)\; \hat F(t,T,Y(t)).
\end{align*}
The stated representation is obtained from
    \begin{align*}
\kappa(t,T,Y(t))
& =  \gamma (\sigma^\top)^{-1} \bar\vartheta-
(1-p)\gamma (\sigma^\top)^{-1}  (\bar\vartheta-\underline\vartheta)\; \hat F(t,T,Y(t))\\
&=   \gamma (\sigma^\top)^{-1} \bar\vartheta \Big(1-(1-p) \hat F(t,T,Y(t))  \Big)+
\gamma (\sigma^\top)^{-1} \underline\vartheta (1-p) \hat F(t,T,Y(t)).
\end{align*}
What is left to prove is that $1-\alpha(t,T,Y(t)) = (1-p) \hat F(t,T,Y(t))\in[0,1].$ Non-negativity is clear. For the upper bound note that
$(1-p) L_T(\underline\vartheta,z+Y(t))\le F(T,z+Y(t)).$ 
\end{proof}

Thus, we can see  that the optimal fraction is always a convex combination between the two possible Merton fractions $ \gamma (\sigma^\top)^{-1} \bar\vartheta$ and $\gamma (\sigma^\top)^{-1} \underline\vartheta.$

\subsection{Pre-commitment}
We will also compare the optimal investment fraction $\kappa(t,T,Y(t))$ under learning with the optimal pre-commitment strategy without learning.
The optimal pre-commitment strategy is a constant, $\mathcal{F}_0^Y-$ measurable investment fraction $\kappa^{\text{pre}}_T$ which is defined  by the optimal constant investment fraction $\kappa$ which solves the optimization problem (\cite{bmbo2022})
\begin{equation}\label{eq:pre_example}
V^{\text{pre}}(x_0)= \sup_\pi \; p\; \EE_{\bar\mu} [u(X_T^\pi)] + (1-p) \EE_{\underline\mu} [u(X_T^\pi)] \text{ s.t. } \pi(t)=\kappa X(t).
\end{equation}
The expectation appearing in \eqref{eq:pre_example} can be computed explicitly and the first-order condition for the optimal $\kappa^{\text{pre}}_T$ can implicitly be written as
\begin{align*}
    \kappa^{\text{pre}}_T &= \gamma (\sigma^\top)^{-1} \bar\vartheta \Big( 
\frac{p e^{ \alpha \kappa^{\text{pre}}_T \cdot \bar\mu T} }{p e^{ \alpha \kappa^{\text{pre}}_T \cdot  \bar\mu T}+ (1-p)  e^{ \alpha \kappa^{\text{pre}}_T\cdot  \underline\mu T}} \Big) + \gamma (\sigma^\top)^{-1} \underline\vartheta \Big( 
\frac{(1-p) e^{ \alpha \kappa^{\text{pre}}_T \cdot \underline\mu T} }{p e^{ \alpha \kappa^{\text{pre}}_T \cdot  \bar\mu T}+ (1-p)  e^{ \alpha \kappa^{\text{pre}}_T\cdot  \underline\mu T}} \Big)\\
&=\kappa^{\text{Mer}}(\gamma,\bar \vartheta)\;  \Big( 
\frac{p e^{ \alpha \kappa^{\text{pre}}_T \cdot \bar\mu T} }{p e^{ \alpha \kappa^{\text{pre}}_T \cdot  \bar\mu T}+ (1-p)  e^{ \alpha \kappa^{\text{pre}}_T\cdot  \underline\mu T}} \Big) + \kappa^{\text{Mer}}(\gamma,\underline \vartheta)\;  \Big( 
\frac{(1-p) e^{ \alpha \kappa^{\text{pre}}_T \cdot \underline\mu T} }{p e^{ \alpha \kappa^{\text{pre}}_T \cdot  \bar\mu T}+ (1-p)  e^{ \alpha \kappa^{\text{pre}}_T\cdot  \underline\mu T}} \Big).
\end{align*} 
Thus, we can again write the optimal pre-commitment fraction as a convex combination of the two Merton fractions $ \kappa^{\text{Mer}}(\gamma,\bar \vartheta)$ and $\kappa^{\text{Mer}}(\gamma,\underline \vartheta).$
For the special case $d=1$ see  \cite{bmbo2022}, Prop. 2 and 3.
Obviously it holds $$\lim_{T\rightarrow 0}\kappa^{\text{pre}}_T= p \kappa^{\text{Mer}}(\gamma,\bar \vartheta) +(1-p) \kappa^{\text{Mer}}(\gamma,\underline \vartheta). $$
In case $d=1$ when $\bar\mu> \underline\mu$ and $\alpha >0$ we further obtain 
$$\lim_{T\rightarrow \infty}\kappa^{\text{pre}}_T=  \kappa^{\text{Mer}}(\gamma,\bar \vartheta).$$ 
These limiting results coincide with the limits in Theorem \ref{theo:sensitivity} where we allow for learning. 

\subsection{Impact of model ambiguity preferences}
Now we consider an investor who is concerned about model ambiguity, i.e. instead of problem \eqref{eq:Bayesproblem} we consider for a second utility function $v(x)=1/\lambda x^\lambda, 0<\lambda <1$ problem \eqref{eq:Aproblem}.
We treat the multi-asset case but with only two possible values $\bar\mu=(\bar\mu_1,\ldots ,\bar\mu_d)$ and $\underline{\mu}=(\underline{\mu}_1,\ldots,\underline{\mu}_d)$ for $\mu=(\mu_1,\ldots,\mu_d)$. The optimization problem is then
\begin{align}\label{prob:441}
&  \sup_{\pi\in\mathcal{A}} \left(\; p
\left(\EE_{\bar\mu} [(X_T^\pi)^{\alpha}]\right)^{\frac{\lambda}{\alpha}} + (1-p)
\left(\EE_{\underline\mu} [(X_T^\pi)^{\alpha}]\right)^{\frac{\lambda}{\alpha}}
\right)^{\frac{1}{\lambda}}.
\end{align}

\subsubsection{Probability adjustment  in the case of less ambiguity concerns}
Let us first consider the case of $0<\alpha<\lambda$ where $\mathbf{p} := \lambda/\alpha>1$. The optimization problem \eqref{prob:441} is then according to \eqref{eq:Aproblem4} and Theorem \ref{theo:Bayes} equivalent to

\begin{eqnarray}\nonumber
&&\sup_{\Q\in\mathfrak{Q}}\left\{ (q_1+q_2) x_0^\alpha \left( \int_{\R} \big(\frac{q_1}{q_1+q_2} L_T(\bar\vartheta,z)+\frac{q_2}{q_1+q_2} L_T(\underline{\vartheta},z)  \big) ^\gamma \varphi_T(z) dz\right)^{1/\gamma}\right\}\\
&=&  x_0^\alpha \sup_{\Q\in\mathfrak{Q}}\left\{ \left( \int_{\R^d} \big(q_1 L_T(\bar\vartheta,z)+q_2 L_T(\underline{\vartheta},z)  \big) ^\gamma \varphi_T(z) dz\right)^{1/\gamma}\right\}.
\end{eqnarray}
Since $\gamma >0$ it is enough to solve
\begin{equation}\label{prob:Qonly1}
  \sup_{\Q\in\mathfrak{Q}}\int_{\R^d} \big(q_1 L_T(\bar\vartheta,z)+q_2 L_T(\underline{\vartheta},z)  \big) ^\gamma \varphi_T(z) dz.  
\end{equation}
The solution of this problem can be summarized  as follows (the same solution is obtained for $\lambda<\alpha<0$):

\begin{lemma}\label{lem:optimaldistorsion1}
Let  $0<\alpha<\lambda$ or $\lambda<\alpha<0$,  thus $\mathbf{p} := \lambda/\alpha>1$. An optimal solution $(q_1^*,h(q_1^*))$ of \eqref{prob:Qonly1} is obtained by 
\begin{eqnarray*}
   &&  \sup_{0\leq q_1\leq q_1^{\text{b}}}
     \left\{\int_{\R^d} \big(q_1 L_T(\bar\vartheta,z)+h(q_1) L_T(\underline{\vartheta},z)  \big) ^\gamma \varphi_T(z) dz  \right\}\\
     && h(x):=  \left(
\frac{1- x^\mathbf{q} \; p^{1-\mathbf{q}}}{(1-p)^{1-\mathbf{q}}}
\right)^{\frac{1}{\mathbf{q}}}
\end{eqnarray*}
with $q_1^{\text{b}}:= p^{1/\mathbf{p}}.$
\end{lemma}

\begin{proof}
Note first that the set $\mathfrak{Q}$ is bounded. Indeed, when we set $q_2=0$, the maximal value of $q_1$ is given by
$p^{1/\mathbf{q}}.$ Moreover, since both $L_T(\bar\vartheta,z)$ and $L_T(\underline\vartheta,z)$ are positive, the optimal pair $(q_1,q_2)$ will satisfy the constraint with equality, i.e.
\begin{align*}
  \left(\frac{q_1}{p}\right)^\mathbf{q} p+ \left(\frac{q_2}{1-p}\right)^\mathbf{q} (1-p) & = 1.
\end{align*}
Solving this equation for $q_2$ yields the function $h$.
\end{proof}

The optimal solution $(q_1^*,q_2^*)$ in Lemma \ref{lem:optimaldistorsion1} then determines the prior distribution $(q_1^*/(q_1^*+q_2^*),q_2^*/(q_1^*+q_2^*))$ which has to be used for the Bayesian problem in order to solve problem \eqref{prob:441}.

\begin{example}
To make things more explicit consider the case that the initial prior distribution is uniform on the two outcomes, i.e.\ $p=\frac12$ and that $0<\alpha<\lambda$ with $\alpha=0.5.$ This implies that $\gamma =1/(1-\alpha)=2.$
In this special case we obtain:
\begin{eqnarray*}
  &&\sup_{\Q\in\mathfrak{Q}}\int_{\R^d} \big(q_1 L_T(\bar\vartheta,z)+q_2 L_T(\underline{\vartheta},z)  \big) ^\gamma \varphi_T(z) dz\\
  &=& 2^{1/\mathbf{q}-1} \sup_{0\le \tilde y\le 1}\int_{\R^d} \big(\tilde{y}^{1/\mathbf{q}} L_T(\bar\vartheta,z)+ (1-\tilde{y})^{1/\mathbf{q}} L_T(\underline{\vartheta},z)  \big)^2 \varphi_T(z) dz\\
  &=& 2^{1/\mathbf{q}-1} \sup_{0\le \tilde y\le 1}   \big(\tilde{y}^{2/\mathbf{q}} \exp{(T\|\bar\vartheta\|^2)}+ (1-\tilde{y})^{2/\mathbf{q}} \exp{(T\|\underline{\vartheta}\|^2)} + 2\tilde{y}^{1/\mathbf{q}}(1-\tilde{y})^{1/\mathbf{q}} \exp{(T\underline{\vartheta}^\top\bar\vartheta)}    \big).
\end{eqnarray*}
The corresponding optimal prior distribution will depend on $T,\bar\vartheta, \underline{\vartheta}. $ It is easy to see that for small time horizon $T$, the optimal prior distribution is again close to the uniform distribution.  For very large $T$, the optimal prior distribution will have large mass on the larger of the two outcomes $\|\bar\vartheta\|^2, \|\underline\vartheta\|^2.$ However, the optimal prior distribution will always be in $(0,1).$ This can easily be seen by inspecting the derivatives of the function to maximize at $\tilde y=0$ and $\tilde y=1.$ 
\end{example}

\subsubsection{Probability adjustment in the case of more ambiguity concerns}
Let us next consider the case of $0<\lambda<\alpha$ where $\mathbf{p} := \lambda/\alpha<1$. Recall that in this case $\mathbf{q}<0.$ The optimization problem \eqref{prob:441} is then according to \eqref{eq:Aproblem4} and Theorem \ref{theo:Bayes} equivalent to

\begin{eqnarray}\nonumber
&&  x_0^\alpha \inf_{\Q\in\mathfrak{Q}}\left\{ \left( \int_{\R^d} \big(q_1 L_T(\bar\vartheta,z)+q_2 L_T(\underline{\vartheta},z)  \big) ^\gamma \varphi_T(z) dz\right)^{1/\gamma}\right\}.
\end{eqnarray}
Since $\gamma >0$ it is enough to solve
\begin{equation}\label{prob:Qonly2}
  \inf_{\Q\in\mathfrak{Q}}\int_{\R^d} \big(q_1 L_T(\bar\vartheta,z)+q_2 L_T(\underline{\vartheta},z)  \big) ^\gamma \varphi_T(z) dz.  
\end{equation}
The solution of this problem can be summarized  as follows. The proof is similar to the proof of Lemma \ref{lem:optimaldistorsion1} and we skip it here (the same solution is obtained for $\alpha<\lambda<0$).

\begin{lemma}\label{lem:optimaldistorsion2}
Let  $0<\lambda<\alpha$ or $\alpha<\lambda<0$, thus $\mathbf{p} := \lambda/\alpha<1$. An optimal solution $(q_1^*,h(q_1^*))$ of \eqref{prob:Qonly2} is obtained by 
\begin{eqnarray*}
   &&  \inf_{ q_1\geq q_1^{\text{b}}}
     \left\{\int_{\R^d} \big(q_1 L_T(\bar\vartheta,z)+h(q_1) L_T(\underline{\vartheta},z)  \big) ^\gamma \varphi_T(z) dz  \right\}
\end{eqnarray*}
with $q_1^{\text{b}}$ and $h$ as in Lemma \ref{lem:optimaldistorsion1}.
\end{lemma}

\begin{example}
We can again consider the case that the initial prior distribution is uniform on the two outcomes, i.e.\ $p=\frac12$ and that $0<\lambda<\alpha$ with $\alpha=0.5.$ This implies that $\gamma =1/(1-\alpha)=2.$
In this special case we obtain:
\begin{eqnarray*}
  && \inf_{ q_1\geq q_1^{\text{b}}} \int_{\R^d} \big(q_1 L_T(\bar\vartheta,z)+q_2 L_T(\underline{\vartheta},z)  \big) ^\gamma \varphi_T(z) dz\\
  &=& 2^{1/\mathbf{q}-1} \inf_{0\le \tilde y\le 1}   \big(\tilde{y}^{2/\mathbf{q}} \exp{(T\|\bar\vartheta\|^2)}+ (1-\tilde{y})^{2/\mathbf{q}} \exp{(T\|\underline{\vartheta}\|^2)} + 2\tilde{y}^{1/\mathbf{q}}(1-\tilde{y})^{1/\mathbf{q}} \exp{(T\underline{\vartheta}^\top\bar\vartheta)}    \big).
\end{eqnarray*}
Again the optimal prior distribution will always be in $(0,1)$ since the function tends to $+\infty$ at the boundary. According to Theorem \ref{theo:sensitivity} b) this implies that for a very large time horizon the optimal investment fraction is hardly influenced by model ambiguity.
\end{example}

\subsubsection{Probability adjustment in the case of much less ambiguity concerns}
Let us finally consider the case of $0<\lambda$ and $\alpha<0$ where $\mathbf{p} := \lambda/\alpha<0$. Recall that in this case $0<\mathbf{q}<1.$ The optimization problem \eqref{prob:441} is then according to \eqref{eq:Aproblem4} and Theorem \ref{theo:Bayes} equivalent to

\begin{eqnarray}\nonumber
&&  x_0^\alpha \inf_{\Q\in\mathfrak{Q}'}\left\{ \left( \int_{\R^d} \big(q_1 L_T(\bar\vartheta,z)+q_2 L_T(\underline{\vartheta},z)  \big) ^\gamma \varphi_T(z) dz\right)^{1/\gamma}\right\}.
\end{eqnarray}
Since $\gamma >0$ it is enough to solve
\begin{equation}\label{prob:Qonly3}
  \inf_{\Q\in\mathfrak{Q}'}\int_{\R^d} \big(q_1 L_T(\bar\vartheta,z)+q_2 L_T(\underline{\vartheta},z)  \big) ^\gamma \varphi_T(z) dz.  
\end{equation}
The solution of this problem can be summarized  as follows. The proof is similar to the proof of Lemma \ref{lem:optimaldistorsion1} and we skip it here (the same solution is obtained for $\alpha>0,\lambda<0$).

\begin{lemma}\label{lem:optimaldistorsion3}
Let  $\alpha<0<\lambda$ or $\lambda<0<\alpha$, thus $\mathbf{p} := \lambda/\alpha<0$. An optimal solution $(q_1^*,h(q_1^*))$ of \eqref{prob:Qonly3} is obtained by 
\begin{eqnarray*}
   &&  \inf_{ 0\le q_1\leq q_1^{\text{b}}}
     \left\{\int_{\R^d} \big(q_1 L_T(\bar\vartheta,z)+h(q_1) L_T(\underline{\vartheta},z)  \big) ^\gamma \varphi_T(z) dz  \right\}
\end{eqnarray*}
with $q_1^{\text{b}}$ and $h$ as in Lemma \ref{lem:optimaldistorsion1}.
\end{lemma}

\section{Numerical illustration and economic discussion}\label{sec_numerics}
The following discussions and numerical illustrations refer to the case of one risky asset with a   two point prior distribution where only two drift scenarios $\underline\mu$ and $\overline\mu$ ($0<\underline\mu\leq\overline\mu$) are possible. The market price of risk is thus positive in both scenarios ($0<\underline\vartheta<\overline\vartheta$).\footnote{From a technical point of view, 
 the results are reversed in the case of a negative market
price of risk (cf. \cite{longo2016learning}).}
Along the lines of Eqn. (\ref{eq:apriori_distributionmulti}), $p$ denotes the prior probability for the higher drift $\overline\mu$ (good or upper drift scenario) such that the prior probability for the lower drift scenario $\underline\mu$ is given by $1-p$. In addition, we refer to the initial date  $t=0$ where $Y(t)=0$.\footnote{The discussion can also be extended to future dates $t>0$ where $Y(t)\neq 0$. Using the log-investor as a benchmark then implies   that the  prior distribution is replaced by the 'updated' distribution.}

We consider intuitive explanations of the previous results and shed further light on    the impact of the two sources of risk on the optimal investment strategy. Along the lines of Eqn. (\ref{eq:Aproblem}), the strategy is obtained by maximizing the expected utility over a double risk situation
 where the two sources of risk are evaluated with different utility functions $u$ and $v$.
Our economic  reasoning is first based on the classical Bayesian problem which coincides with the special  case that the two utility functions  $v$ and $u$, are identical ($\lambda=\alpha$, respectively). 
Since the general case can be captured by modifying the probability $p$ towards  $q_1^*/(q_1^*+q_2^*)$ (cf. Lemma \ref{lem:optimaldistorsion1}, \ref{lem:optimaldistorsion2} and \ref{lem:optimaldistorsion3}), 
all sensitivities can be explained by the probability adjustment and the sensitivities of the classical setup in $p$.
Thus, 
we first discuss and illustrate the Bayesian case ($\lambda=\alpha$) and consider the impact of $\lambda\neq\alpha$ on the prior probability adjustment subsequently.  
In addition, we comment on the implications of pre-commitment instead of learning.
 
In what follows the classical Merton problem with a constant  market price of risk $\vartheta$ (known drift, respectively) is referred to as the {\em inner risk situation}. Here, the optimal investment fraction $\kappa^{\text{Mer}}(\gamma,\vartheta)= \gamma \frac{\vartheta}{\sigma}$ is constant (it does not depend on the investment horizon $T$). It is increasing in the market price of risk $\vartheta$ and decreasing in the level of relative risk aversion $\frac{1}{\gamma} =1-\alpha$. 
In contrast to the inner risk, we refer to the  probability distribution over the market price of risk as the {\em outer risk situation}. The impact of this outer risk on the optimal investment strategy is more demanding since, in general, it depends on the investment horizon $T$. An exception is the log-investor (an investor with a log utility function) who can be used as an intuitive benchmark for other investors (more or less risk averse investors).  

If not mentioned otherwise, we refer to the three benchmark parameter constellation summarized in Table \ref{tab_bench_NA}.

\begin{table}[]
    \centering
    \begin{tabular}{|c|c|c|c|}\hline
    \multicolumn{4}{|c|}{model parameter}\\\hline
 $\underline\mu$     & $\overline\mu$ & $\sigma$ & p \\
0.03 & 0.09 & 0.15 & 0.5\\ \hline
\end{tabular}
 \begin{tabular}{|c|c|c|}\hline
    \multicolumn{3}{|c|}{level of relative risk aversion $\frac{1}{\gamma}=1-\alpha$}\\\hline
 less than log-investor    & log-investor & more than log-investor \\
$1/\gamma=0.5$ & $1/\gamma=1$ & $1/\gamma=2$\\ \hline
\end{tabular}
    \caption{Benchmark parameter setup (classical Bayesian case). If not otherwise mentioned, the investment horizon is $T=10$ years.}
    \label{tab_bench_NA}
\end{table}

\subsection{Sensitivities in the classical Bayesian case}\label{subsec_disc_a}
\begin{figure}[tb]
	\begin{center}
			{\bf{Weight $g(\alpha,p, T,\underline\vartheta,\overline\vartheta)$ on lower Merton investment fraction}}
		\end{center}
	\begin{center}
			\includegraphics[width=0.45\textwidth]{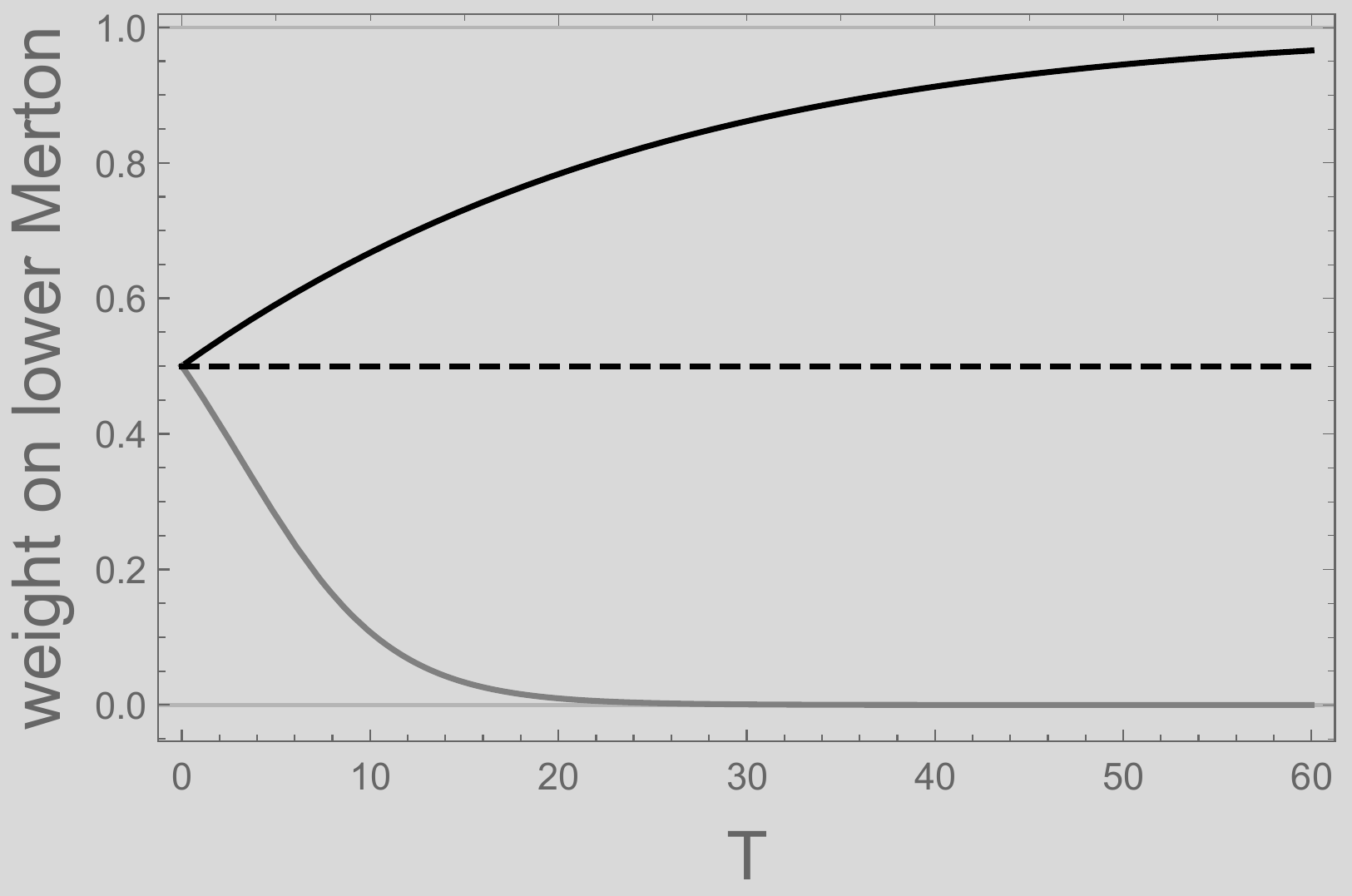}	\includegraphics[width=0.45\textwidth]{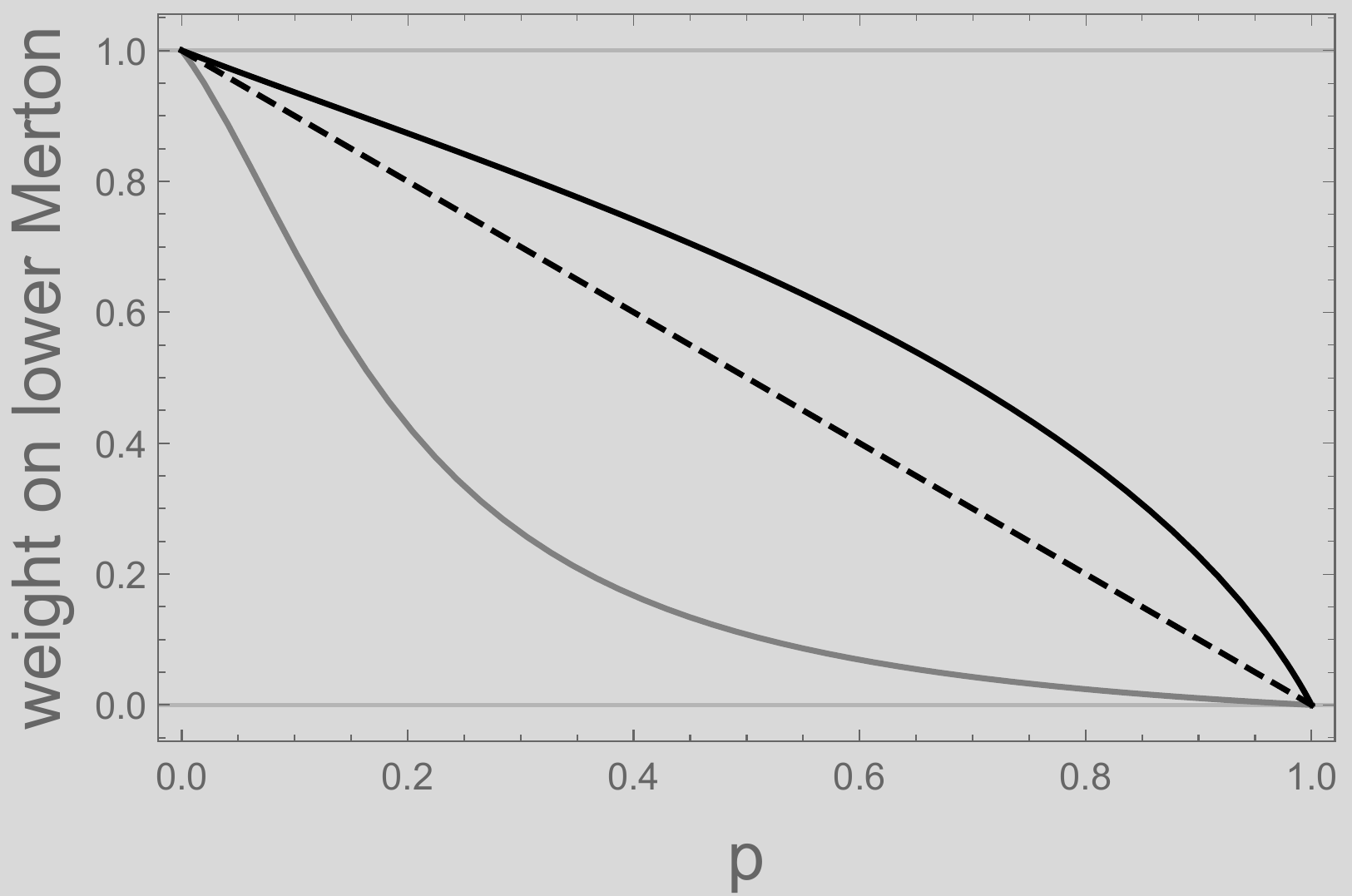}
			\includegraphics[width=0.45\textwidth]{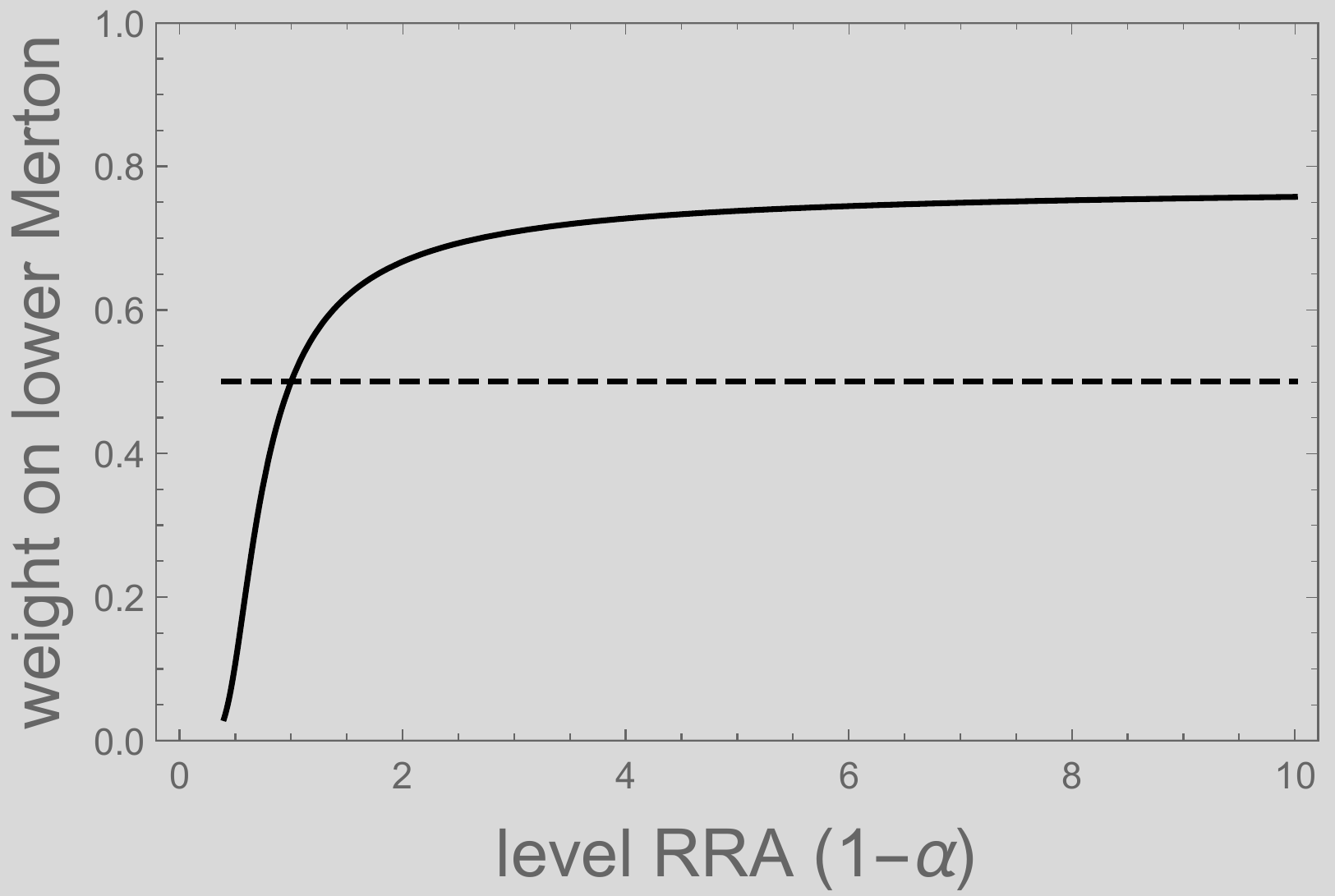}
			\includegraphics[width=0.45\textwidth]{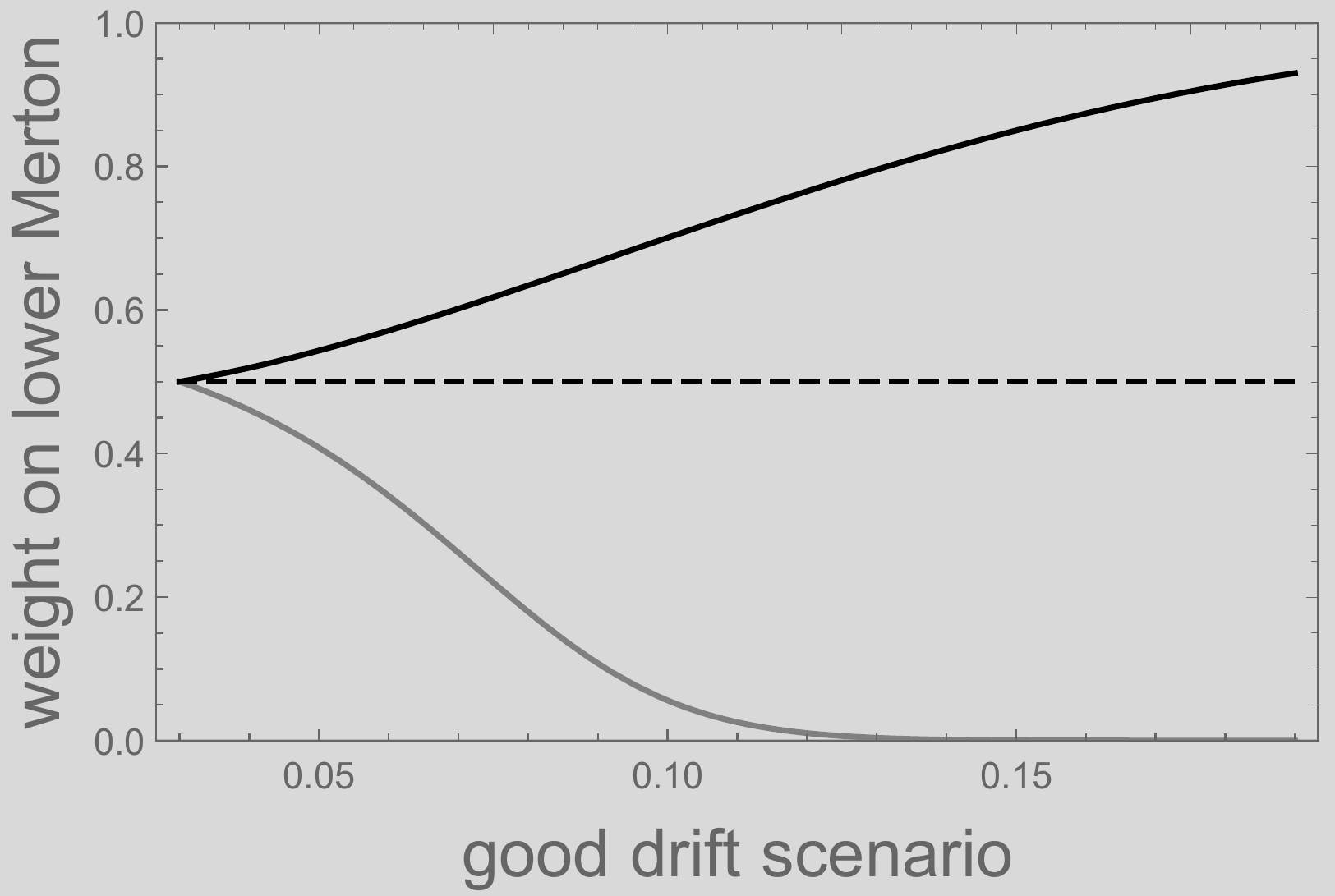}
\end{center}		
	\caption{The figure is based on the benchmark parameter setup of Table \ref{tab_bench_NA} but considers the dependence on the investment horizon (upper left plot), the probability for the good drift scenario $p$ (upper right plot), the level of relative risk aversion $1-\alpha$ (lower left plot), and the parameter for the good drift scenario $\overline\mu$ (lower right plot). 
	The plots depict the weight on the lower Merton fraction in the optimal investment strategy. In addition, each figure contains the  log-investor (dashed line) and the more (black line) and less (gray line) risk averse investor than the log-investor.}\vspace*{0.15cm}
	\label{Fig_new_a}
\end{figure}

If $u$ and $v$ coincide, Lemma \ref{lem_opt_invmulti} immediately separates the impact of the inner and outer risk on the optimal investment fraction.
The optimal myopic investment fraction  of all investors (all levels of relative risk aversion $1-\alpha$), is given by the convex combination of the Merton fractions of the inner risk situations (good and bad scenario). Thus, the impact of the inner risk is simply as in the classical Merton problem while the impact of the outer risk can be stated by the weight on the good (or bad scenario).
Along the lines of Lemma \ref{lem_opt_invmulti}, the weight on the lower Merton solution (the bad scenario)
is given by 
\begin{align*}
1-\alpha(t,T,Y(t))=(1-p) \hat F(t,T,Y(t))
\end{align*}
 where 
\begin{align*}
    \hat F(t,T,Y(t)) & :=\frac{ \int_{\R}
L_T(\underline\vartheta,z+Y(t))
 (F(T,z+Y(t)))^{\gamma-1}\varphi_{T-t}(z) dz}
 {\int_{\R} (F(T,z+Y(t))^{\gamma}\varphi_{T-t}(z) dz}.
\end{align*}
Notice that 
\begin{align*}
L_T\left(\vartheta,z+Y(t)\right)
 & = L_t\left(\vartheta,Y(t)\right)L_{T-t}\left(\vartheta,z\right).
\end{align*}
This hints at two (competing) effects: one is due to {\it{learning}} (i.e. the observation of $Y(t)$ at time $t$), the other effect is  caused by the remaining investment horizon $T-t$. The second effect is called time-inconsistency and vanishes for the log-investor ($\alpha\rightarrow 0$, $\gamma\rightarrow 1$, respectively), i.e.
$ \hat F(t,T,Y(t))= \frac{L_t(\underline\vartheta,Y_t)}{F(t,Y_t)}$.\footnote{It is well understood that the problem of time inconsistency naturally arises if one aggregates utilities (cf. e.g. \cite{desmettre2021optimal}).} We first discuss the second effect and consider the special case $t=0$ ($Y(t)=0$) where

\begin{align*}
g(\alpha,p, T,\underline\vartheta,\overline\vartheta) :=   
1-\alpha(0,T,0)=(1-p) \hat F(0,T,0).
\end{align*}
\subsubsection{Benchmark log-investor}
Formally, the justification that the log-investor $\alpha\rightarrow 0$ (with a level of relative risk aversion $1-\alpha$ equal to 1) defines a convenient benchmark is given by 
Theorem \ref{theo:sensitivity} e), i.e.
\begin{align*}
\lim_{\alpha\rightarrow 0}  g(\alpha,p, T,\underline\vartheta,\overline\vartheta)=\lim_{T\rightarrow 0}  g(\alpha,p, T,\underline\vartheta,\overline\vartheta)
=1-p.
\end{align*}
Observe that the log-investor  behaves in a myopic way.  
In the short run ($T\rightarrow 0$), the impact of the outer risk (weight on the bad scenario/lower Merton solution)   
is given by the prior probability $1-p$.
While the within regime Merton fractions depend on the individual levels of risk aversion, the myopic solution of the outer risk is {\it{risk neutral}} in the sense that it is given by the expected value (under the prior distribution) of the within regime Merton fractions. 
Obviously, for $T\rightarrow 0$ there is no hedging demand.
For the log-investor, this is also true for investment horizons $T>0$ (cf. Figure \ref{Fig_new_a} upper left plot).
In addition, the log-investor gives an important distinction: 
an investor who is less (more) risk averse uses a lower (higher) weight on the lower Merton fraction (bad regime), cf. Figure \ref{Fig_new_a} (lower left plot).\footnote{  \cite{rieder2005portfolio} Theorem 9 d).}
Economically, this is explained by the trade-off between speculating on the better regime (and following the optimal strategy for the good regime) and hedging against the worse regime (and following the optimal strategy for the bad regime). While the log-investor acts neutral, the hedging (speculation) motive dominates for the  more (less) risk averse investor. 

Intuitively, the larger the difference between the two regimes, the stronger is the hedging (speculating) motive, and the more the optimal strategy moves towards the worst(best)-case strategy. An illustration is given in Figure \ref{Fig_new_a} (lower right plot) which depicts the weight on the lower Merton solution for varying $\overline\mu$. The higher $\overline\mu$ (the higher the difference between the regimes), the larger is the impact of the hedging (speculation) motive for the more (less) risk averse investor. 
\subsubsection{Long run investor}
Similar reasoning applies to the investment horizon $T$. The outer risk is 'higher' for longer investment horizons.
For large $T$, the investor who is more risk averse than the log-investor only considers the Merton solution of the bad regime, i.e. for $\alpha<0$ 
\begin{align*}
 \lim_{T\rightarrow \infty}  g(\alpha,p, T,\underline\vartheta,\overline\vartheta)
=1.
\end{align*}
Thus, w.r.t.\ the outer risk, the long term investor who is more risk averse than the log-investor acts extremely risk averse. 
In contrast, a long term investor who is less risk averse than the log-investor only considers the good regime.\footnote{  \cite{bauerle2017extremal} Theorem 3.1, Theorem 3.3}
In consequence,   for long term investment horizons ($T\rightarrow\infty$), the prior distribution has no impact on the weight $\alpha(0,T,y)$, i.e. 
 $\alpha(0,T,y)=0$ ($\alpha(0,T,y)=1$)
for $\alpha<0$ (for $\alpha>0$).
The long term investor only considers the worst (best) case drift, i.e. $\underline\mu$ ($\overline\mu$) and behaves along the lines of  a maximin (maximax) decision rule. 
This is further emphasized by means of the sensitivities of the optimal weight on the lower Merton investment fraction.   An illustration is given by Figure \ref{Fig_new_a} (upper left plot) which depicts  the weight on the lower Merton  fraction for varying investment horizons $T$. See also subsections \ref{subseq:learning}, \ref{subsec:value_of_learning}.

\subsubsection{Impact of prior distribution}\label{subsubsec_prior}
\begin{figure}[tb]
	\begin{center}
			{\bf{Difference to log-investor ($g(\alpha,p, T,\underline\vartheta,\overline\vartheta)-(1-p)$)}}
		\end{center}
	\begin{center}
			\includegraphics[width=0.45\textwidth]{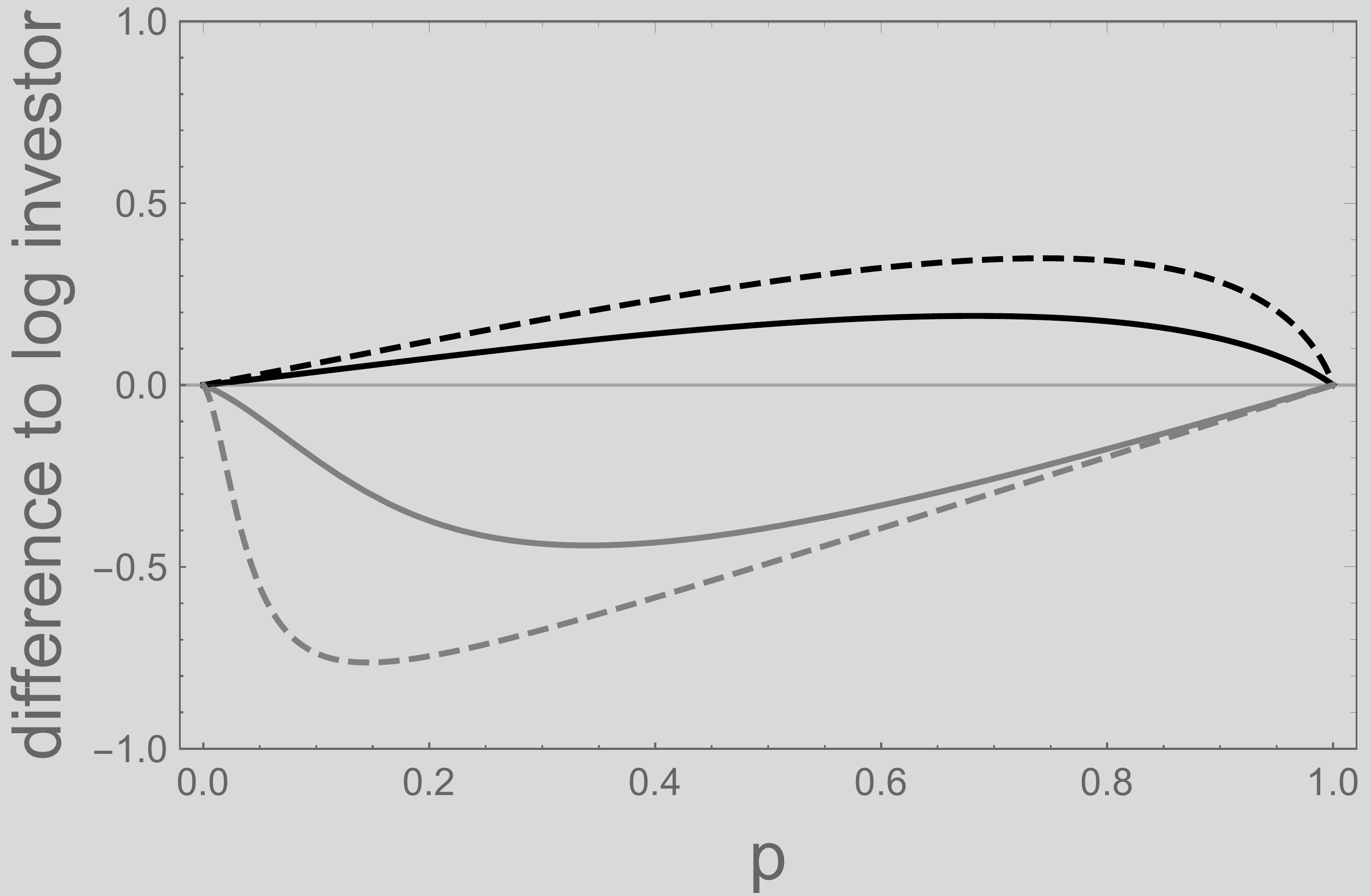}
				\includegraphics[width=0.45\textwidth]{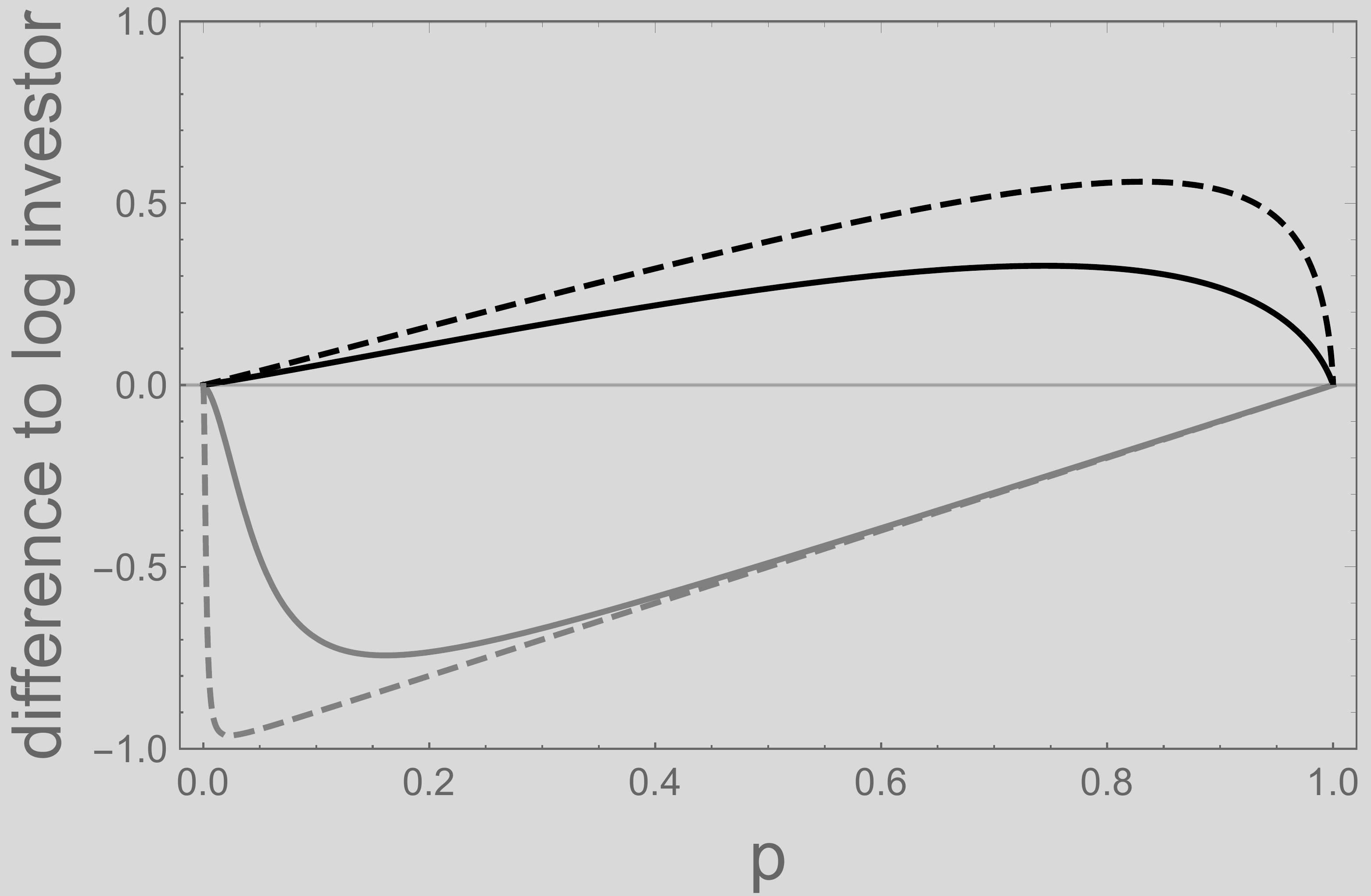}
					\includegraphics[width=0.45\textwidth]{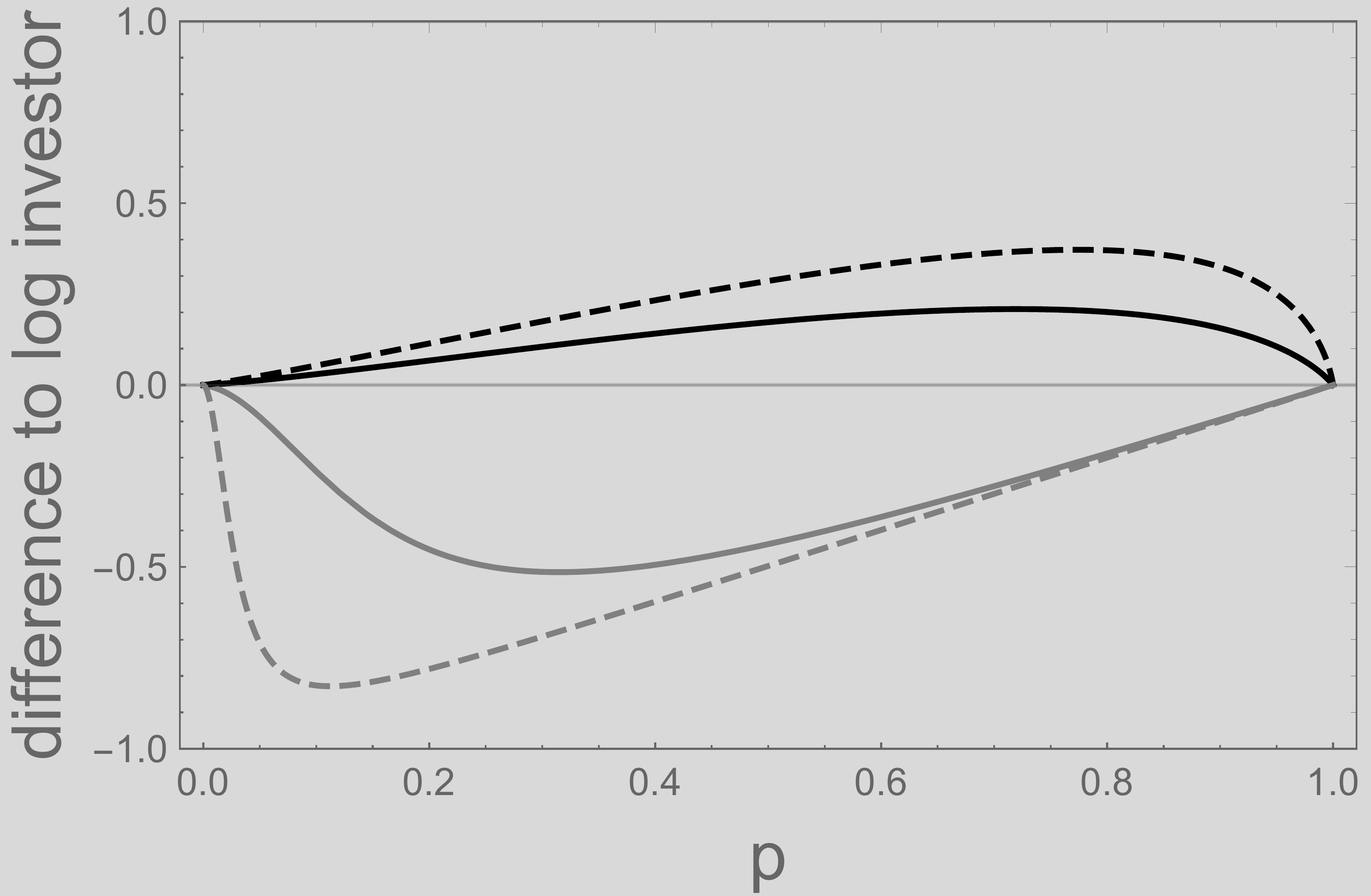}
					\includegraphics[width=0.45\textwidth]{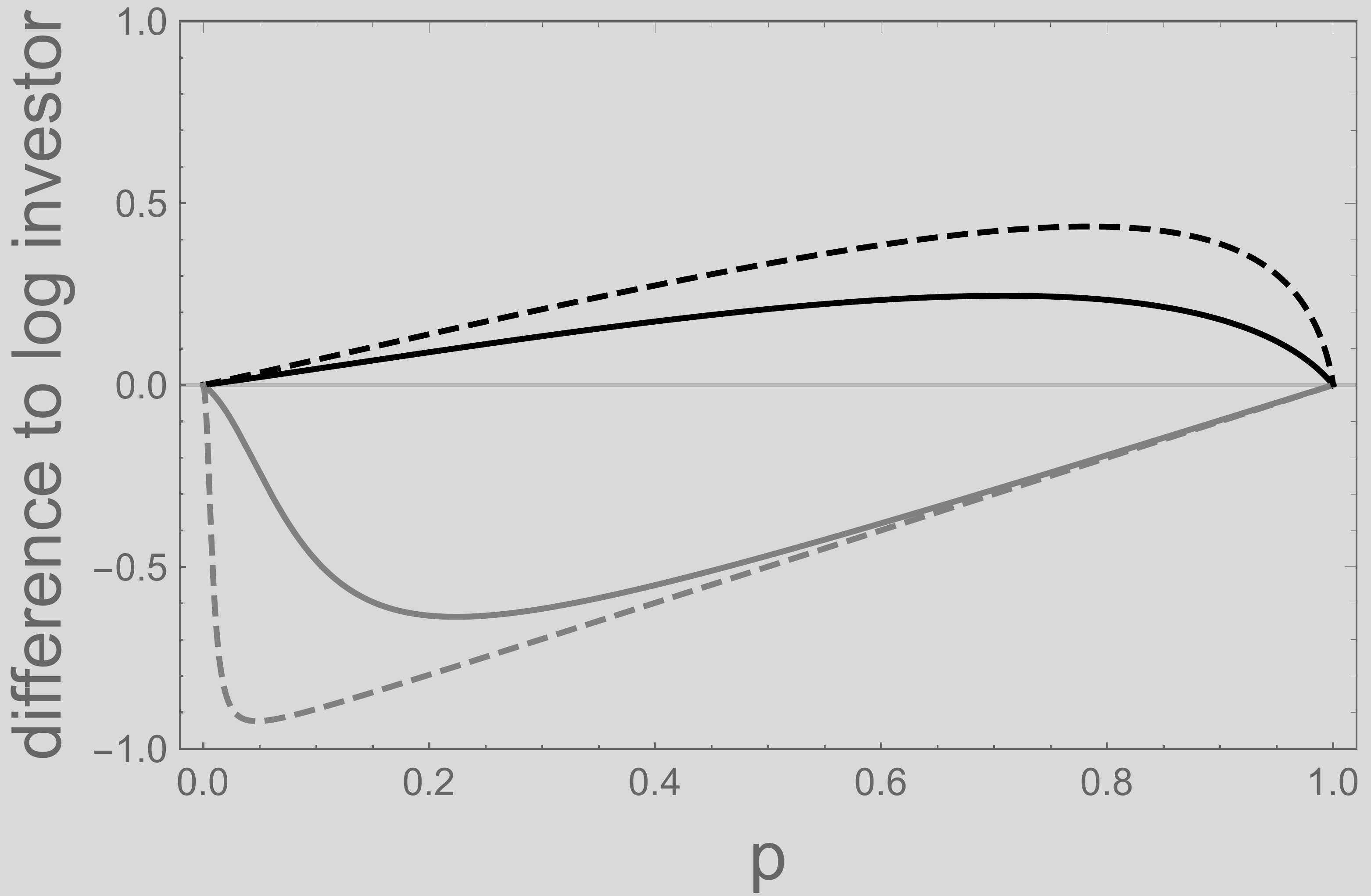}	
\end{center}		
	\caption{If not otherwise mentioned below, the figure refers to the benchmark parameter scenario of Table \ref{tab_bench_NA}.
	The figure compares, for varying priors $p$, the weight on the lower Merton fraction to the one of  the log-investor, i.e. it depicts the difference $g(\alpha,p, T,\underline\vartheta,\overline\vartheta)-(1-p)$. The black (gray) lines refer to the investor who is more (less) risk averse than the log-investor. The investment horizon is $T=10$ (thick lines) and $T=20$ (dashed lines). The upper right figure is then based on  $\overline\mu= 0.12$ (instead of the benchmark parameter $\overline\mu= 0.09$). The lower left plot then considers the benchmark scenario with the exception that  $\underline\mu= 0.01$ (instead of $\underline\mu= 0.03$). The last figure (lower right plot) considers a  higher risk aversion ($\alpha=-2$ instead of $\alpha=-1$) for the black lines and a lower risk aversion ($\alpha=0.6$ instead of $\alpha=0.5$) for the gray lines.
}
	\vspace*{0.15cm}
	\label{Fig_new_log}
\end{figure}
Intuitively, the weight on the lower Merton investment fraction is decreasing (increasing) in the prior probability $p$ ($1-p$) (upper right plot of Figure \ref{Fig_new_a}). This is obvious for the log-investor who simply relies on the prior probability. However, notice again the important distinction based on the benchmark log-investor. While the more risk averse investor (black line) uses a higher weight on the lower Merton solution, the opposite is true for the less risk averse investor (gray line). Intuitively, the discrepancy is the highest for some intermediate $p$. While there is no outer risk implied in the extreme cases that $p\rightarrow 0$ or   $p\rightarrow 1$, the outer risk is 'maximal' for some intermediate $p$ (which depends on the investment horizon $T$, the level of relative risk aversion $1-\alpha$, and   the relation between good and bad scenario).  The outer risk situation also increases in the distance between the good and bad regime ($\overline\mu$ and $\underline\mu$). Thus, the deviation of a more (less) risk averse investor from the log-investor is the higher, the higher the distance is, e.g. it is increasing in $\overline\mu$ (cf lower right figure of Figure \ref{Fig_new_a}). An additional illustration is given by Figure  \ref{Fig_new_log} which depicts the difference to the log-investor by means of 
$$g(\alpha,p, T,\underline\vartheta,\overline\vartheta)-(1-p),$$
i.e. the difference of the weight on the lower Merton solution and $1-p$ which is the weight of the log-investor. All figures include two investment horizons ($T=10$ and $T=20$).
Observe that the deviation from the log-investor is the higher, the higher the investment horizon is. For the investor who is more (less) risk averse than the log-investor, the highest  deviation is obtained for $p>0.5$ ($p<0.5$). Recall that the more (less) risk averse investor tends to hedge (speculate) the bad (good) drift scenario. In consequence, the prior probability $p$ with the highest deviation increases (decreases) in the investment horizon.   
Analogous reasoning is true w.r.t. the sensitivity to $\overline \mu$ ($\underline \mu$) and the level of risk aversion $1-\alpha$ (cf. Figure  \ref{Fig_new_log}).

\subsubsection{Learning vs.\ pre-commitment}\label{subseq:learning}
\begin{figure}[tb]
	\begin{center}
			{\bf{Comparison of optimal initial learning and pre-commitment strategy}}
		\end{center}
	\begin{center}
			\includegraphics[width=0.45\textwidth]{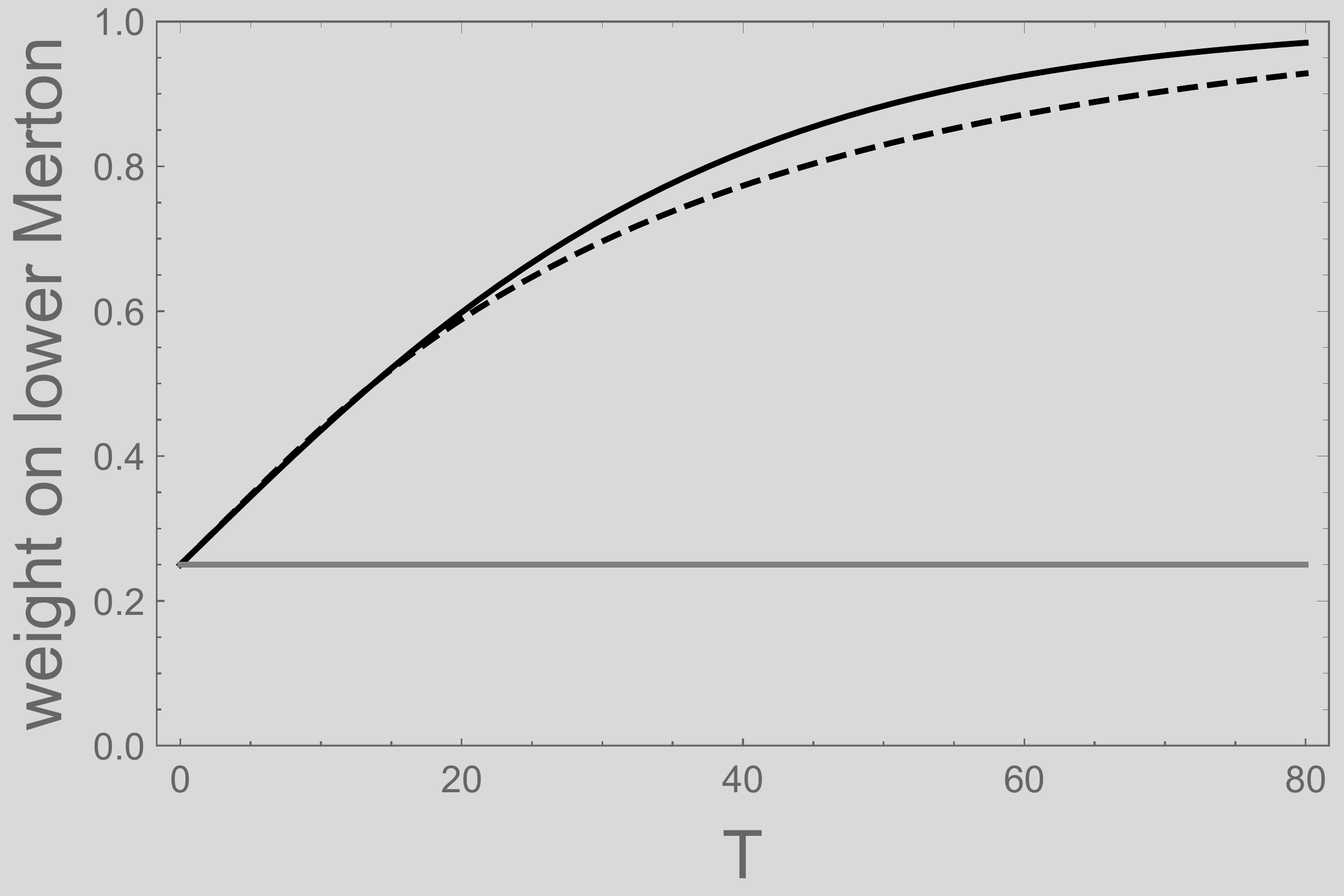}	\includegraphics[width=0.45\textwidth]{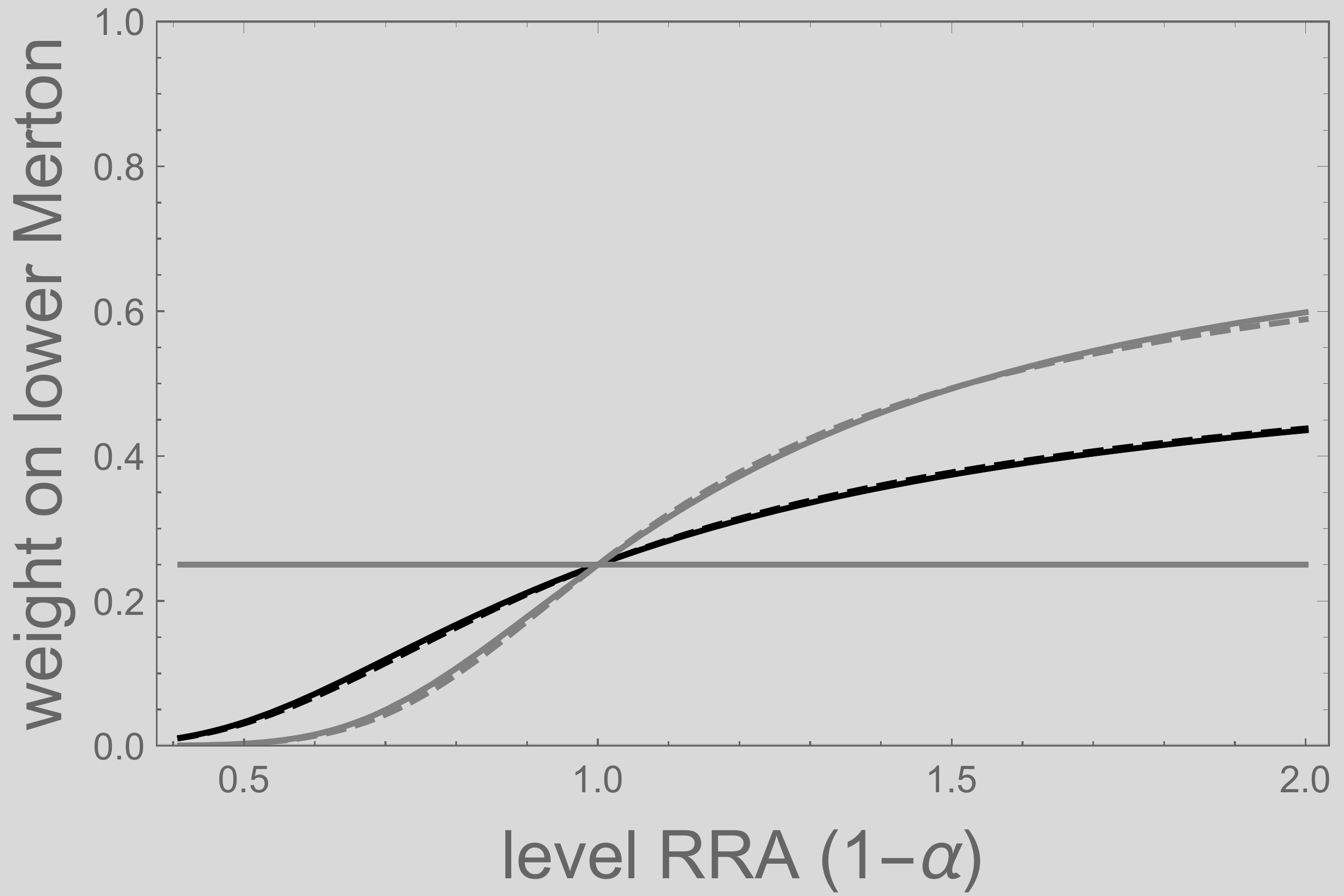}	
			\includegraphics[width=0.45\textwidth]{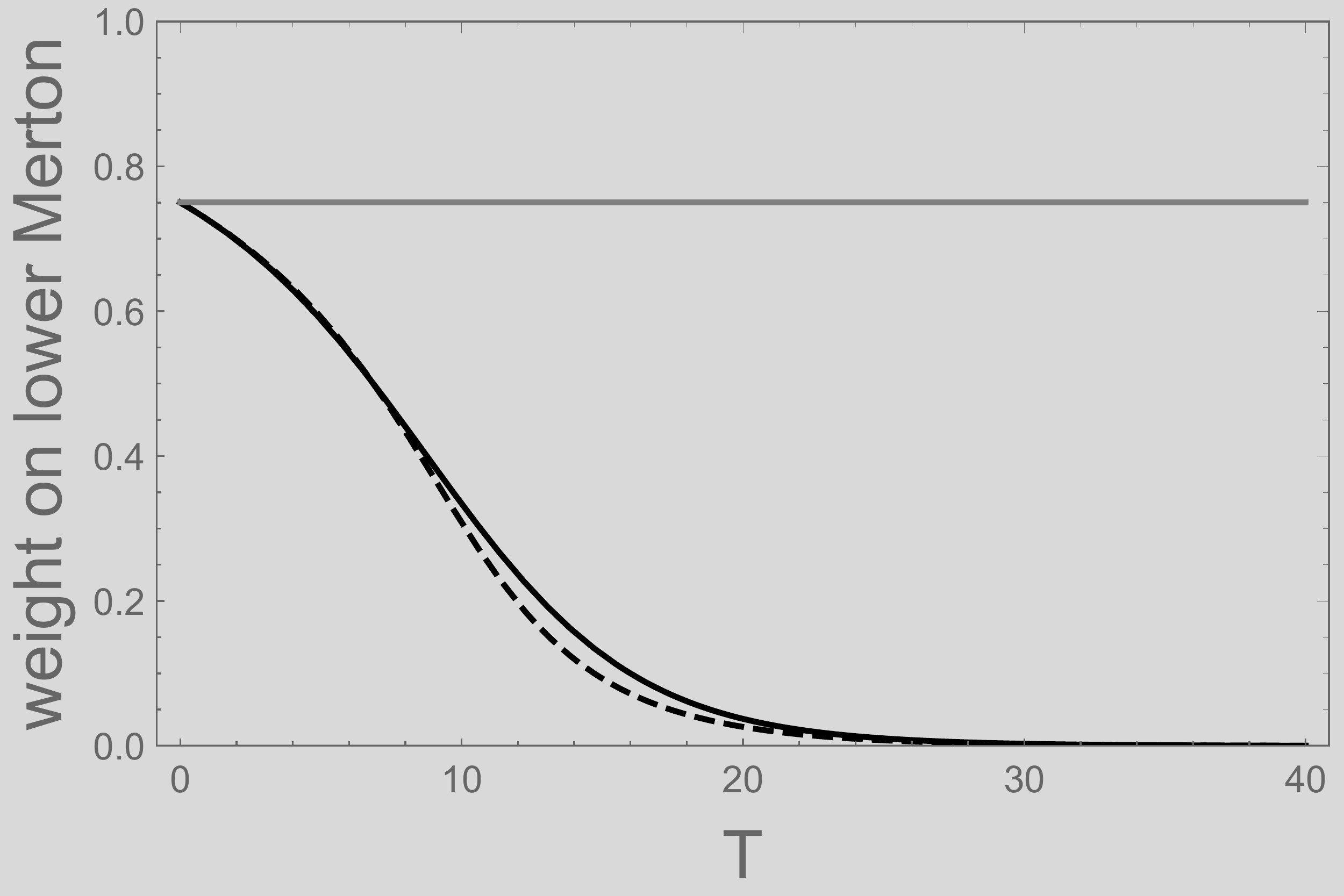}	\includegraphics[width=0.45\textwidth]{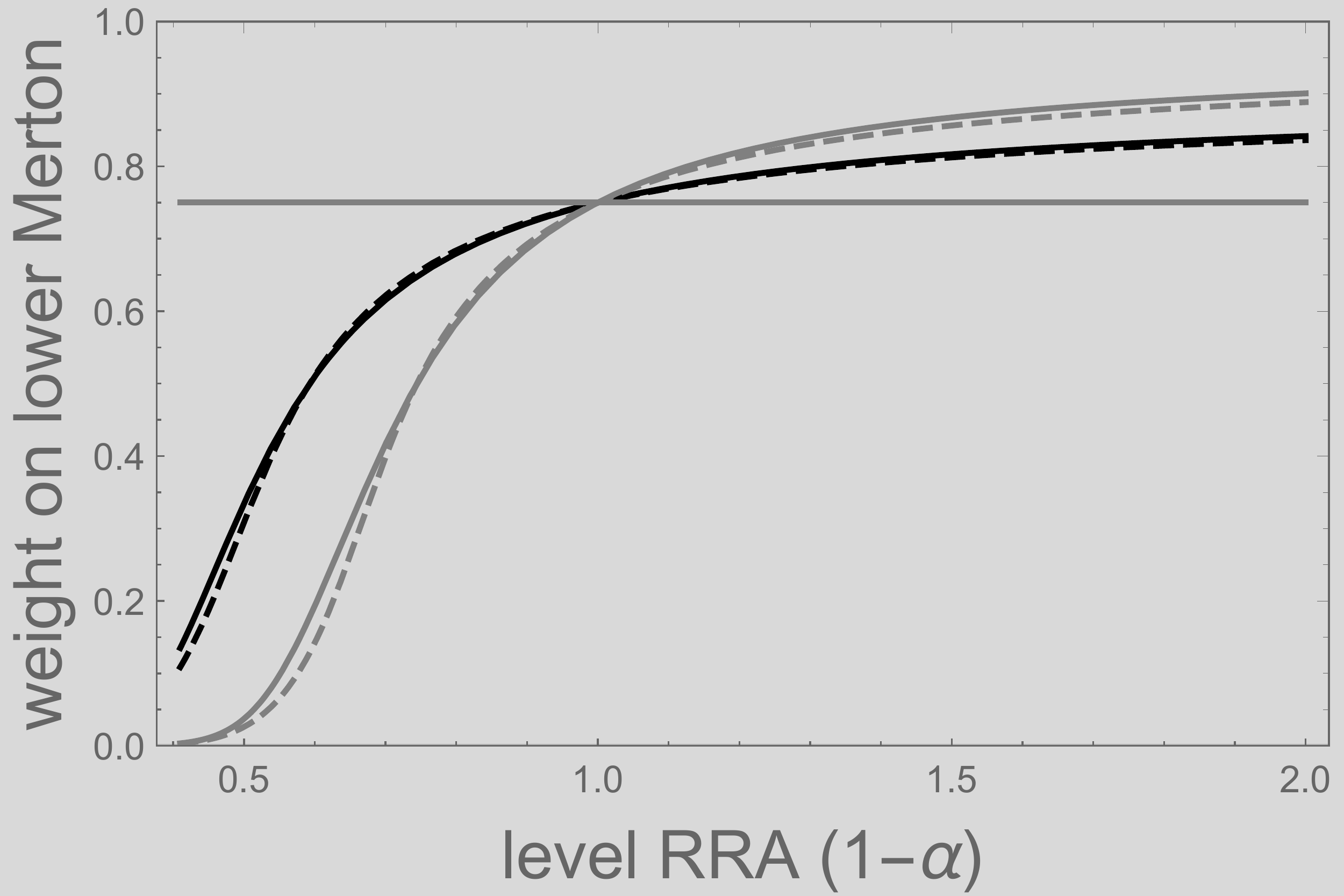}
\end{center}		
	\caption{If not otherwise mentioned below, the figure refers to the benchmark parameter scenario of Table \ref{tab_bench_NA}. The upper plots  are based on $p=0.75$ (i.e. $1-p=0.25$) while the lower plots refer to 
	 $p=0.25$ (i.e. $1-p=0.75$).
	The left hand figures compare  the optimal weight on the lower Merton solution of the learning (black) and pre-commitment setup (dashed) for varying investment horizons $T$. The upper (lower) figure concerns the investor who is more (less) risk averse than the log-investor.  The right hand figures depict, for $T=10$ (black) and  $T=20$ (gray), the weight on the lower Merton solution for varying levels of relative risk aversion. Again, the black lines refer to strategies under learning while the dashed ones are pre-commitment strategies.}\vspace*{0.15cm}
	\label{Fig_comp}
\end{figure}
It is worth to mention that similar reasoning as above is valid in a pre-commitment setup (cf. optimization problem \eqref{eq:pre_example}). 
First notice that, in the classical Bayesian setup,  the optimal investment fraction   of the log-investor does not depend on $T$.
Thus, the optimal pre-commitment and learning strategies coincide at $t=0$.
In addition, this is true for the limiting cases of a myopic investor ($T\rightarrow 0$) and the long term investor ($T\rightarrow \infty$).
For an investor who is more risk averse than the log-investor, recall now the hedging motive. 
W.r.t.\ the trade-off between speculating on the better regime (and following the optimal strategy for the good regime) and hedging against the worse regime (and following the optimal strategy for the bad regime), the hedging motive dominates. In addition, observe that the hedging motive is increasing in the investment horizon $T$. Compared to the setup under learning, pre-commitment implies that the investor must pre-commit to a constant investment fraction at $t=0$. In the optimum, she anticipates that gradually over time, the remaining investment horizon decreases.
Therefore, at $t=0$, the hedging demand is, for rather high investment horizons $T$, lower in the pre-commitment setup than under learning. 
The above reasoning is illustrated in Figure \ref{Fig_comp}. The left figure compares (for an investor who is more risk averse than the log-investor) the optimal weight on the lower Merton solution of the learning (black) and pre-commitment setup (dashed) for varying investment horizons $T$. 
Observe that the strategies coincide in the limiting cases. In addition, for intermediate investment horizons $T$, the pre-commitment strategy has a lower weight on the bad scenario (lower Merton solution). This implies that (at $t=0$) the pre-commitment strategy is more aggressive than the learning strategy. As mentioned above, this is due to the fact that the investor who must pre-commit already has to anticipate that, in the future,  the remaining investment horizon is lower which is associated with a lower hedging demand (against the bad scenario).   
The right plot of Figure \ref{Fig_comp} depicts the weight on the lower Merton solution for varying levels of relative risk aversion. Again, the black line refers to learning while the dashed one is pre-commitment. Observe that the strategies (the weights on the lower Merton solution) coincide in the case of a log-investor  ($\alpha\rightarrow 0$) with level of relative risk aversion equal to one. 
Overall, however, our numerical results do not show a very big difference between the pre-commitment and the learning strategy.

\subsubsection{Value of learning and impact of investment horizon}\label{subsec:value_of_learning}
An intuitive explanation for the similarity of learning and pre-commitment at $t=0$ is given by means of the conditional distribution of $\vartheta$. Along the lines of Remark 2.2, it holds
\begin{align*}
P\left(\left.\Theta=\overline\vartheta\right|\mathcal F^Y(t)\right)
& =  p\frac{L(\overline\vartheta,Y(t))}{F(t,Y(t))}=:\hat p_t \text{ and }
P\left(\left.\Theta=\underline\vartheta\right|\mathcal F^Y(t)\right)=1-\hat p_t
\end{align*}
such that $\hat\Theta(t)=E\left[\left.\Theta \right|\mathcal F^Y(t) \right]=\hat p_t\overline\vartheta+(1-\hat p_t)\underline\vartheta$.
Recall that, at $t=0$, the investor maximizes her expected utility under the prior distribution. Thus, the optimal strategy depends on the expected value of   $\hat\Theta(t)$ which is
\begin{align*}
\EE\left[\hat\Theta(t)\right] & =  p \overline\vartheta
+(1-p)\underline\vartheta
\end{align*}
by the martingale property of $\hat\Theta$. This is the same for a pre-commitment and a learning investor.

In summary, at $t=0$, both (pre-commitment and learning) have the same information and the investor only expects to learn the prior distribution such that for the log-investor  there is no difference between the pre-commitment and learning strategy (cf. Remark 2.2), and a negligible difference for an investor with $\alpha\neq 0$.

However, the {\em evolutions} of the strategies under learning and  pre-commitment are rather different. While the pre-commitment investment fraction is (per definition) constant over time, the time $t$ investment fraction under learning is a random variable defined by $Y(t)$.   Notice that in our stylized illustration the true model is  either given in terms of $\vartheta=\overline\vartheta$ (Model 1 where $Y(t)\sim N\left(\overline\vartheta t, t\right)$) or $\vartheta=\underline\vartheta$ (Model 2 where $Y(t)\sim N\left(\underline\vartheta t, t\right)$) such that we can compare the conditional distributions of the investment fraction, i.e. given Model 1 and Model 2. In order to focus on the impact of learning, consider the log-investor whose strategy  is exclusively specified in terms of $\hat p_t$.
Straightforward calculations imply that with $Z^*\sim N(0,1)$ it holds
\begin{align*}
  \left.\hat p_t\right|_{\text{Model 1}} & \sim  \left(1+q \;L_t(\underline\vartheta-\overline\vartheta,\sqrt{t}Z^*)\right)^{-1}\\
  \left.\hat p_t\right|_{\text{Model 2}} & \sim  \left(1+q \exp\{(\underline\vartheta-\overline\vartheta)^2t\} \;L_t(\underline\vartheta-\overline\vartheta,\sqrt{t}Z^*)\right)^{-1}
\end{align*}

\begin{figure}[tb]
	\begin{center}
			{\bf{Conditional distribution of $\hat p_t$ }}
		\end{center}
	\begin{center}
\includegraphics[width=0.45\textwidth]{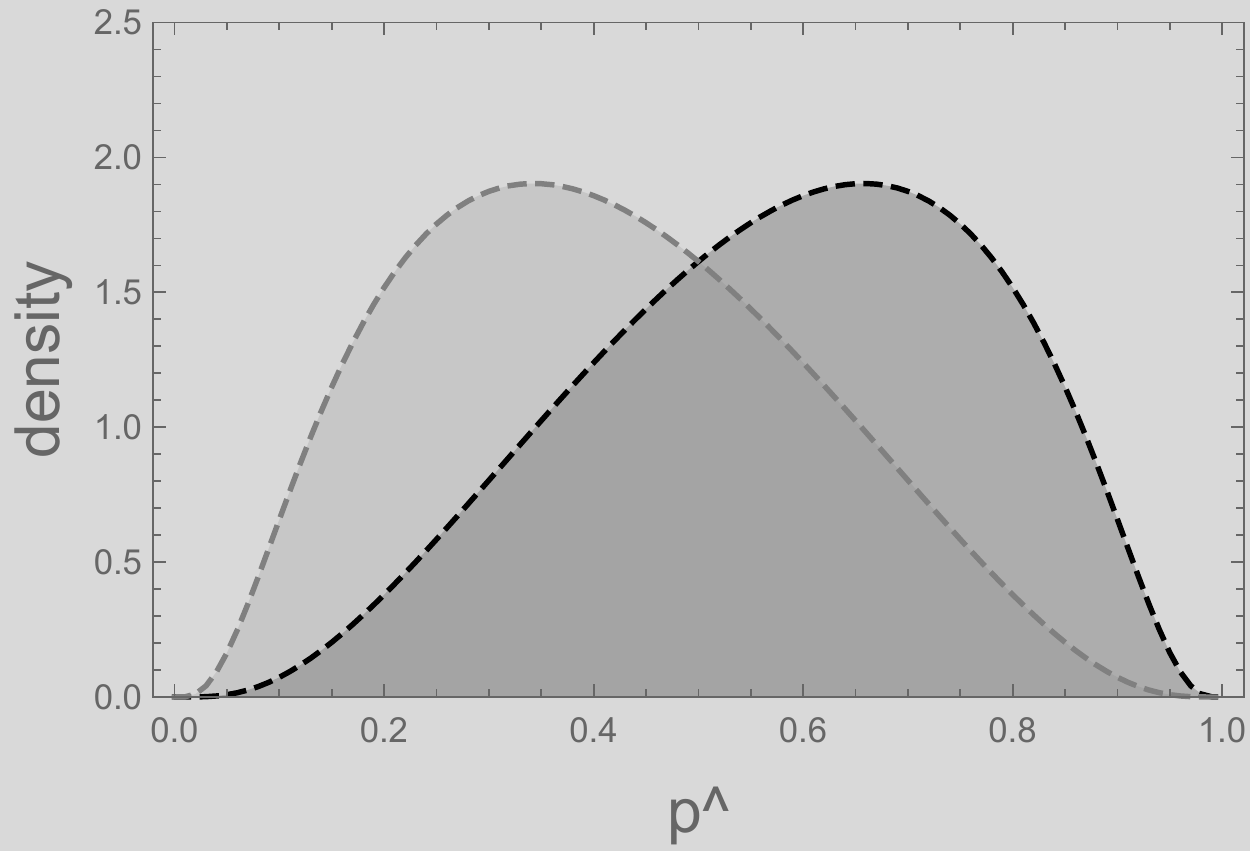}	
			\includegraphics[width=0.45\textwidth]{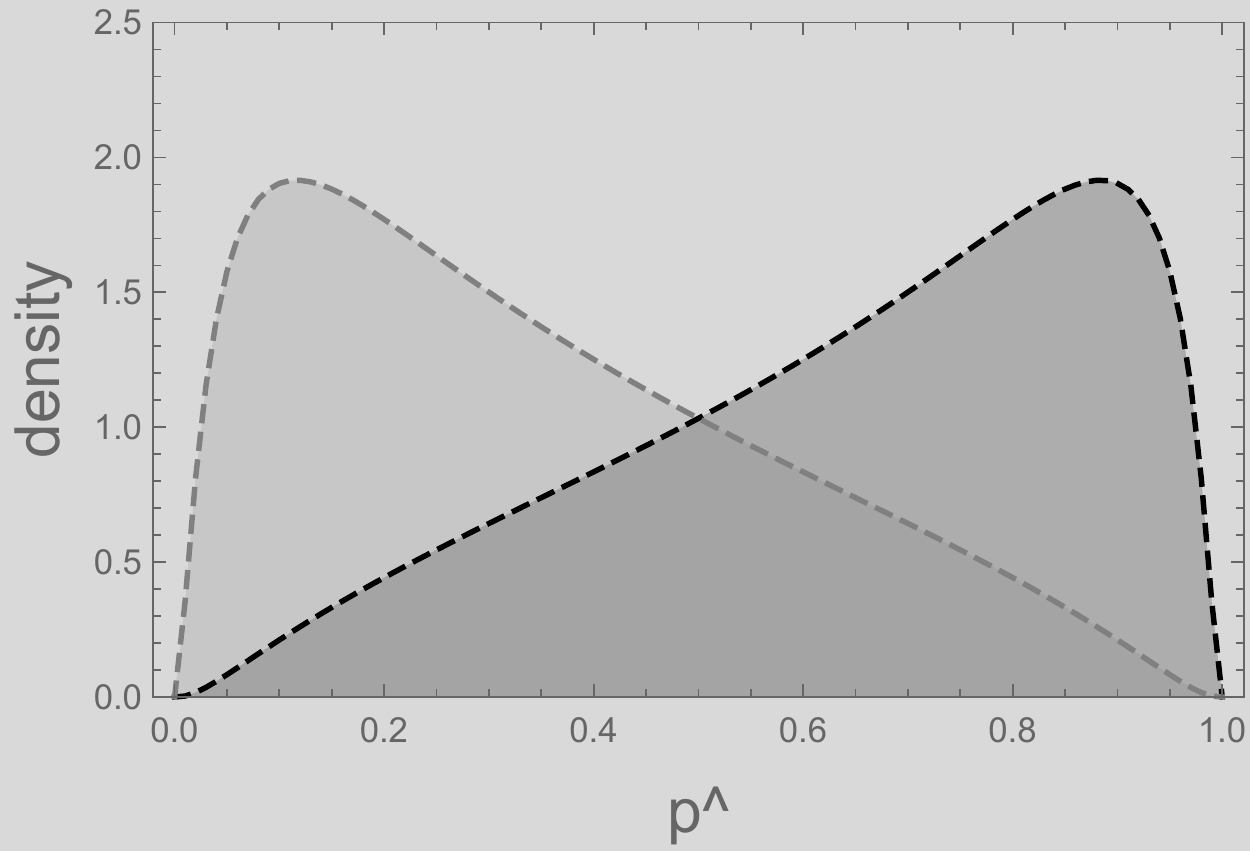}
\includegraphics[width=0.45\textwidth]{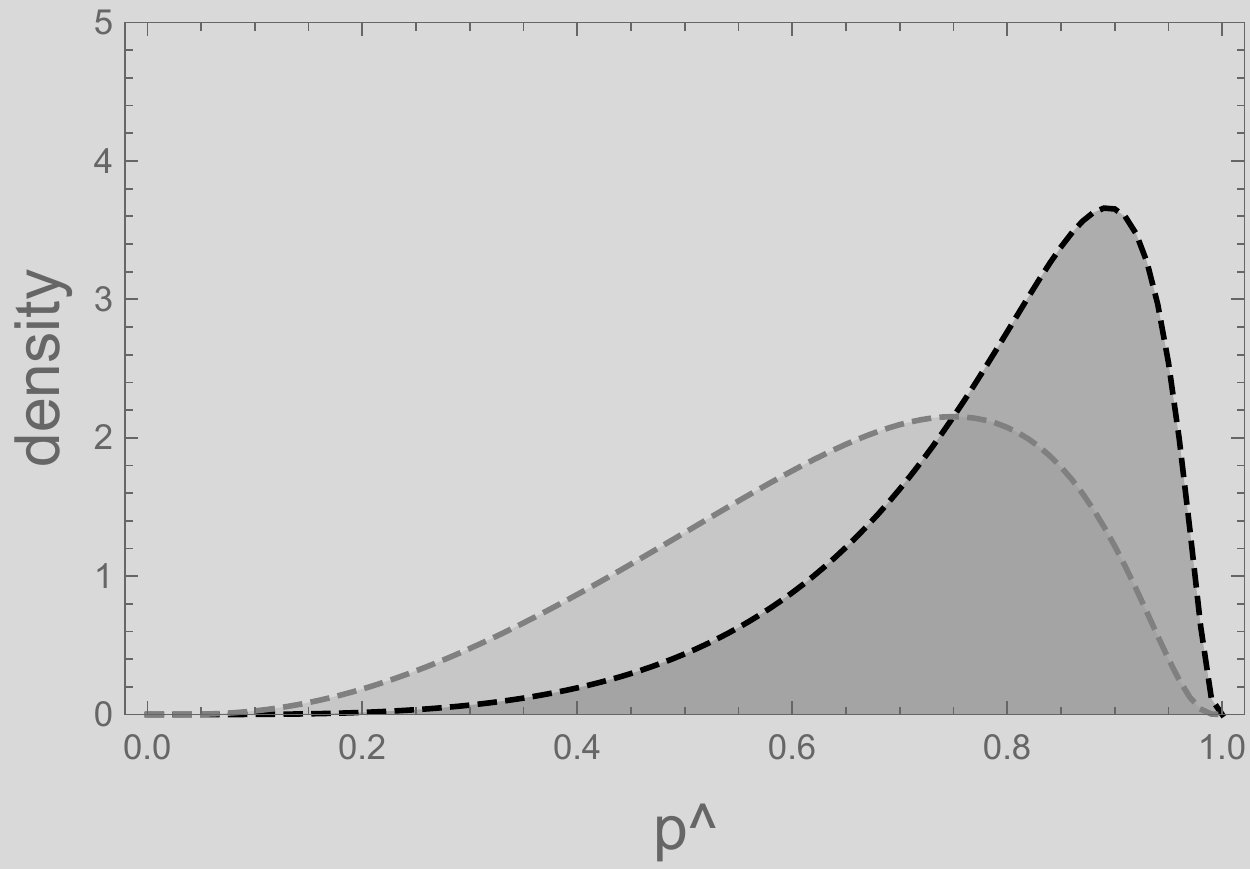}	
			\includegraphics[width=0.45\textwidth]{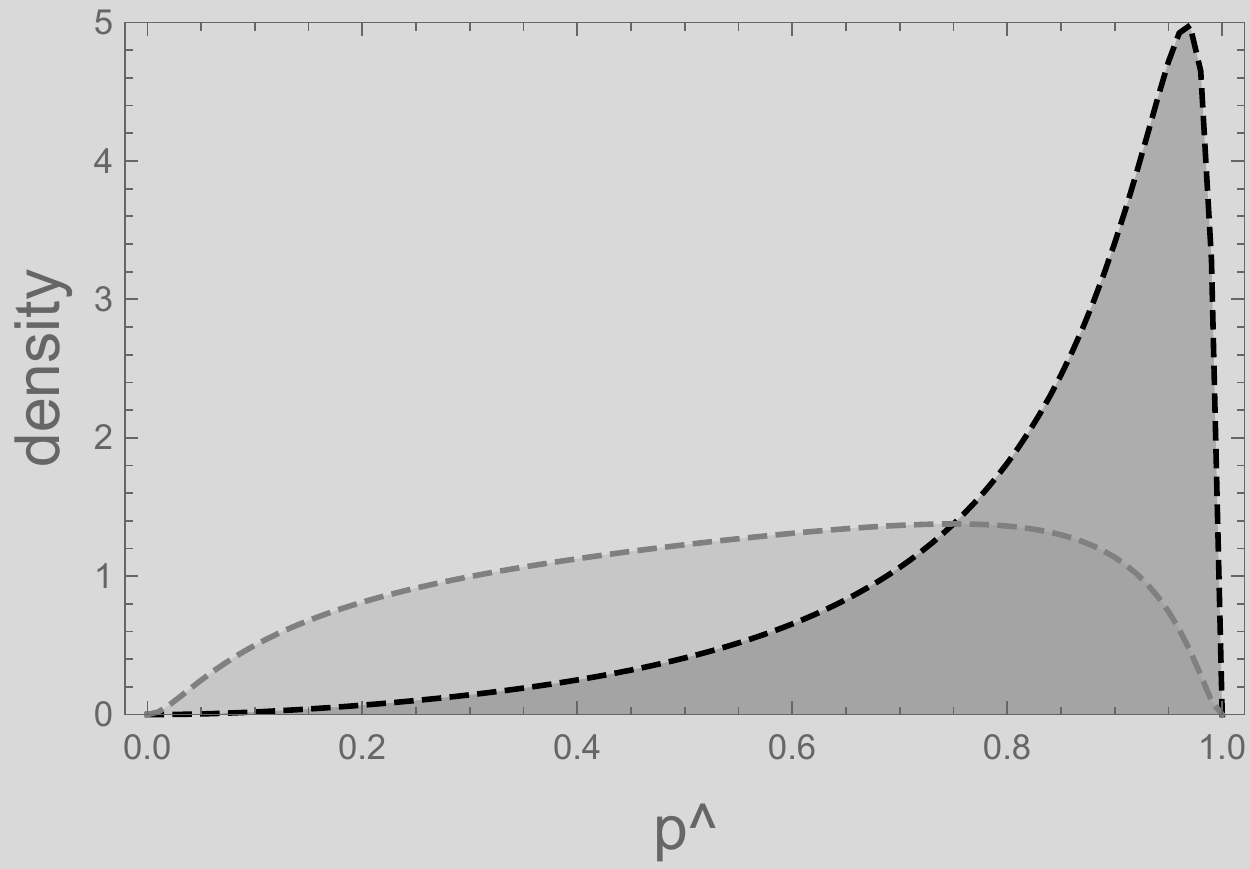}
\end{center}		
	\caption{The figure illustrates the conditional distribution of $\hat p_t$ given Model 1 (black) and Model 2 (gray). It is based on the benchmark model parameter setup of Table \ref{tab_bench_NA}. The left (right) hand plots refer to $t=5$ ($t=10$). For the upper (lower) plots it holds $p=0.5$ ($p=0.75$.)}\vspace*{0.15cm}
	\label{Fig_cond}
\end{figure}
An illustration of the conditional distributions is given in Figure \ref{Fig_cond}. Notice that comparing the densities in the first row at time points $t=5$ (left) and $t=10$ (right) that the true model is learned by the density. Indeed, at time point $t=10$ depending on what has been observed, the investment strategies will be quite different. However, the second row of pictures shows that learning may be slow when the initial prior is mainly wrong. Then the probability mass shifts rather slowly to the correct model.

However, a {\it{pure}} learning effect is only immanent in the strategy of the log-investor. Any deviation from the log-investor gives, in addition, rise to a competing effect caused by the dependence of the remaining investment horizon $T-t$. This is illustrated in Figure \ref{Fig_Y} which depicts the weight on the lower Merton solution for varying observations of $Y(t)$. Obviously, the weight on the lower Merton solution is decreasing in $Y(t)$, i.e. a higher observation is in favour of Model 1 (where $\vartheta=\overline\vartheta$). While the weights of the log investor (dashed lines) do not depend on the remaining investment horizon $T-t$ (upper plots versus lower plots), the horizon effect for an investor who deviates from the log-investor implies that the dependence of the weight on $Y(t)$ vanishes for increasing $T-t$, i.e. the more (black line) and less (gray line) risk averse investor moves to the lower (upper) Merton solution.   

\begin{figure}[tb]
	\begin{center}
			{\bf{Weight on lower Merton solution for varying $Y_t$}}
		\end{center}
	\begin{center}
	\includegraphics[width=0.45\textwidth]{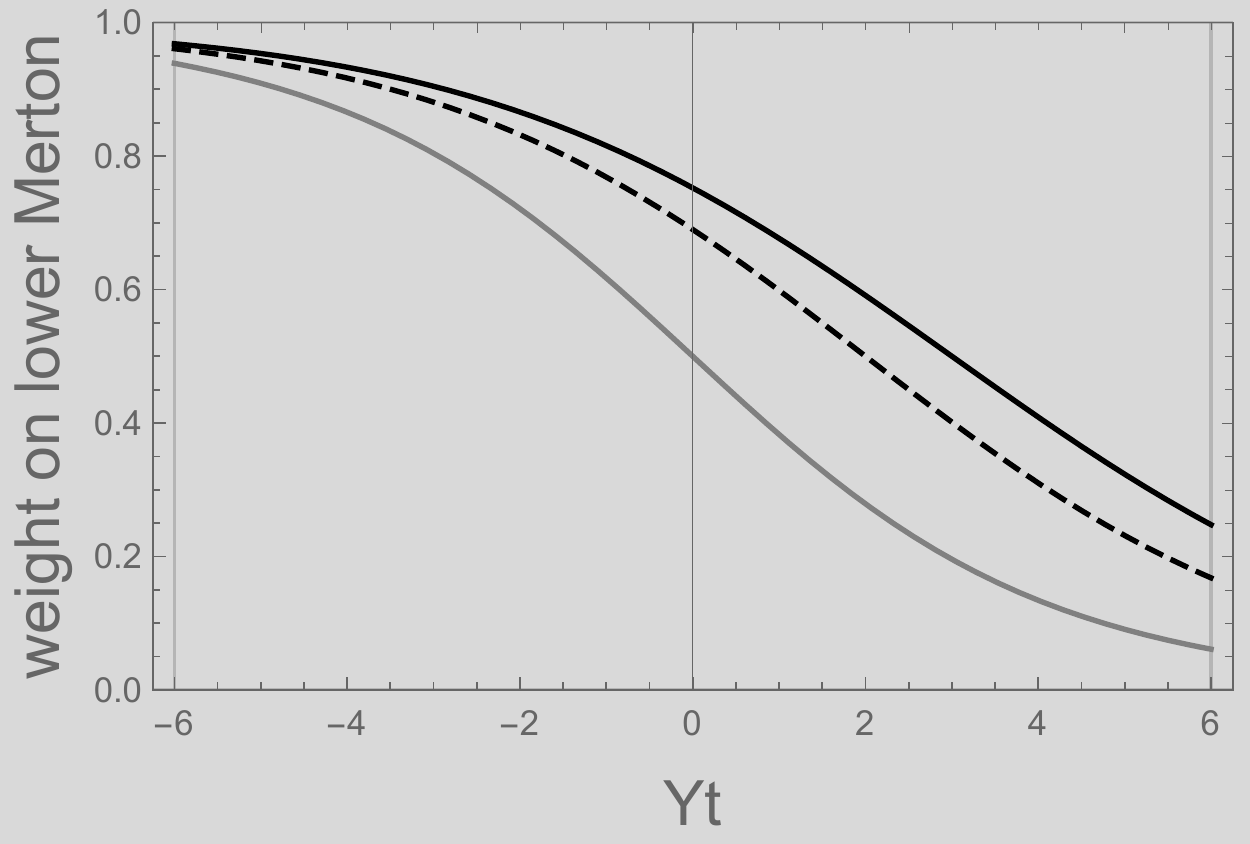}	
			\includegraphics[width=0.45\textwidth]{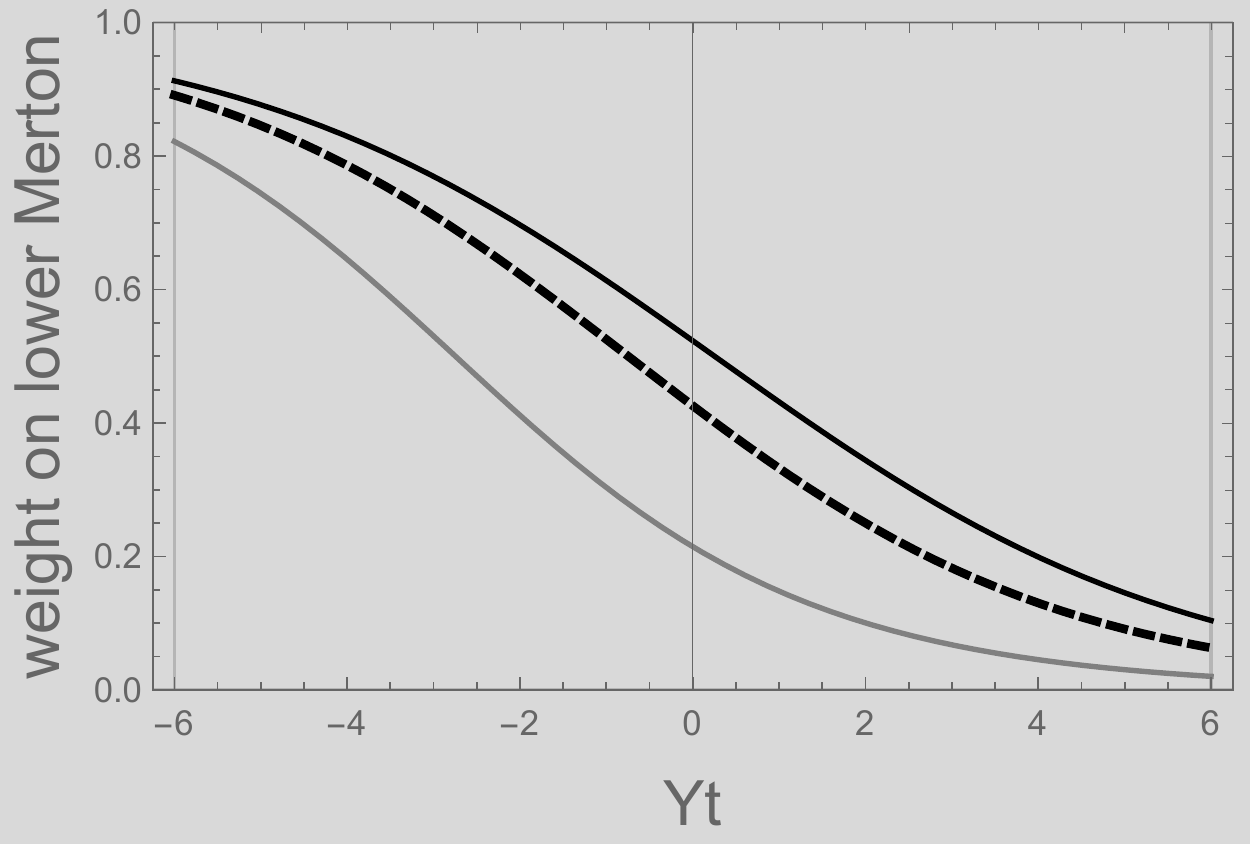}
             \includegraphics[width=0.45\textwidth]{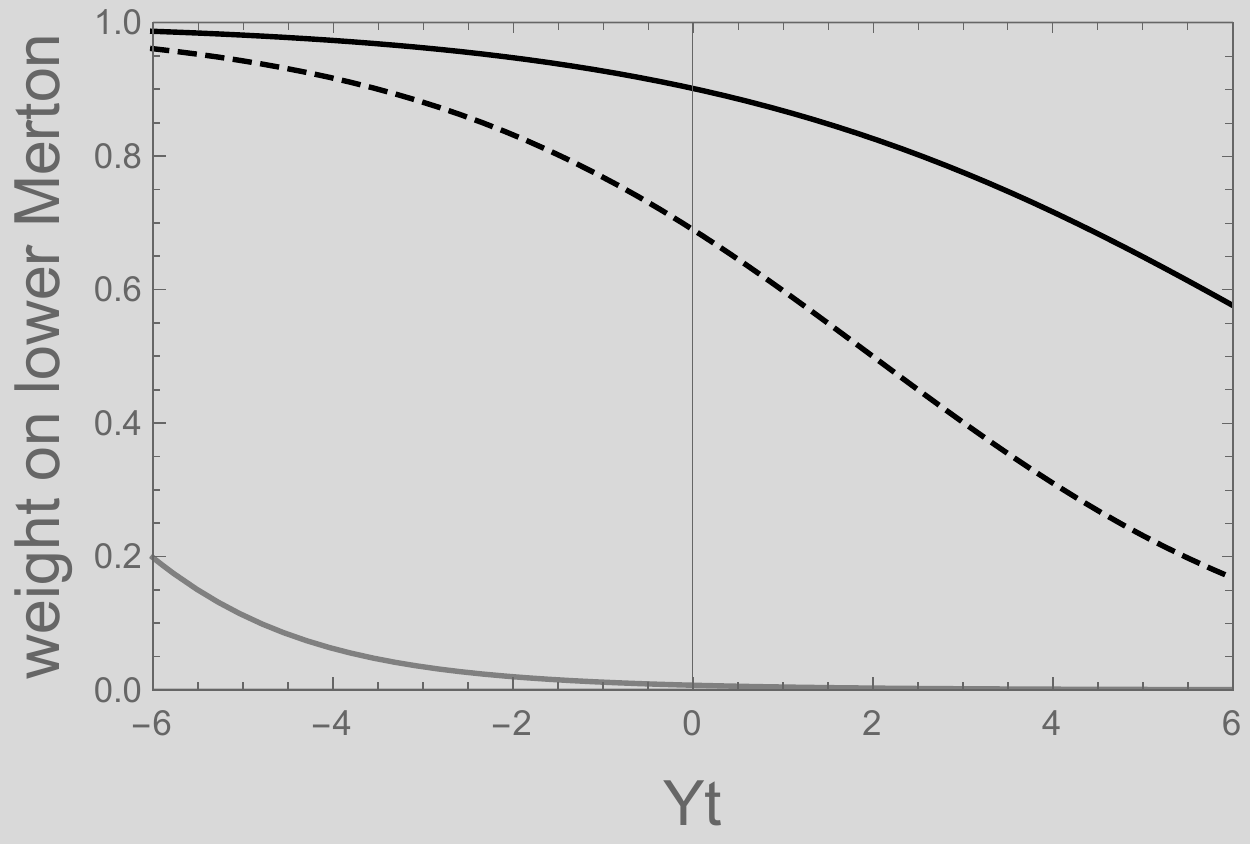}	
			\includegraphics[width=0.45\textwidth]{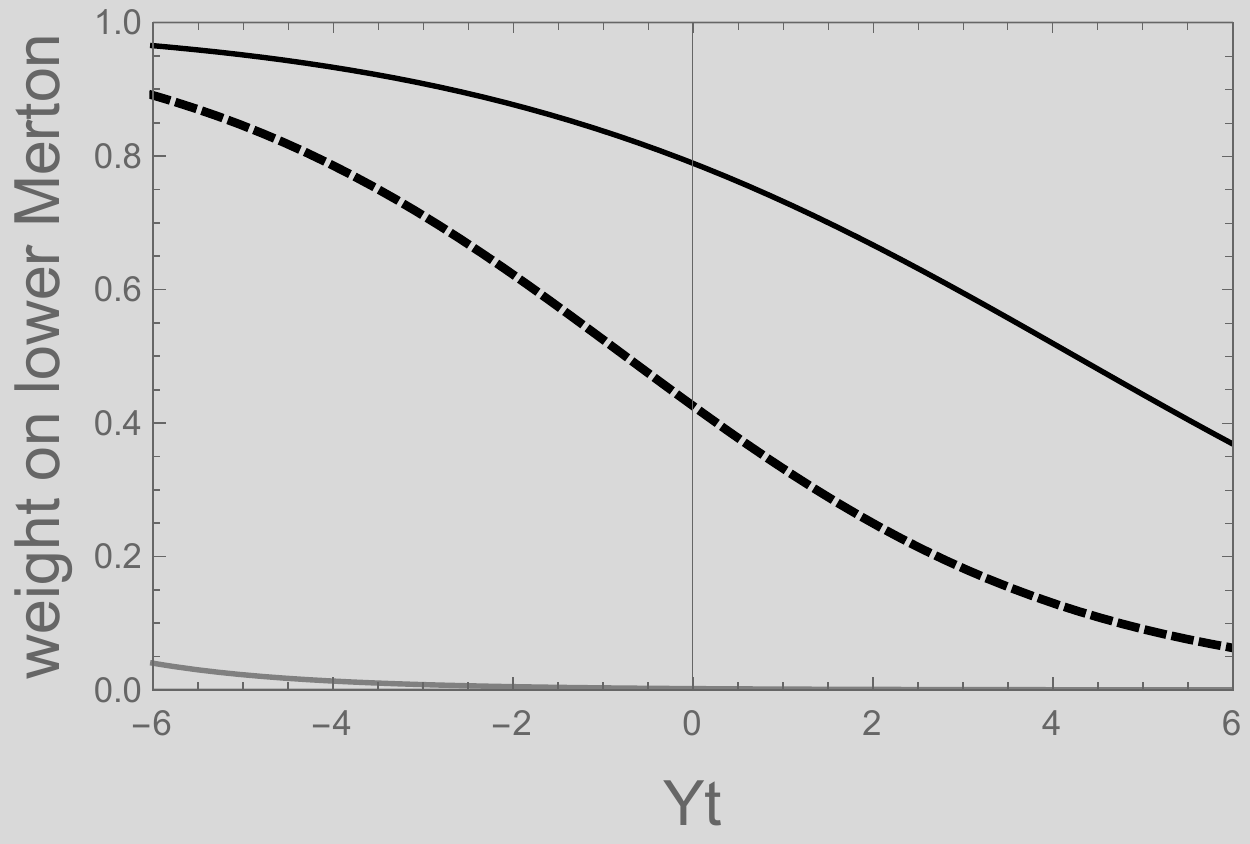}		
\end{center}		
	\caption{The figures illustrate the weight on the lower Merton solution at time $t$ and an investment horizon $T$ for varying $Y_t$. Each figure contains the log-investor
(dashed line) and the more (black line) and less (gray line) risk averse investor
than the log-investor. It is based on the benchmark model parameter setup of Table \ref{tab_bench_NA}. The upper (lower) plots refer to $t=5$ and $T=10$ ($T=30$). For the left (right) plots it holds $p=0.5$ ($p=0.75$.)}\vspace*{0.15cm}
	\label{Fig_Y}
\end{figure}

To assess the value of learning, we have to compare the expected utilities (certainty equivalents) of the optimal strategy under learning versus pre-commitment. For a {\it{pure}} value of learning, we consider  the log-investor.
Notice that the expected utility $v$ 
of the optimal strategy is (for $x_0=1$) given by
\begin{align*}
  v(x_0) & = 
  \int_{\R} F(T,z)\ln F(T,z)\varphi_{T}(z) dz.
\end{align*}
In contrast, the value of the optimal pre-commitment strategy is 
\begin{align*}
  v^{\text{pre}}(x_0) & =
\frac{( p \overline\mu + (1-p)\underline\mu)^2}{2\sigma^2}T.
\end{align*}
 For the log-investor, the value of learning in terms of the savings rate difference which is here given by
 \begin{align*}
\frac{1}{T} (v(x_0)- v^{\text{pre}}(x_0))
 \end{align*}
 is illustrated by means of Figure \ref{Fig_value_l}.
Here, the value of learning is increasing in $T$ (cf. left hand figure). Obviously, the value of learning is zero in the degenerate cases that $p=0$ ($p=1$, respectively) and obtains its maximum value approximately for $p=0.5$ (cf. right hand figure).

\begin{figure}[tb]
	\begin{center}
			{\bf{Value of learning}}
		\end{center}
	\begin{center}
	\includegraphics[width=0.45\textwidth]{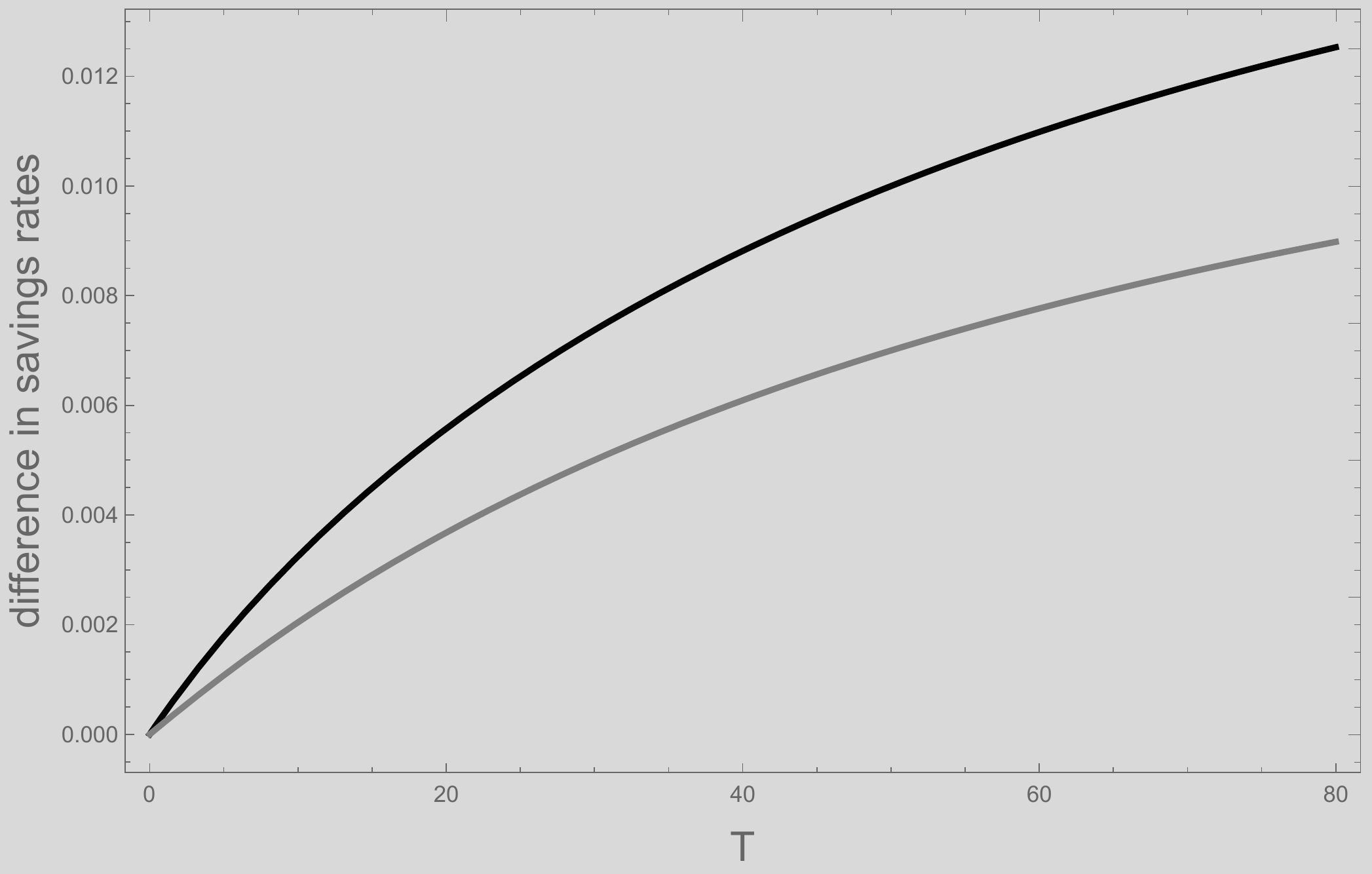}	
	\includegraphics[width=0.45\textwidth]{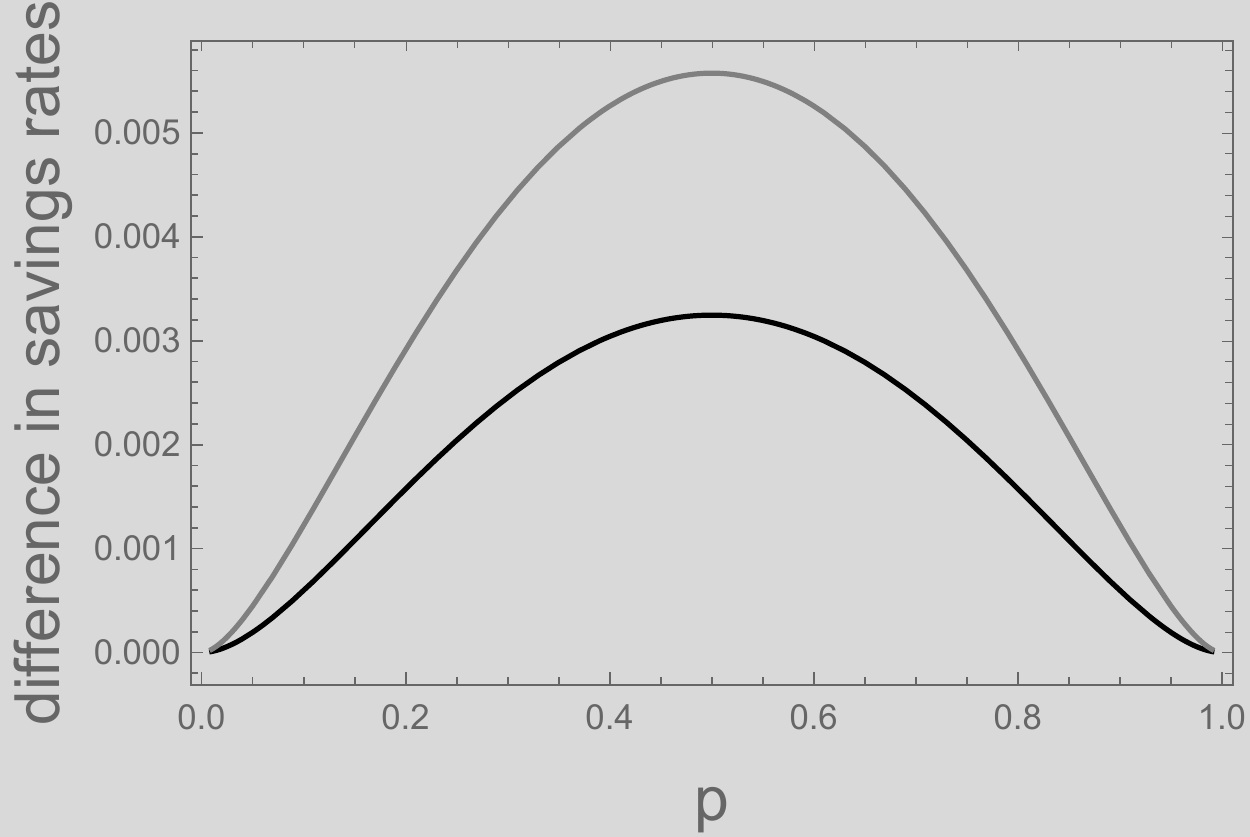}			
\end{center}		
	\caption{The figures illustrate the value of learning in terms of the savings rate difference (learning versus pre-commitment).  It is based on the benchmark model parameter setup of Table \ref{tab_bench_NA} and $p=0.5$ (black), $p=0.75$ (gray) in the left plot, $T=10$ (black), $T=20$ (gray) in the right plot. }\vspace*{0.15cm}
	\label{Fig_value_l}
\end{figure}

\subsubsection{Impact of good drift scenario and volatility on investment fraction.}
The previous sensitivities mainly referred to the  {\it{outer risk}} parameters $p$ and $T$. We now consider the sensitivities of the {\it{inner}} risk parameters  $\overline\mu$ and $\sigma$. While the {\it{outer}} risk parameters have no impact on the {\it{inner}} risk situation, the model parameters $\overline\mu$ and $\sigma$ give rise to an {\it{inner}} and  {\it{outer}} effect. Thus, we have to discuss the combined effects on the optimal investment fraction $\kappa(t,T,Y(t))$. To simplify the exposition, we only consider the case of a more risk averse investor than the log-investor along with our benchmark parameters and $t=0, Y(t)=0$.

First consider the drift parameter of the good scenario $\overline\mu$. The {\it{inner}} risk effect concerns the Merton fraction in the good scenario which is linearly increasing in  $\overline\mu$ while there is no impact on the lower Merton fraction. The directional effect on the investment fraction  $\kappa$ is positive.
However, the outer risk effect is captured by the hedging motive. Recall (cf. Figure \ref{Fig_new_a} and its explanation) that the higher $\overline\mu$ (the higher the difference between the
regimes), the larger is the impact of the hedging motive for the more (than log) risk 
averse investor. The higher $\overline\mu$, the higher is the weight on the lower Merton solution which decreases the investment fraction $\kappa$. In summary, there are opposing effects.
 
Now, consider the volatility $\sigma$. Notice that the volatility is a scaling factor to the Merton fractions. The higher the volatility is, the lower are both, the upper and lower Merton fraction, i.e. the {\it{inner}} risk effect to the investment fraction $\kappa$ is negative. In contrast, the {\it{outer}} risk effect is positive, i.e. a higher volatility decreases the difference between good and bad scenario (in terms of the market price per unit of risk or, alternatively, the Merton fractions of the inner risk situation) which reduces the hedging motive. Intuitively, the inner effect is the dominating one since $\kappa$ is given in terms of a weighted average of upper and lower Merton fraction.

\begin{figure}[tb]
\begin{center}{\bf{Impact of $\overline\mu$ and $\sigma$}}
\end{center}
	\begin{center}
\includegraphics[width=0.45\textwidth]{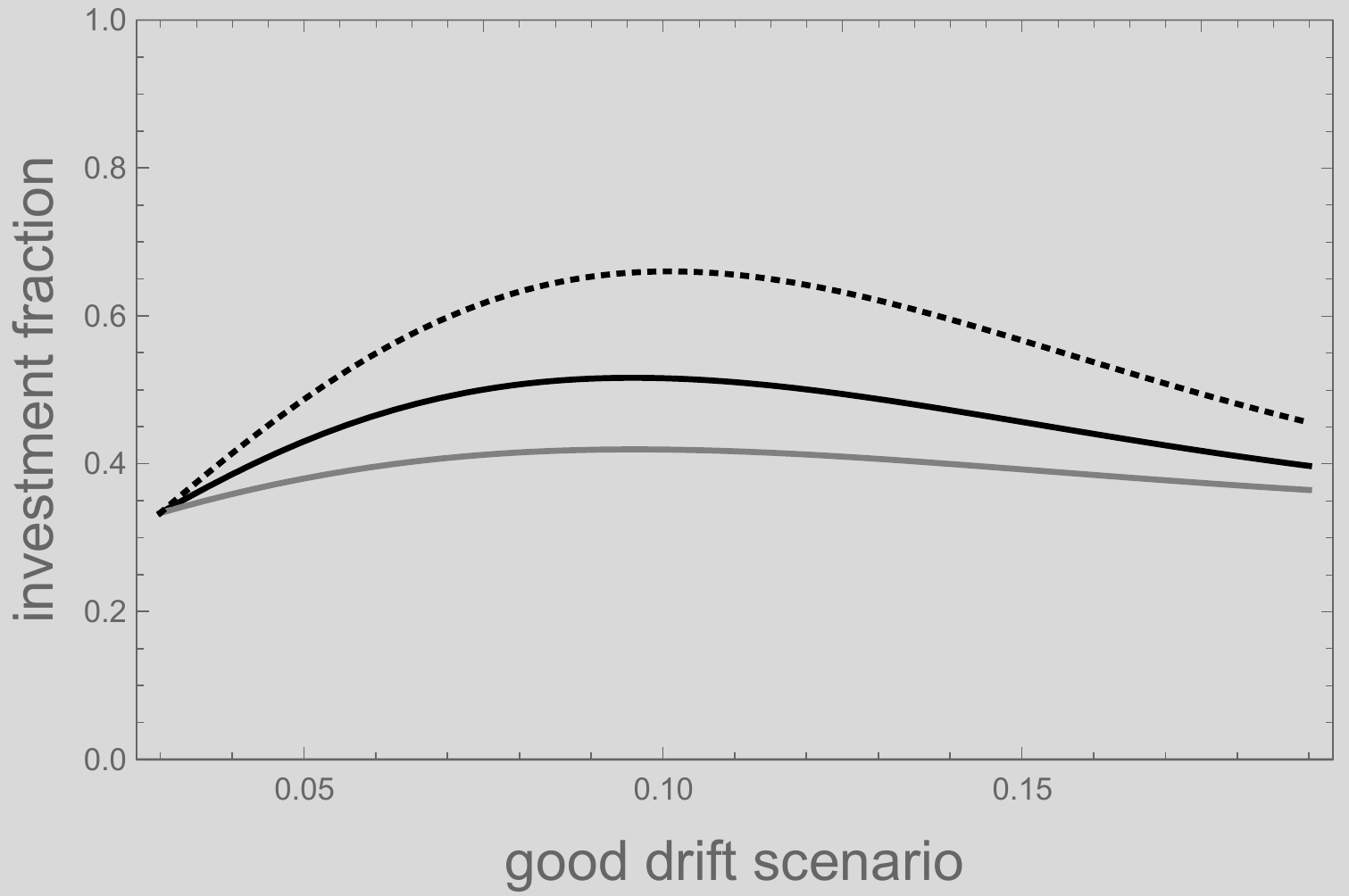}	
\includegraphics[width=0.45\textwidth]{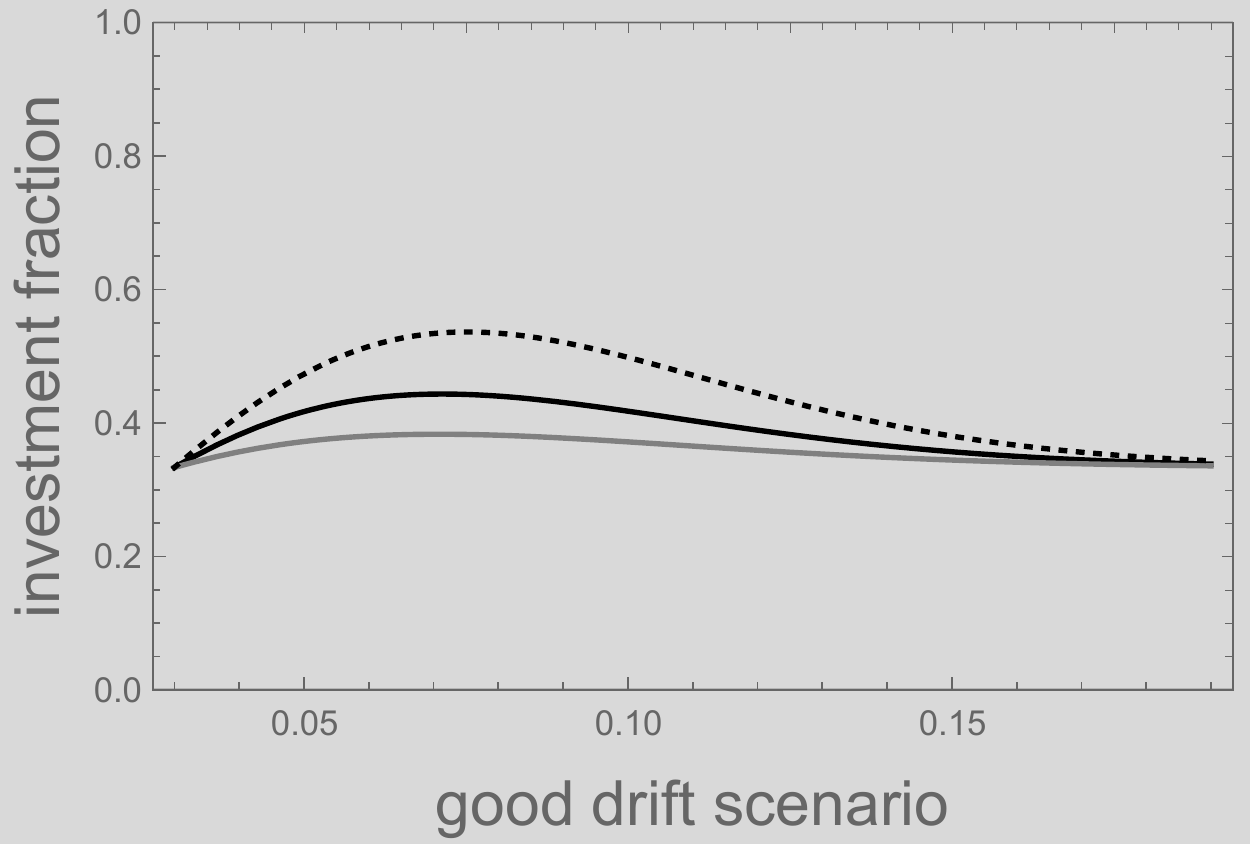}
\includegraphics[width=0.45\textwidth]{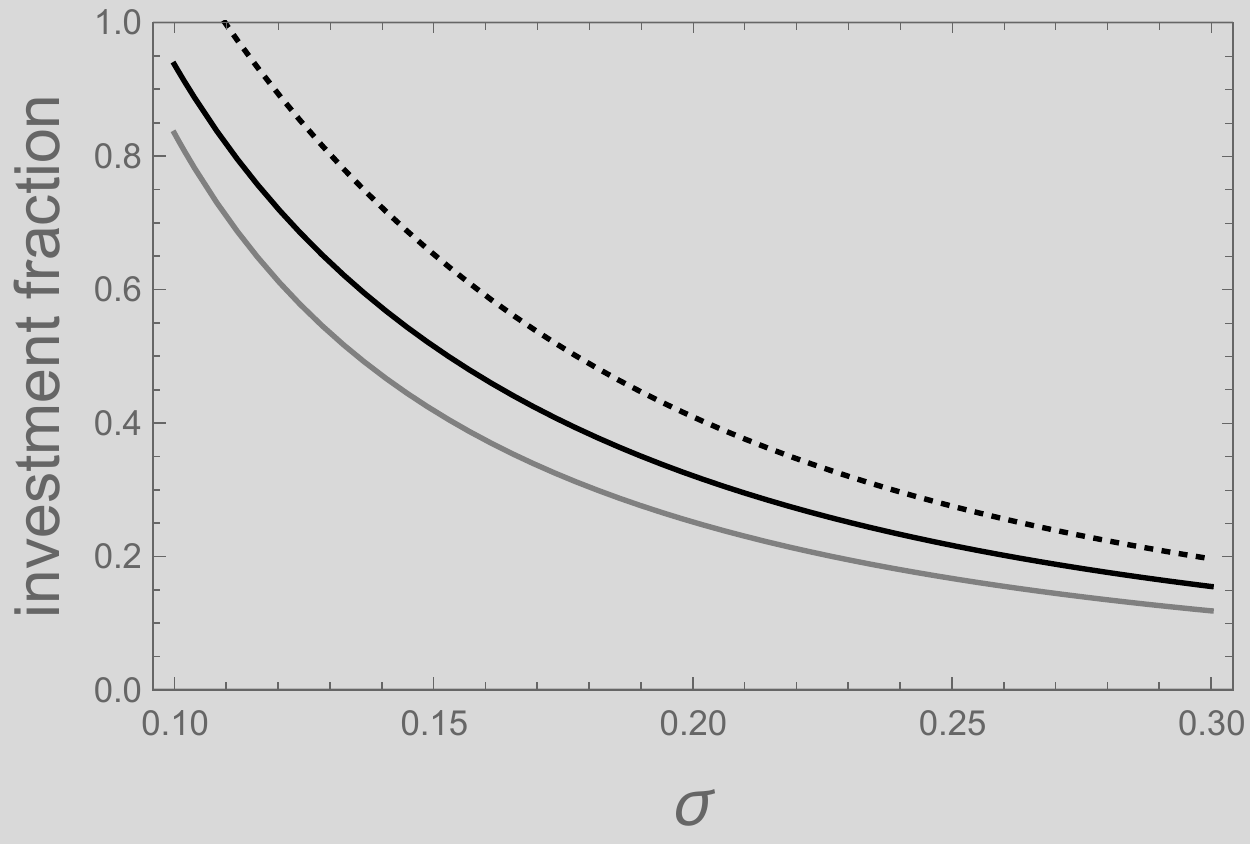}	
\includegraphics[width=0.45\textwidth]{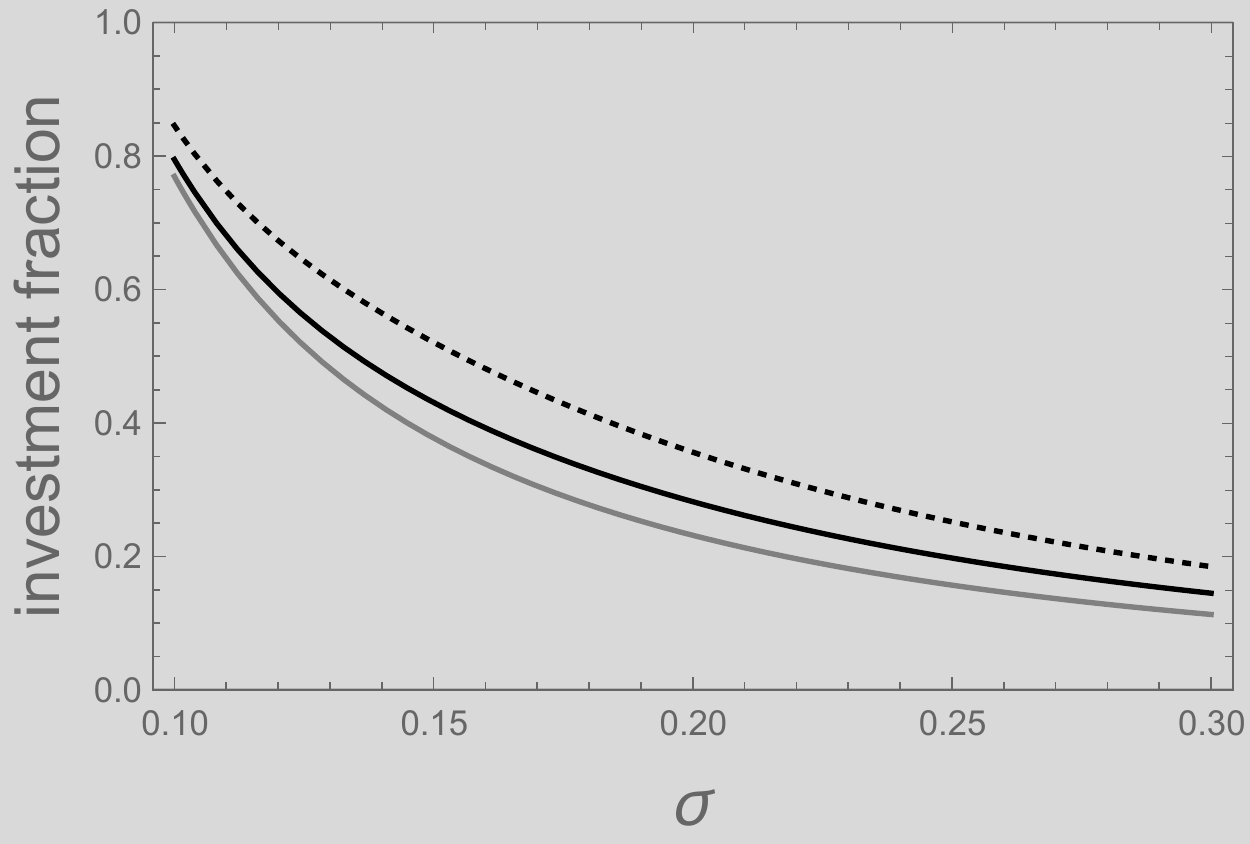}
\end{center}		
\caption{The figure depicts the impact of varying $\overline\mu$ (upper plots) and varying $\sigma$ (lower plots) on the investment fraction $\kappa^*$. It is based on the benchmark model parameter setup of Table \ref{tab_bench_NA} with the exception that the level of relative risk aversion is equal to 4 ($\alpha=-3$). While the black lines refer to $p=0.5$ (benchmark), the gray (dotted) lines are based on $p=0.25$ ($p=0.75$). In addition, the left (right) plots refer to $T=10$ ($T=20$).
}\vspace*{0.15cm}
	\label{Fig_inner_risk}
\end{figure}

An illustration is given by means of Figure \ref{Fig_inner_risk}. It depicts the impact of varying $\overline\mu$ (upper plots) and varying $\sigma$ (lower plots) on the investment fraction $\kappa$. Observe that, in spite of the opposing effects, the optimal investment fraction $\kappa$ is decreasing in $\sigma$. However, varying only the upper drift $\overline\mu$ emphasizes the opposing effects. The investment fraction is first increasing (i.e. the inner risk impact is dominating) while, after a critical drift level, it is decreasing (i.e. the hedging motive is dominating). In addition, observe that the critical level is decreasing in $T$.
Finally, it is also worth to emphasize that the inner risk parameters have a crucial impact on learning. Notice that $Y_t$ depends on both, $\overline\mu$ and $\sigma$.  

\subsection{Impact of ambiguity}
 Throughout the following, we assume that the investor is more risk averse than the log-investor, i.e. $\alpha<0$.
 Notice that the investor is risk averse if and only if $u$ is concave, while she is ambiguity averse if and only if $v$ is a  concave  transformation  of $u$.
 Assuming that $u$ and $v$ are CRRA functions with relative risk (ambiguity) aversion $1-\alpha$ ($1-\lambda$)  implies that an ambiguity neutral investor is characterized by $\lambda=\alpha$ while she is ambiguity averse (loving)   for $\lambda<\alpha<0$ ( $0>\lambda>\alpha$). 
 
Again, we consider the impact of the outer risk (evaluated by the ambiguity function $v$) by means of the weight on the lower (upper) Merton solution (implied by the utility function $u$), i.e. we are interested in $g(\alpha, p^{\text{mod}} ,\dots)$ where 
$$p^{\text{mod}}:=q_1^*/(q_1^*+q_2^*)$$
(cf. Lemma \ref{lem:optimaldistorsion1}, \ref{lem:optimaldistorsion2} and \ref{lem:optimaldistorsion3}). 
We call $p^{\text{mod}}$ the adjusted or modified probability.
The ambiguity neutral investor serves as a benchmark since  $p^{\text{mod}}=p$, i.e. the modified probability coincides with the prior probability.

Recall (cf. Subsection \ref{subsubsec_prior})  that the weight on the lower Merton solution is decreasing in the prior probability $p$.
Thus, introducing ambiguity aversion (loving) is equivalent to decreasing (increasing) the prior probability $p$ for the good drift scenario. 
Assume now that both, the level of risk aversion and ambiguity aversion are above the one of the benchmark log-investor, i.e. $\alpha<0$ and $\lambda<0$. Recall that we already discussed the special case $\lambda=\alpha$ in Subsection \ref{subsec_disc_a}. Intuitively, the investor chooses less risk  (a higher weight on the lower Merton solution) under ambiguity aversion, i.e. if the relative ambiguity aversion is higher than the level of risk aversion, i.e. if $\lambda<\alpha$ ($1-\lambda>1-\alpha$). In terms of the probability adjustment this implies that 
the modified  probability  $p^{\text{mod}}$ is smaller than $p$, i.e. the investor is characterized by a lower weight on the upper Merton solution than the ambiguity neutral investor. In contrast,  if $\lambda>\alpha$ ($1-\lambda<1-\alpha$) implies that $p^{\text{mod}}$ is higher than $p$.  
A numerical illustration of the modified probability is given in Figure \ref{Fig_ill_23_a} (left plot).
The level of relative risk aversion is $1-\alpha=4$. Notice that the modified prior probability $p^{\text{mod}}$  is decreasing in the level of relative ambiguity aversion $1-\lambda$  and tends to zero for $1-\lambda\to\infty$. Mathematically this is clear since the smaller of the two values in \eqref{prob:441} (note $\lambda<0$) dominates the optimization problem. From an economic point of view the very ambiguity averse investor tries to maximize her worst-case utility. In particular, it is above (below) the prior $p$ in the case that   $1-\lambda<1-\alpha$ ($1-\lambda>1-\alpha$).
Notice that the effect is more pronounced for higher times to maturity $T$, i.e. the impact on the modified prior probability is increasing in $T$. However, recall that the impact of the prior probability is decreasing in $T$, i.e. for $T\rightarrow \infty$ the (more than log) risk averse investor only relies on the worst drift scenario. Hence the effect of ambiguity aversion is fading for large time horizons. This is natural, since a large time horizon allows for perfectly learning the model.  For $T\to 0$ the modified prior tends to the initial given prior. This is because for a very short time horizon there is almost no difference between the two scenarios. However, recall that $p^{mod}$ depends on both, time horizon and ambiguity aversion. Thus, in the end the combination of both effects is crucial.  The overall effect on the investment fraction (represented by the weight on the lower Merton solution) is {\it{smooth}} (cf. Figure \ref{Fig_ill_23_a} (right plot)). 

\begin{figure}[tb]
	\begin{center}
			{\bf{Impact of ambiguity aversion on modified probability $p^{\text{mod}}$ and weight on lower Merton solution}}
		\end{center}
	\begin{center}
			\includegraphics[width=0.45\textwidth]{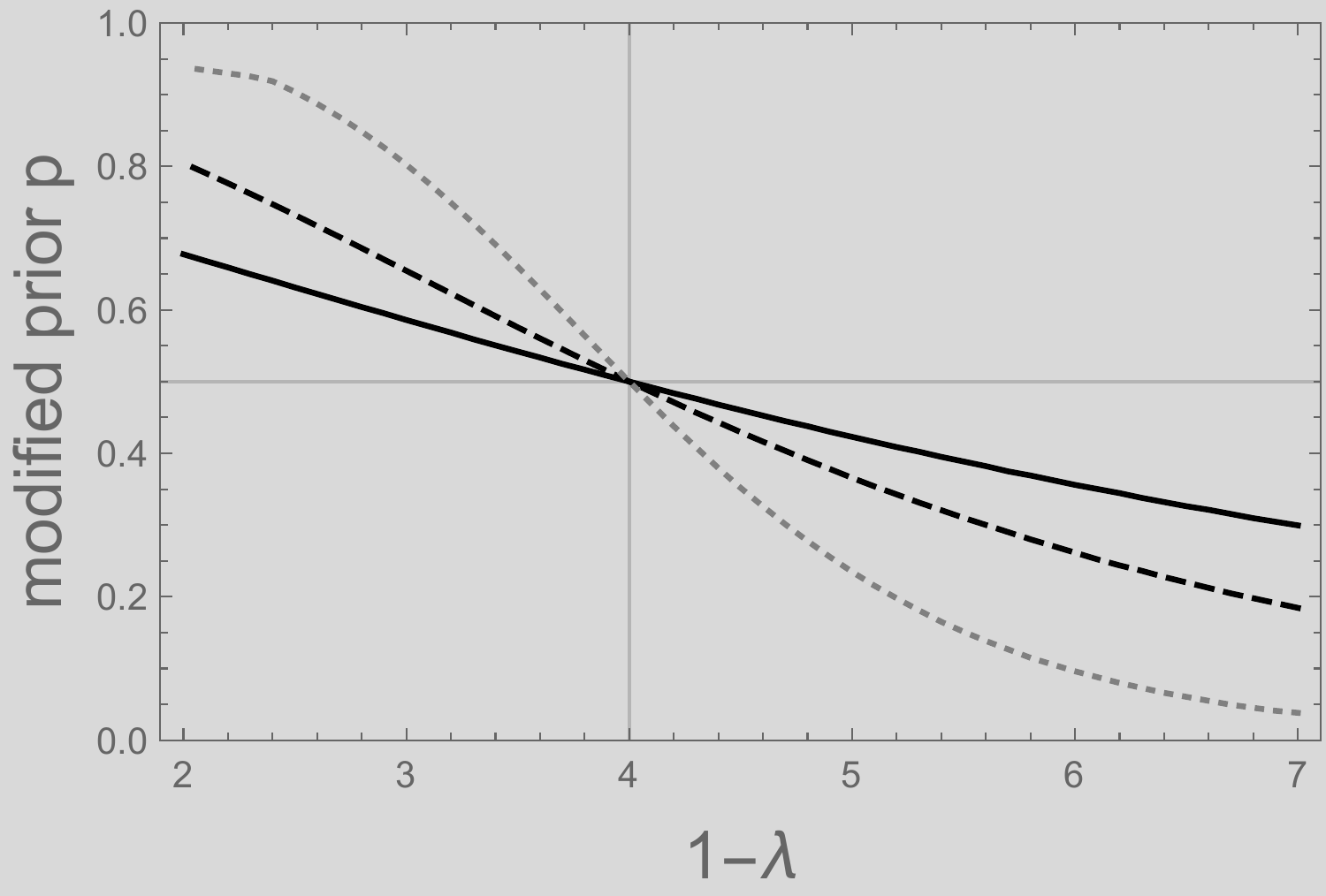}	
			\includegraphics[width=0.45\textwidth]{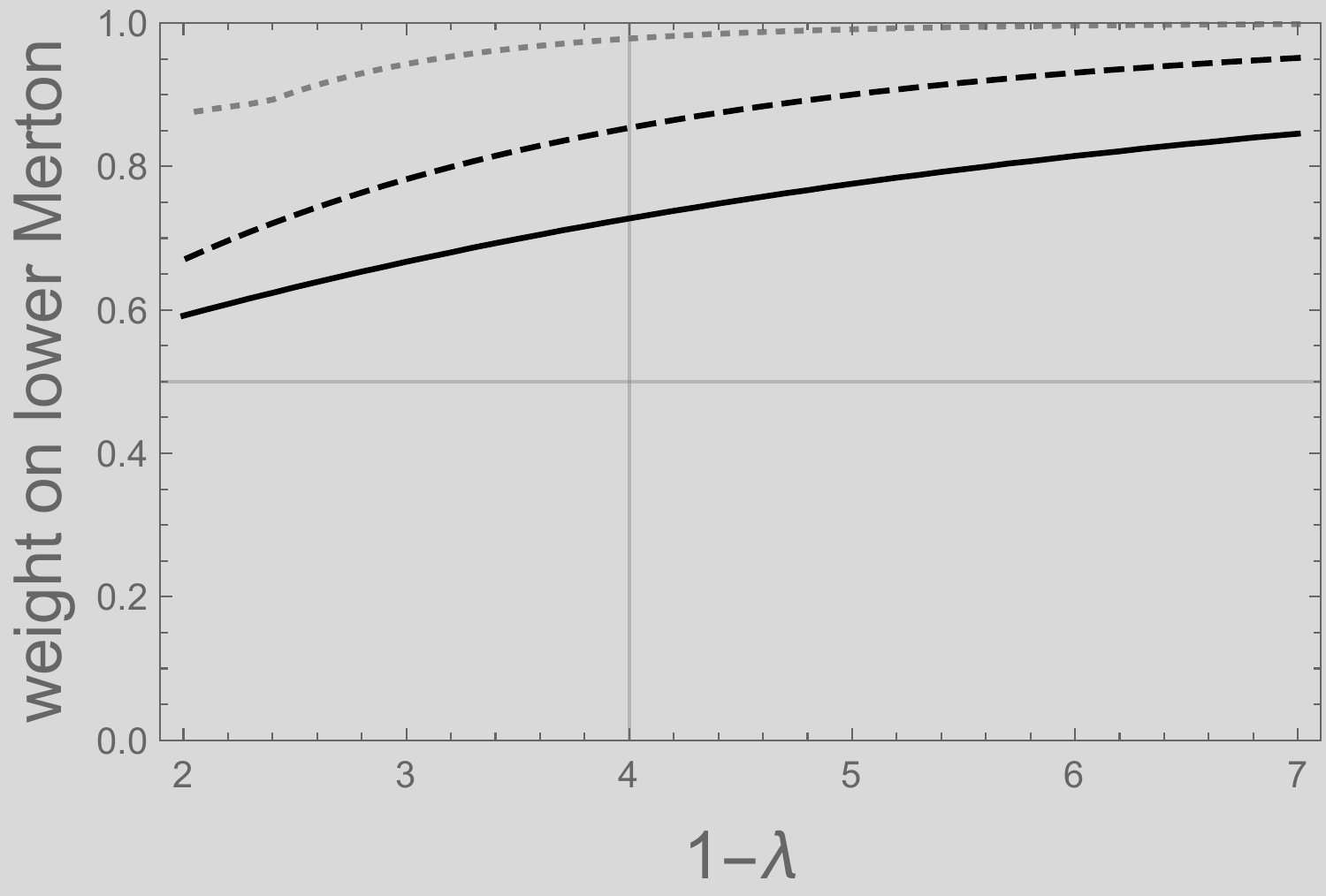}	
\end{center}
\caption{Model parameters are given as before (cf. Table \ref{tab_bench_NA}).  The prior probability for the good drift scenario is $p=0.5$, the level of relative risk aversion is equal to 4 ($1-\alpha=4$). The figure on the left hand side depicts the modified prior probability for varying levels of relative ambiguity aversion $1-\lambda$.
The figure on the right hand side gives the corresponding weights on the lower Merton solution.
The black lines refers to an investment horizon of  $T=10$ years, the dashed one to $T=20$, and the gray one to $T=50$.}\vspace*{0.15cm}
	\label{Fig_ill_23_a}	
\end{figure}

\subsection{Comments on implementation and practical relevance}
We give some brief comments on the practical relevance of the derived strategies w.r.t.\  their implementation on historical return data, the possibilities to test the out-of-sample performance and the complexity caused by these problems.
Since the objective to maximize the expected value (w.r.t.\ different drift scenarios) of the expected utility of an investor with constant relative risk aversion is motivated by estimation risk, we
also refer to other approaches which deal with estimation risk and some empirical results.

In the first instance, it is convenient to consider  the problem in the classic mean variance setup.
The discrete time version of the price dynamics of the risky assets $S_1,\dots, S_d$ has returns $r_{it}=\frac{S_{it}}{S_{i,t-1}}-1$ and log returns $\ln (1+r_{it})\approx r_{it}$.
Let $r_{ft}$ denote the return of the {\it{risk free}} asset, the excess return is defined by $R_t:=r_t-r_{ft}1_d$ where $1_d$ is an $d$-vector of ones. It is assumed that $R_t$ follows a multivariate normal distribution with mean $\mu$ and covariance matrix $\Sigma$.
For portfolio weights $\kappa$ (for the risky assets)
the objective function of a short-term expected utility maximizing  investor with constant relative risk aversion  is, in the above context, equivalent  to the mean-variance objective function
\begin{align*}
U(\kappa) = \mu^{\text{PF}}-\frac{1-\alpha}{2}\left(\sigma^{\text{PF}}\right)^2= \mu^{\text{PF}}-\frac{1}{2\gamma}\left(\sigma^{\text{PF}}\right)^2
\end{align*}
where
$1-\alpha=\frac{1}{\gamma}$ denotes the level of relative risk aversion,
 $\mu^{PF}=\kappa' \mu$ the portfolio mean, and $\left(\sigma^{\text{PF}}\right)^2=\kappa'\Sigma^{-1}\kappa$ the portfolio variance. For known $\mu$ and $\Sigma$, the optimal portfolio weights are
\begin{align*}
  \kappa^*=\frac{1}{1-\alpha} \Sigma^{-1}\mu= \gamma (\sigma^T)^{-1}\vartheta=: \kappa^{\text{Mer}}(\gamma,\vartheta)
\end{align*}
Since the true model parameters $\mu$ (and $\Sigma$) are not known, it is not possible to implement the optimal investment fraction  $\kappa^*$.  A naive way to solve the problem is the so called plug-in rule 
where the investor relies on the sample estimates of $N$ periods of observed returns (training data) and treats them as true parameters.

Along the lines of \cite{kan2007optimal}, a portfolio rule $\hat\kappa$ is defined by a function $f$ of the training data $R_1,\dots, R_N$ and the out-of-sample performance of a portfolio rule  can be measured by the expected utility $\tilde U$ conditional on the weights being chosen by $\hat \kappa$. In particular, they consider the loss by means of the difference $U(\kappa^*)-E[\tilde U(\hat\kappa)]$.  \cite{kan2007optimal} analytically derive the expected loss function for the (naive) plug-in approach and show that this approach can lead to very poor out-of-sample performance.

To implement our  (dynamic) strategy (which includes learning), we need a prior distribution for $\mu$. Thus,  it is, in the first instance necessary to determine a {\it{suitable}} estimator for the prior distribution  (cf. e.g. \cite{bauder2021bayesian}). An interesting question is whether it is possible to reduce the problem to a two point prior without causing too much sub-optimality (in particular, w.r.t the dynamic version of the strategy).
In addition, the implementation of the strategy affords an estimator for the covariance matrix $\Sigma$ (where $\sigma^T=\sigma^{-1} \Sigma$).
For example, \cite{ledoit2017nonlinear}
promote a nonlinear shrinkage estimator and show that their estimator dominates its competitors on historical stock returns data.

\begin{figure}[tb]
\begin{center}{\bf{S\&P500 index value and daily returns}}
\end{center}
	\begin{center}
\includegraphics[width=0.45\textwidth]{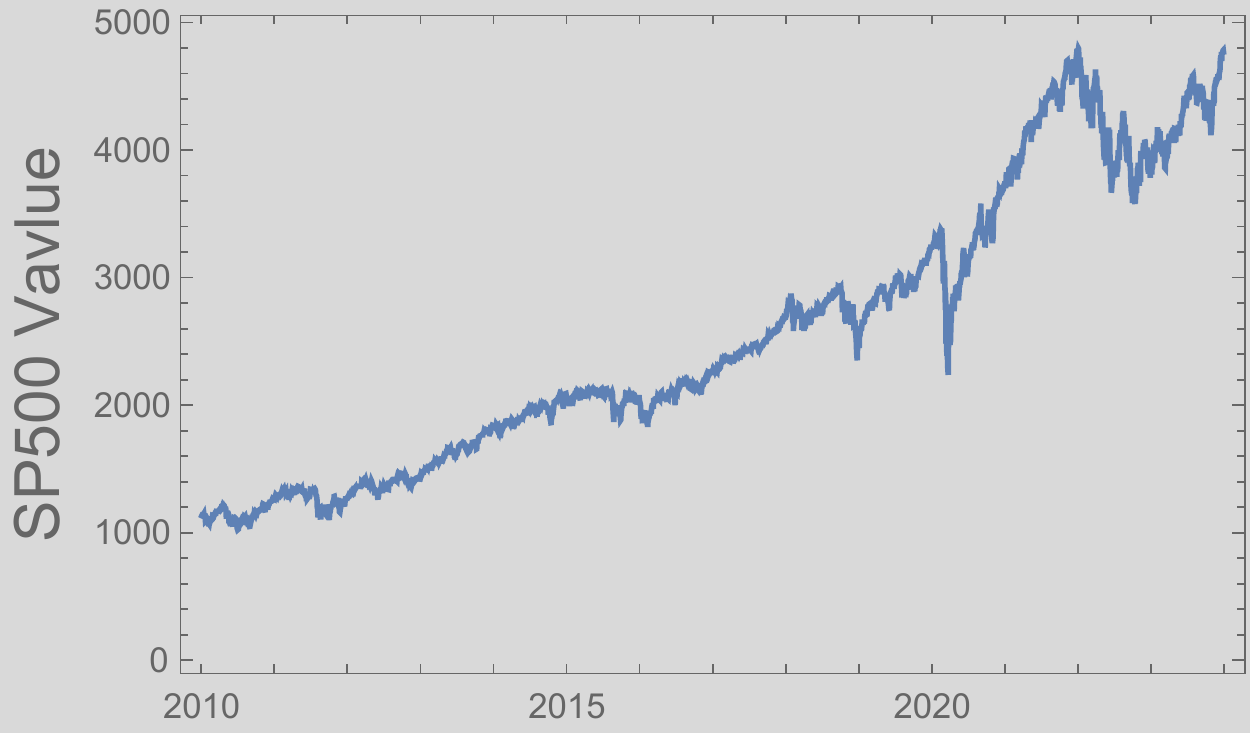}
\includegraphics[width=0.45\textwidth]{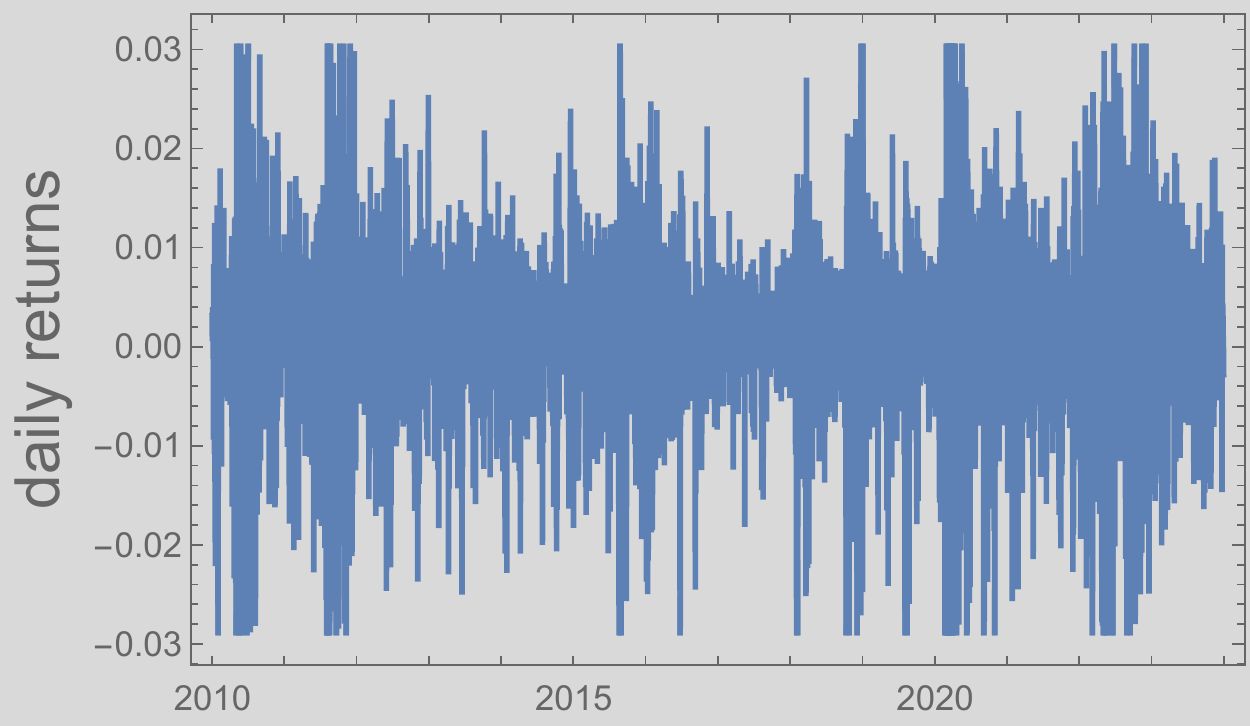}
\end{center}
\caption{S\&P500 price path and daily returns from $2010.01.01$ to $2024.01.01$.}
\label{Fig_SP500}
\end{figure}

\begin{figure}[tb]
\begin{center}{\bf{Comparison of strategies for volatility windows (250 versus 50 days)}}
\end{center}
	\begin{center}
\includegraphics[width=0.45\textwidth]{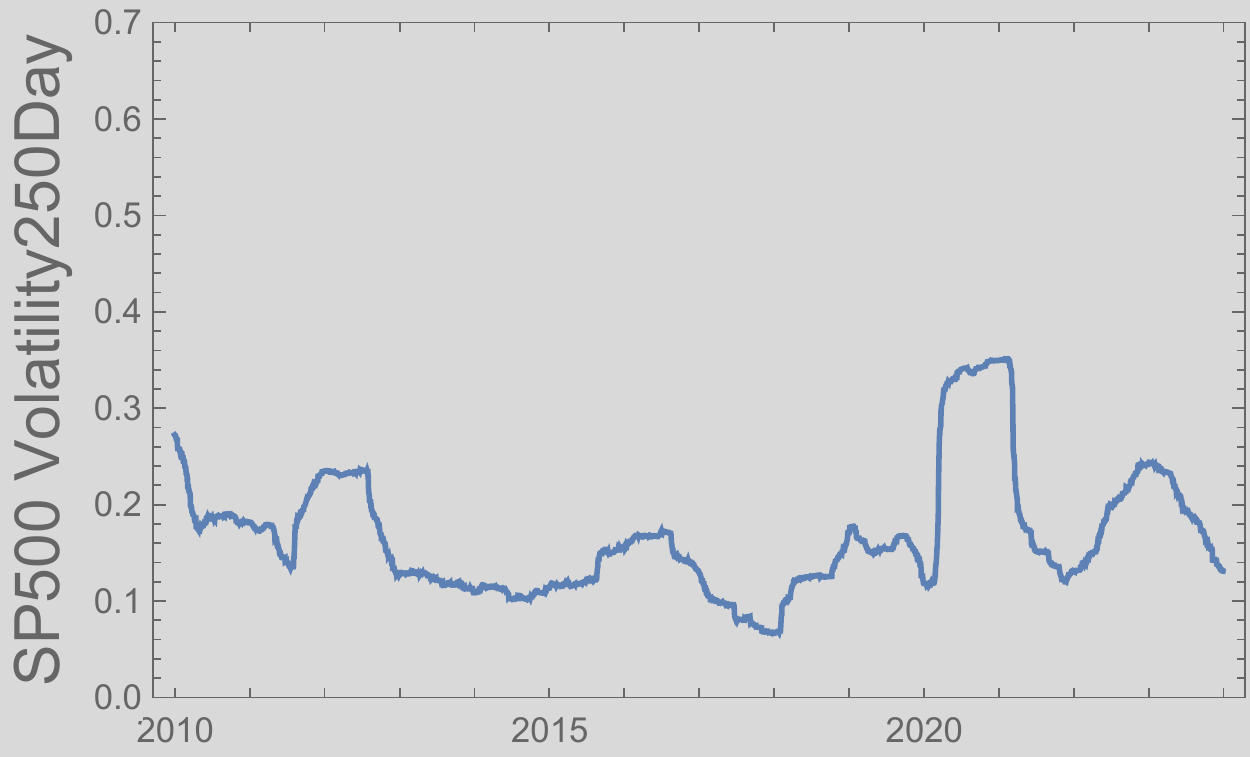}
\includegraphics[width=0.45\textwidth]{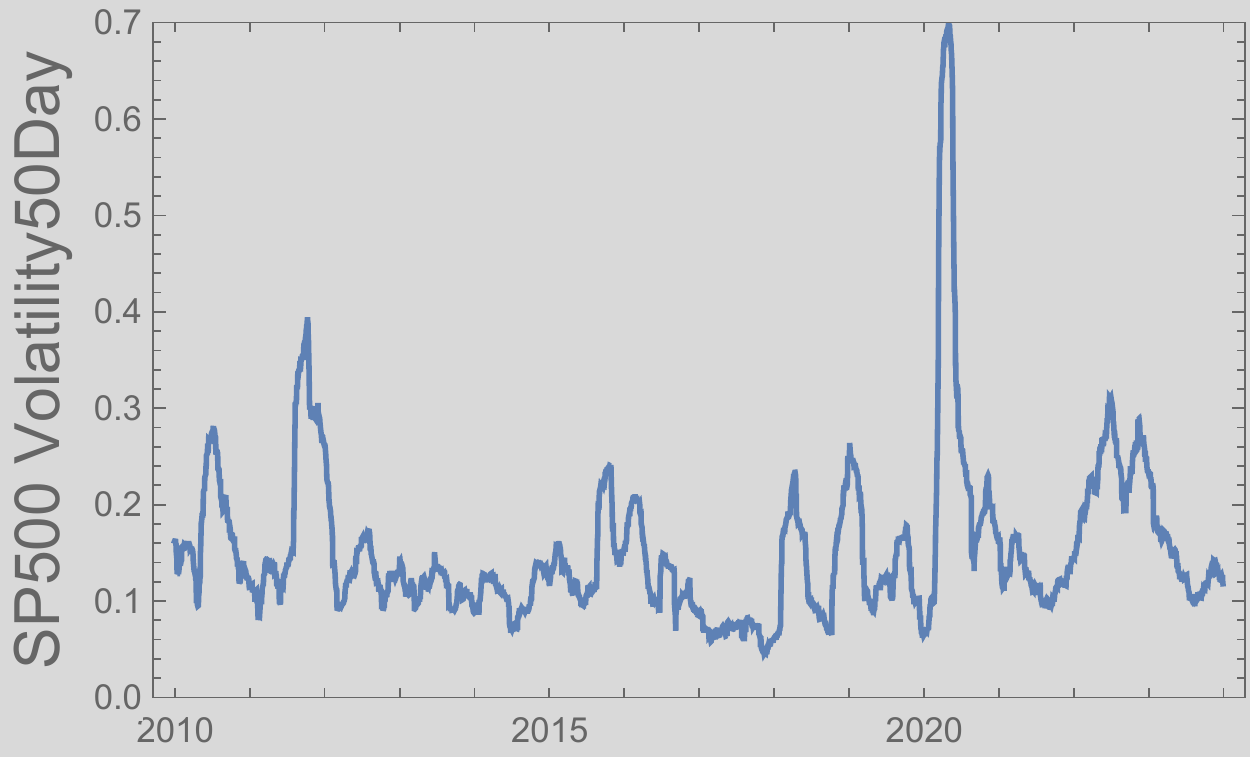}
\includegraphics[width=0.45\textwidth]{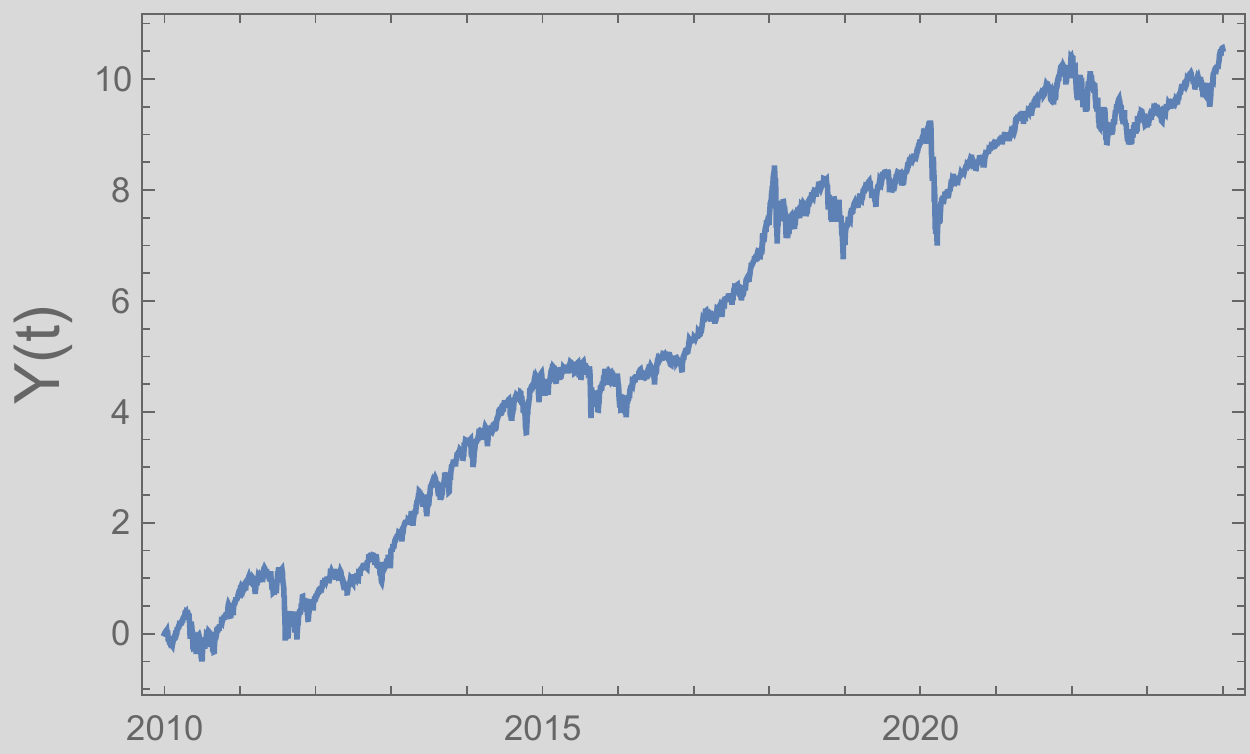}
\includegraphics[width=0.45\textwidth]{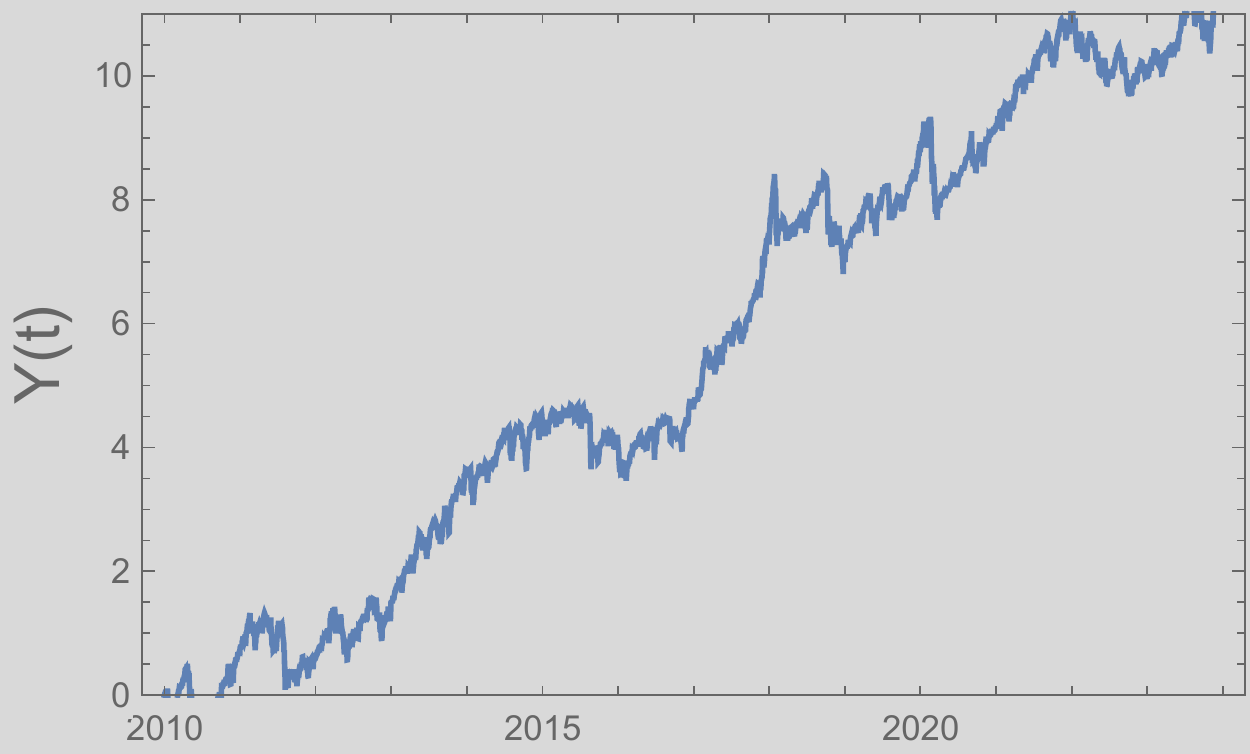}
\includegraphics[width=0.45\textwidth]{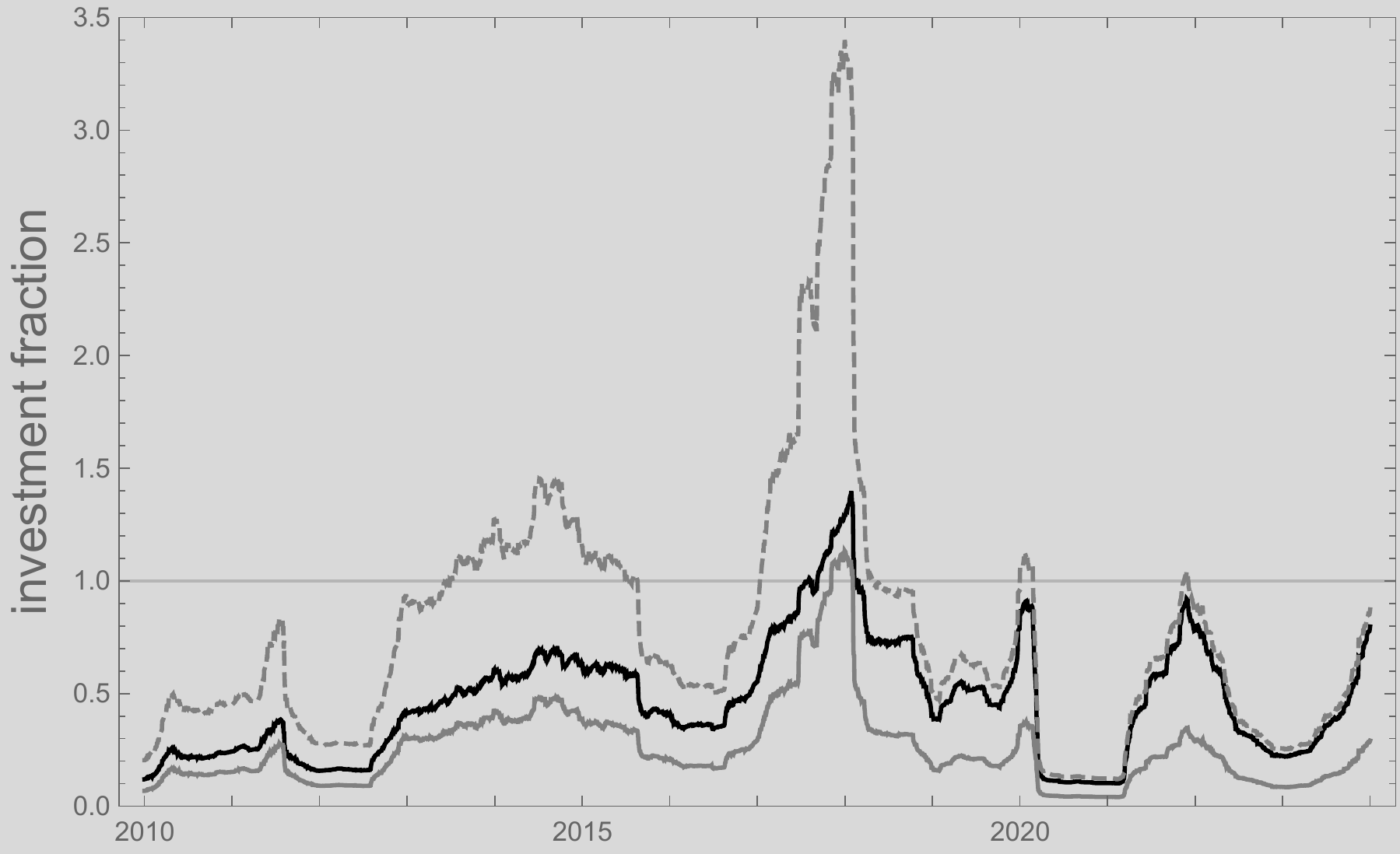}
\includegraphics[width=0.45\textwidth]{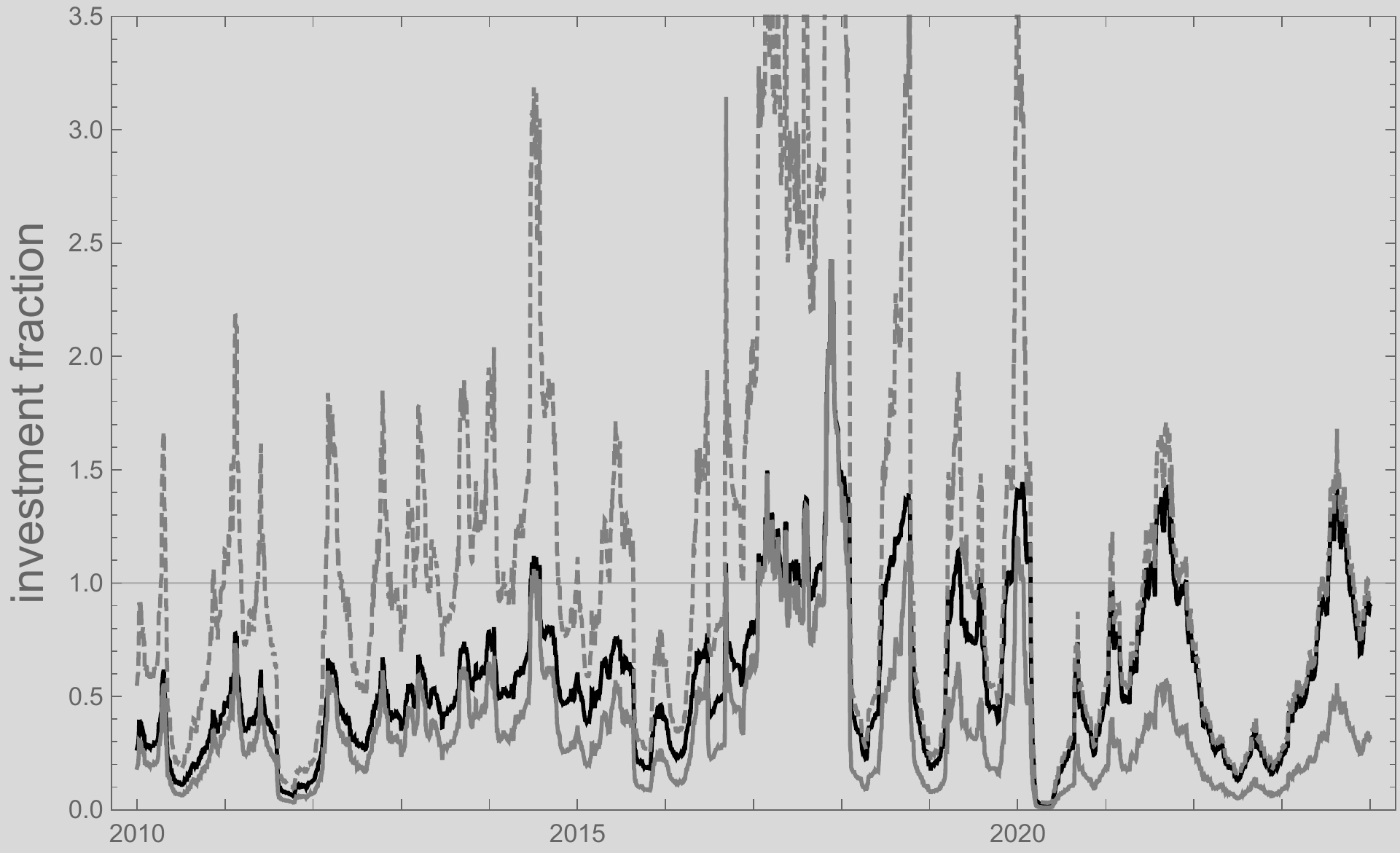}
\end{center}
\caption{The figure depicts the historical volatilities of the S\&P500 index for a rolling window of 250 and 50 trading dates (upper plots). Scaling the daily S\&P500 index returns with the historical volatilities then gives the daily increments of $Y_{t}-Y_{t-1}$ where $Y(0)=0$. The plots in the middle illustrate $Y(t)$ based on the volatility estimates resulting from a window of the prior 250 (left hand plot) and 50 (right hand plot) trading dates. The lower plots then depict the investment fractions of an investor with a level of relative risk aversion of 6 ($\alpha=-5$) (and $T=14$ years) who relies on $\overline \mu=0.09$ and $\underline\mu=0.03$ with a prior probability $p=0.5$ (black line). The learning strategy $\kappa$ is compared to a {\it{naive}} investor who uses $\overline\mu=\underline\mu=0.9$ (dashed gray) and one with $\overline\mu=\underline\mu=0.3$ (gray).}
\label{Fig_imp}
\end{figure}
\begin{figure}[tb]
\begin{center}{\bf{Comparison of Bayesian strategies (different $(\overline\mu,\underline\mu)$ and $p$)}}
\end{center}
	\begin{center}
\includegraphics[width=0.45\textwidth]{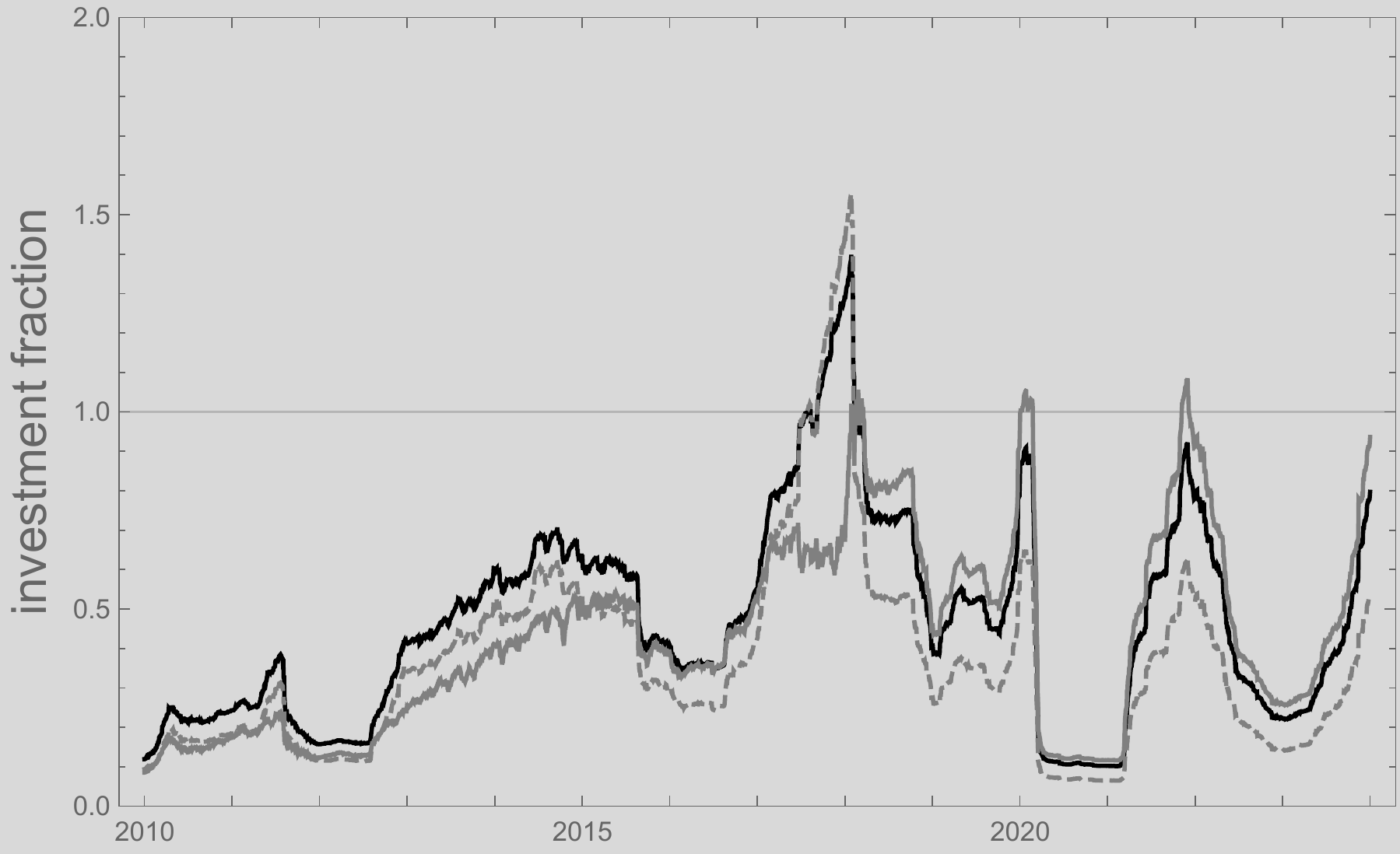}
\includegraphics[width=0.45\textwidth]{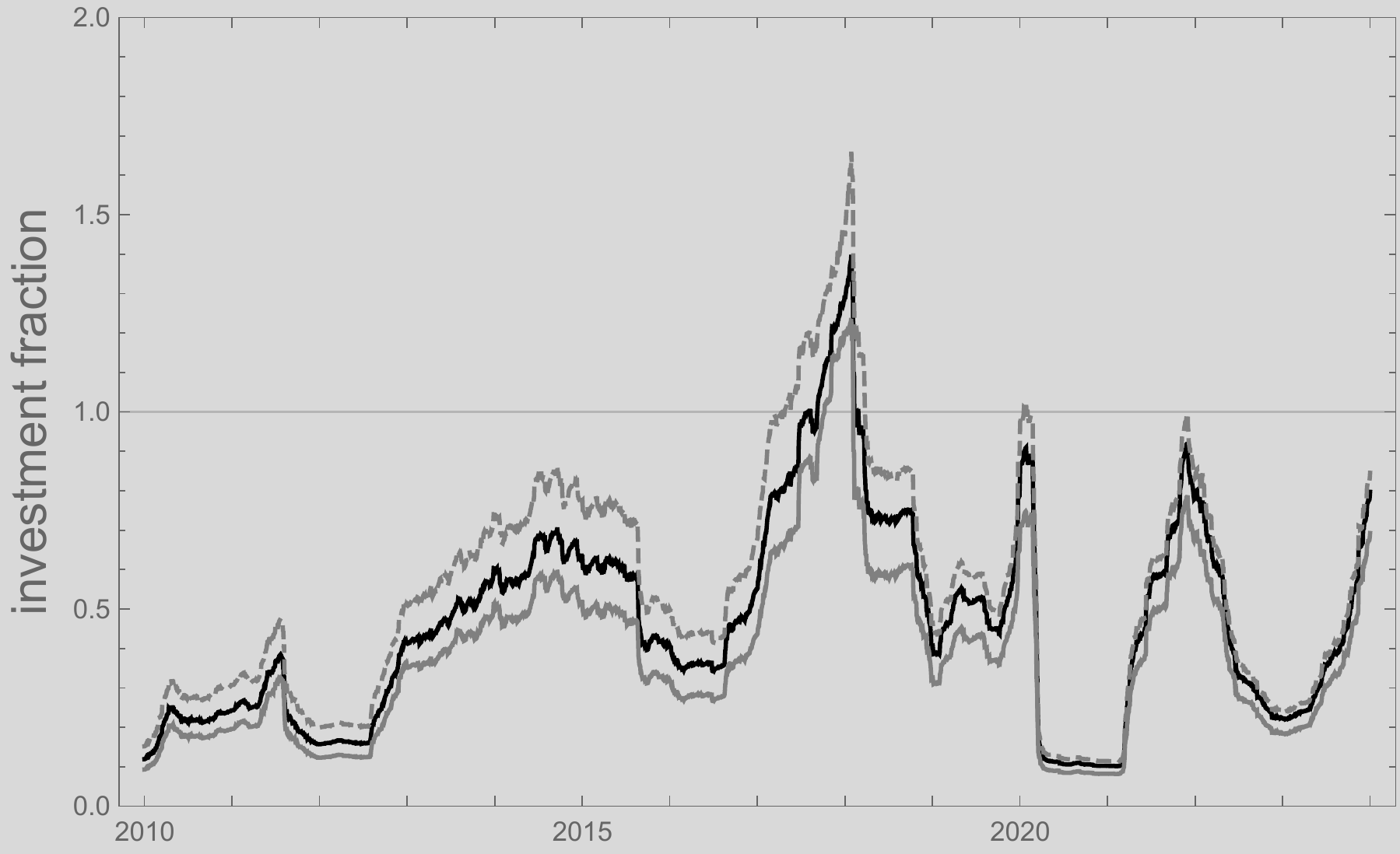}
\end{center}
\caption{The figure illustrates variations of the investment fraction of the benchmark investor ($(\overline\mu,\underline\mu)=(0.09;0.03)$ and $p=0.5$) considered in Figure \ref{Fig_imp}. The benchmark is indicated by the black line. In the left plot, the gray (dashed) line is implied by $(\overline\mu,\underline\mu)=(0.1;0.01)$ ( $(\overline\mu,\underline\mu)=(0.06;0.02)$). The right plot uses modified a priori probabilities, i.e. $p=0.25$ (gray line) and $p=0.75$ (dashed line).}
\label{Fig_imp2}
\end{figure}

Once  $\hat\Sigma$ and the prior distribution of $\mu$ are calculated by means of the training data,   the investor can initialize  the strategy with $Y(0)=0$. Afterwards (at time $t$), it is necessary to plug-in the observed cummulated increments $Y_{t-1}+\Delta \hat Y_t$.
Notice that the increments are based on the estimate $\hat\sigma$ (the estimated covariance matrix $\hat\Sigma$, respectively).
Thus, a practical application may also reconsider the estimated covariance matrix $\hat\Sigma$ and  the possible (two point) values of $\mu$.

To get some further intuition about the implementation and the problems herein, consider e.g. 
the case of one risky asset where 
$S_t=S_0\exp(\sigma W_t+\mu t-\frac12 \sigma^2 t)=S_0\exp(\sigma Y_t -\frac12 \sigma^2 t)$
such that
$$ Y_t-Y_{t-1} = \frac1\sigma\ln (S_t/S_{t-1})+\frac12 \sigma \approx  \frac1\sigma (S_t/S_{t-1}-1)+\frac12 \sigma   $$


For a given (discounted) price path of the risky asset (e.g. with daily observations),  a naive way is to plug for $\hat\sigma$ the historical volatility of the price path over a rolling window of prior trading dates. To illustrate the implementation, assume that the (discounted) price path of the risky asset is given by the S\&P500 index values from $2010.01.01$ to $2024.01.01$
(cf. Figure \ref{Fig_SP500}).
The upper plots of Figure \ref{Fig_imp} illustrate the corresponding historical volatilities based on the last 250 (and 50) trading dates.
The resulting $Y$ pathes are depicted in the middle plots. It is worth to emphasize that the asset price path (S\&P500 index path) is one realization of the {\it{true data generating process}} such that the corresponding path of $Y$ only depends on the volatility estimator.
Now consider the path of the investment fraction $\kappa$ of an Bayesian investor with constant relative risk aversion $1-\alpha=6$ and an investment horizon of $T=14$ years. In addition to $Y(t)$, her investment fraction  depends on her assumptions about  $(\underline\mu,\overline\mu)$, and her prior probability $p$.
 For both volatility windows (250 and 50 days),  the lower plots of Figure \ref{Fig_imp} illustrate the investment fraction $\kappa$ for $\overline \mu=0.09$ and $\underline\mu=0.03$ with a prior probability $p=0.5$ (black line). The strategy $\kappa$ is then compared to a {\it{naive}} investor who uses $\overline\mu=\underline\mu=0.9$ (gray dashed line) and one with $\overline\mu=\underline\mu=0.3$ (gray lines).
 Figure \ref{Fig_imp2} then illustrates variations of the investment fraction of the benchmark investor (i.e. variations of $(\overline\mu,\underline\mu)=(0.09;0.03)$ and $p=0.5$).
The question which presents itself concerns the robustness of the Bayesian strategies and their performance w.r.t. the distributional assumption posed on $\mu$ (and the volatility estimation).
A thorough empirical analysis (or Monte Carlo study based on return data) which also incorporates the impact of ambiguity aversion (and a suitable out-of-sample performance measure to compare our strategies with other ones) is beyond the scope of the paper and left for future research.


\section{Conclusion} We consider a classical multi-asset investment problem within a Black Scholes market with terminal utility of CRRA type. In contrast to established literature we include ambiguity aversion by means of the KMM (\cite{klibanoff2005smooth}) approach  where we assume that the drift of the stocks is not known (model ambiguity), but only a prior distribution is given and the investor is able to learn about the drifts by observing the stock prices. Thus, a second utility function of CRRA type for ambiguity aversion is included. We show analytically how problems of this type can be reduced to the solution of a classical Bayesian investment problem with adapted prior. Based on this result we are able to carry out an extensive numerical study where we discuss the impact of ambiguity preferences. It turns out that investors choose less risk (in terms of a higher weight on the lower Merton ratio) under ambiguity aversion. We consider in particular the short-term and the long-term investment behavior. For the long-term investment it turns out that only the risk aversion of the investor matters  and the investor behaves like one who knows the drift and this is either the most optimistic one (if she is less risk averse than the log-investor) or the most pessimistic one (if she is more risk averse than the log-investor) independent from model ambiguity. Whereas it is clear what this means in the single-asset case, in the multi-asset scenario we first have to figure out what the worst and best drift parameters indeed are. We have shown that the Euclidean norm of the drift vector is here the relevant quantity. For the short-term horizon investment,  only the prior distribution is important and the solution is given by the average of within regime Merton solutions. This is again independent from model ambiguity. We indicate that an approach like this may be generalized to situations where uncertainty and ambiguity are measured by other means.

\section{Appendix}
\subsection{Duality results}\label{app:duality}
Throughout we assume here that $X\ge 0$ and all integrals exist. The duality result for the $L^\mathbf{p}$ norm (Lemma \ref{lem:duality_1}) is well-known \cite{rudin1991functional} and we do not repeat a proof here.

\subsubsection{Dual representation for $0<\mathbf{p}<1$}
Again let $\mathbf{q}$ be  such that $1/\mathbf{p}+1/\mathbf{q}=1$ which implies $\mathbf{q}<0$ and let $\mathfrak{Q}'$ be as defined in Lemma \ref{lem:duality2}. In order to show the statement we first need the following variant of the H\"older inequality.

\begin{lemma}\label{lem:holder}
 Let $X,Y$ be two non-negative random variables  such that all integrals exist.
 Then
 $$ \int XY d\PP \ge \Big( \int X^\mathbf{p} d\PP \Big)^{1/\mathbf{p}} \Big( \int Y^\mathbf{q}d\PP \Big)^{1/\mathbf{q}}. $$
\end{lemma}

\begin{proof}
    Set $Z:= XY$ and w.l.o.g.\ we assume that $\int Z d\PP=1$ and $\Big( \int Y^\mathbf{q}d\PP \Big)^{1/\mathbf{q}}=1.$ Otherwise multiply the inequality by appropriate constants.  We then have to show that
$$ \int (Z/Y)^\mathbf{p} d\PP \le 1. $$
This can be done with the usual H\"older inequality. In what follows consider $Z^\mathbf{p}\in L^{1/\mathbf{p}}$ and $(1/Y)^\mathbf{p} \in L^{1/(1- \mathbf{p})}.$ We then obtain:
\begin{align*}
    \int Z^\mathbf{p} (1/Y)^\mathbf{p} d \PP  \le & \Big( \int (Z^\mathbf{p})^{1/\mathbf{p}} d\PP \Big)^{\mathbf{p}} \Big( \int (1/Y)^{\mathbf{p}/(1-\mathbf{p})}d\PP \Big)^{1-\mathbf{p}}\\
  =  &  \Big( \int Z d\PP \Big)^{\mathbf{p}}  \Big( (\int Y^\mathbf{q} d\PP )^{1/\mathbf{q}}\Big)^{(1-\mathbf{p})\mathbf{q}}=1.
\end{align*}
The last equation follows from our normalization.
\end{proof}

Now we prove Lemma \ref{lem:duality2}. In order to do this we show that for every $\Q\in \mathfrak{Q}'$ we have that 
$$ \int X d \Q  \ge \left(\int  X^\mathbf{p} d\PP\right)^{1/\mathbf{p}} $$
and there exists a $\Q^* \in \mathfrak{Q}'$ which attains equality. Thus, let first be $\Q\in \mathfrak{Q}'$. Then using Lemma \ref{lem:holder}
\begin{align*}
     \int X d \Q &=  \int X   \frac{d\Q}{d\PP} d \PP \ge \Big( \int X^\mathbf{p} d\PP \Big)^{1/\mathbf{p}} \Big( \int \Big(\frac{d\Q}{d\PP}\Big)^\mathbf{q}d\PP \Big)^{1/\mathbf{q}} \ge \Big( \int X^\mathbf{p} d\PP \Big)^{1/\mathbf{p}}.
\end{align*}
Next define $\Q^*$ by 
$$ \frac{d\Q^*}{d\PP} =  \Big(\int X^\mathbf{p} d\PP\Big)^{-1/\mathbf{q}}X^{1/(\mathbf{q}-1)} \ge 0.$$
Then $\Q^*\in \mathfrak{Q}'$ since
\begin{align*}
  \int   \Big(\frac{d\Q^*}{d\PP}\Big)^{\mathbf{q}} d\PP &= \Big(\int X^\mathbf{p} d\PP \Big)^{-1}\int X^{\mathbf{q}/(\mathbf{q}-1)} d\PP = 1
\end{align*}
where $\mathbf{q}/(\mathbf{q}-1)=\mathbf{p}.$ Moreover, 
\begin{align*}
 \int X d \Q^*  &=   \int X   \frac{d\Q^*}{d\PP} d \PP  =  \Big(\int X^\mathbf{p} d\PP\Big)^{-1/\mathbf{q}}  \int X X^{1/(\mathbf{q}-1)}   d \PP\\
 &= \frac{\int  X^{\mathbf{q}/(\mathbf{q}-1)}   d \PP}{\Big(\int X^\mathbf{p} d\PP\Big)^{1/\mathbf{q}} } =  \frac{\int  X^{\mathbf{p}}   d \PP}{\Big(\int X^\mathbf{p} d\PP\Big)^{1/\mathbf{q}} } = \left(\int  X^\mathbf{p} d\PP\right)^{1/\mathbf{p}}
\end{align*}
which concludes the proof.

\subsubsection{Dual representation for $\mathbf{p}<0$}
Let $0<\mathbf{q}<1$ be  such that $1/\mathbf{p}+1/\mathbf{q}=1$ and
$ \mathfrak{Q}'$ as before. In this case Lemma \ref{lem:duality2} holds too.
The proof is the same as before where we have to interchange the role of $X$ and $Y$ in the application of Lemma \ref{lem:holder}.

\subsubsection{Summary of different cases}
For the following table let $$F(\pi):=  \left(  \EE \left[\left( \EE_\Theta [(X_T^\pi)^\alpha]\right)^{\lambda/\alpha} \right]\right)^{\alpha/\lambda}, \quad G(\pi,\Q) :=\int \EE_{\vartheta} [(X_T^\pi)^\alpha] \Q(d\vartheta) . $$
Problem \eqref{eq:Aproblem} is equivalent to
$$\begin{array}{cc|c|c}
 & & \lambda >0 & \lambda < 0\\ \hline
  & \lambda/\alpha > 1 &  \sup_\pi F(\pi)  & \\
  & & = \sup_\pi \sup_{\Q \in \mathfrak{Q} }G(\pi,\Q) & \sup_\pi  F(\pi) \\
  \alpha >0 &&& \\
  & \lambda/\alpha < 1  & \sup_\pi F(\pi) & = \sup_\pi \inf_{\Q \in \mathfrak{Q}' } G(\pi,\Q)\\
  & & = \sup_\pi \inf_{\Q \in \mathfrak{Q}' } G(\pi,\Q) &\\ \hline
  & \lambda/\alpha > 1 &  & \inf_\pi  F(\pi)  \\
  &&  \inf_\pi F(\pi) & = \inf_\pi \sup_{\Q \in \mathfrak{Q} } G(\pi,\Q) \\
  \alpha < 0 && &\\
  && = \inf_\pi \inf_{\Q \in \mathfrak{Q}' } G(\pi,\Q  &  \inf_\pi  F(\pi) \\
  &\lambda/\alpha < 1 && = \inf_\pi \inf_{\Q \in \mathfrak{Q}' } G(\pi,\Q)\\ \hline
  \end{array}$$
  \vspace*{0.4cm}
  
  In all cases sup and inf can be interchanged.  

\subsection{Proof of Theorem \ref{theo:sensitivity}}\label{app:sensi}
Part a) follows directly by inspecting the expression in Theorem \ref{theo:Bayes}.  Note that for $T\to 0 $ the density $\varphi_{T-t}$ concentrates on the Dirac measure on zero and thus the integral vanishes.

Next we prove b).
The optimal fraction invested in stock $i$ at time $t$ given observation $y\in\R$ can more explicitly be written as (we denote $(\sigma^\top)^{-1}=: \tilde{\sigma}$):
\begin{align}\label{eq:optimalpi} \kappa_i(t,T,y) = \gamma  \frac{\int_{\R^d} \sum_{k=1}^m (\tilde{\sigma})_i \cdot \vartheta_k p_k L_T(\vartheta_k,y+z) F(T,y+z)^{\gamma-1}\varphi_{T-t}(z)dz}{\int_{\R^d} F(T,y+z)^\gamma\varphi_{T-t}(z)dz}\end{align}
where $(\tilde{\sigma})_i \cdot \vartheta_k = \sum_{j=1}^d (\tilde{\sigma})_{ij}  \vartheta_{kj}$ and
$\varphi_T(z) = (2\pi T)^{-d/2} e^{- \|z\|^2/2T}$ is the density of $N(0,TI)$.

We will show the statement for $t=0$ and $y=0$. The proof of the general case is similar.
First suppose that $\alpha\in (0,1)$. Define for $k=1,\ldots, m$:
$$ f_k(T) := \frac{\int_{\R^d} p_k L_T(\vartheta_k,z) F(T,z)^{\gamma-1}\varphi_{T}(z)dz}{\int_{\R^d} F(T,z)^{\gamma}\varphi_{T}(z)dz}.$$
Obviously we have by definition of $L_T$ and $F$ that $0\le f_k(T)$ and $\sum_{k=1}^m f_k(T)=1$.  Thus, it is enough to show that $\lim_{T\to\infty} f_m(T)=1$. Let us now consider the following inequality where the integral in the last equation can be computed with the formula of the moment generating function of a multivariate normal distribution.
\begin{eqnarray*}
f_m(T) &\ge& \frac{\int_{\R^d} p_m L_T(\vartheta_m,z) (p_m L_T(\vartheta_m,z))^{\gamma-1}\varphi_{T}(z)dz}{\int_{\R^d} F(T,z)^\gamma\varphi_{T}(z)dz}\\
&=& \frac{\int_{\R^d}  (p_m L_T(\vartheta_m,z))^\gamma\varphi_{T}(z)dz}{\int_{\R^d} F(T,z)^\gamma\varphi_{T}(z)dz}= \frac{p_m^\gamma \exp\big(\frac12 \|\vartheta_m\|^2 T\gamma(\gamma-1)\big)}{\int_{\R^d} F(T,z)^\gamma\varphi_{T}(z)dz}.\end{eqnarray*}
We will show that the lower bound tends to $1$ for $T\to\infty$. In what follows we consider the denominator. We can write it as $\mathbb{E}[F(T,Z)^\gamma]$ with $Z\sim \mathcal{N}(0,TI)$. Let $A$ be an $(m,d)$-matrix with rows consisting of $\vartheta_1,\ldots,\vartheta_m$. Then $X:=AZ\sim \mathcal{N}(0,T AA^\top)$. In particular the marginal distribution is given by $X_j \sim \mathcal{N}(0,T \|\vartheta_j\|^2)$ and we can write
$$ \mathbb{E}[F(T,Z)^\gamma] = \mathbb{E}\Big[\Big( \sum_{k=1}^m p_k \exp(X_k -\frac12 \|\vartheta_k\|^2 T)\Big)^\gamma\Big].$$
Now let $\gamma\in \mathbb{Q},$ i.e.\  we can write $\gamma=\frac{n}{l}$. Recall that $\gamma = 1/(1-\alpha)>1$ in this case. Then we obtain with $\beta=(\beta_1,\ldots,\beta_m)\in\N_0^m$ and with the notation $|\beta|=\beta_1+\ldots +\beta_m$ and
$$ { n \choose \beta_1,\ldots, \beta_m} = \frac{n!}{\beta_1!\ldots \beta_m!} $$
using the multinomial formula that
\begin{eqnarray*}
&&\Big[\Big( \sum_{k=1}^m p_k \exp(X_k -\frac12 \|\vartheta_k\|^2 T)\Big)^n\Big]^\frac{1}{l} =\\
&&  \Big[ \sum_{|\beta|=n} { n \choose \beta_1,\ldots, \beta_m} p_1^{\beta_1}\exp\big(\beta_1 (X_1 -\frac12 \|\vartheta_1\|^2 T)\big)\ldots p_m^{\beta_m}\exp\big(\beta_m (X_m -\frac12 \|\vartheta_m\|^2 T)\big) \Big]^\frac{1}{l}.
\end{eqnarray*}
Now since $(x_1+\ldots +x_K)^\frac{1}{l} \le x_1^\frac{1}{l}+\ldots +x_K^\frac{1}{l}$ for $x_i\ge 0$ we further obtain
\begin{eqnarray*}
&&  \Big[ \sum_{|\beta|=n} { n \choose \beta_1,\ldots, \beta_m} p_1^{\beta_1}\exp\big(\beta_1 (X_1 -\frac12 \|\vartheta_1\|^2 T)\big)\ldots p_m^{\beta_m}\exp\big(\beta_m (X_m -\frac12 \|\vartheta_m\|^2 T)\big) \Big]^\frac{1}{l}\\
&\le &  \sum_{|\beta|=n} { n \choose \beta_1,\ldots, \beta_m}^\frac{1}{l} p_1^{\frac{\beta_1}{l}}\exp\big(\frac{\beta_1}{l} (X_1 -\frac12 \|\vartheta_1\|^2 T)\big)\ldots p_m^{\frac{\beta_m}{l}}\exp\big(\frac{\beta_m}{l} (X_m -\frac12 \|\vartheta_m\|^2 T)\big)\\
&=& p_1^\gamma \exp\big(\gamma (X_1 -\frac12 \|\vartheta_1\|^2 T)\big)+\ldots + p_m^\gamma \exp\big(\gamma (X_m -\frac12 \|\vartheta_m\|^2 T)\big)\\
&& +  \sum_{|\beta|=n \atop \beta_i\neq n} C_\beta\exp\big(\frac{\beta_1}{l} (X_1 -\frac12 \|\vartheta_1\|^2 T)+\ldots +\frac{\beta_m}{l} (X_m -\frac12 \|\vartheta_m\|^2 T)\big),
\end{eqnarray*}
for some constants $C_\beta$ which do not depended on $T$. The last summands can be written as
$$ \exp\big(\frac{\beta_1}{l} (\sqrt{T} \|\vartheta_1\| \tilde{X}_1 -\frac12 \|\vartheta_1\|^2 T)+\ldots +\frac{\beta_m}{l} (\sqrt{T} \|\vartheta_m\|  \tilde{X}_m -\frac12 \|\vartheta_m\|^2 T)\big) $$
with marginal distribution $\tilde{X}_i\sim \mathcal{N}(0,1)$. Since the function $$(x_1,\ldots,x_m) \mapsto \exp\big(\frac{\beta_1}{l} (\sqrt{T} \|\vartheta_1\| x_1 -\frac12 \|\vartheta_1\|^2 T)+\ldots +\frac{\beta_m}{l} (\sqrt{T} \|\vartheta_m\|  x_m -\frac12 \|\vartheta_m\|^2 T)\big) $$ is supermodular  (follows e.g.\ with Lemma 2.1 in \cite{bauerle1997inequalities}) we obtain with the Lorentz-inequality (see e.g. Lemma 2.4 a) in \cite{bauerle1997inequalities})
\begin{eqnarray*}
&&\mathbb{E}\Big[  \exp\big(\frac{\beta_1}{l} (\sqrt{T} \|\vartheta_1\| \tilde{X}_1 -\frac12 \|\vartheta_1\|^2 T)+\ldots +\frac{\beta_m}{l} (\sqrt{T} \|\vartheta_m\|  \tilde{X}_m -\frac12 \|\vartheta_m\|^2 T)\big)  \Big] \\
&\le & \mathbb{E}\Big[  \exp\big(\frac{\beta_1}{l} (\sqrt{T} \|\vartheta_1\| {X} -\frac12 \|\vartheta_1\|^2 T)+\ldots +\frac{\beta_m}{l} (\sqrt{T} \|\vartheta_m\|  {X} -\frac12 \|\vartheta_m\|^2 T)\big)  \Big]
\end{eqnarray*}
with the same random variable $X \sim \mathcal{N}(0,1)$.
Now taking the expectation and using the formula of the moment generating function of a normal distribution yields
\begin{eqnarray}
\nonumber&& \mathbb{E}[F(T,Z)^\gamma]   \le  p_1^\gamma \exp\big(\frac12 \|\vartheta_1\|^2 \gamma(\gamma-1) T)\big)+\ldots + p_m^\gamma \exp\big(\frac12 \|\vartheta_m\|^2\gamma(\gamma-1) T)\big)\\
\label{eq:ineqproof}&& +  \sum_{|\beta|=n \atop \beta_i\neq n} C_\beta\exp\Big(\frac12 T \Big[ (\|\vartheta_1\| \beta_1+\ldots + \|\vartheta_m\|\beta_m)^2 \frac{1}{l^2}-(\beta_1\|\vartheta_1\|^2+\ldots \beta_m \|\vartheta_m\|^2)\frac{1}{l}\Big]\Big).
\end{eqnarray}
Let us now consider the exponent in the last line for an arbitrary admissible $\beta$ without the factor $\frac12 T$ in front. Obviously we can choose numbers $\|\vartheta^*\| $ and $\|\bar{\vartheta}\| $ such that
\begin{eqnarray*}
\|\vartheta_1\| \beta_1+\ldots + \|\vartheta_m\|\beta_m &=& \|\bar{\vartheta}\|  |\beta|\\
\beta_1\|\vartheta_1\|^2+\ldots \beta_m \|\vartheta_m\|^2 &=& \|\vartheta^*\|^2  |\beta|.
\end{eqnarray*}
This implies
\begin{eqnarray*}
\|\vartheta^*\| &=& \sqrt{\frac{\beta_1}{ |\beta|}\|\vartheta_1\|^2+\ldots \frac{\beta_m}{ |\beta|} \|\vartheta_m\|^2}\\
&\ge & \frac{\beta_1}{ |\beta|}\|\vartheta_1\|+\ldots \frac{\beta_m}{ |\beta|} \|\vartheta_m\| = \|\bar{\vartheta}\|.
\end{eqnarray*}
Moreover we have that $\|\vartheta^*\| < \|\vartheta_m\|$ since at least two $\beta_i$ are non-zero. This implies
\begin{eqnarray*}
&&  (\|\vartheta_1\| \beta_1+\ldots + \|\vartheta_m\|\beta_m)^2 \frac{1}{l^2}-(\beta_1\|\vartheta_1\|^2+\ldots \beta_m \|\vartheta_m\|^2)\frac{1}{l}\\
&\le & \|\vartheta^*\|^2 \frac{ |\beta|^2}{l^2}-\|\vartheta^*\|^2 \frac{ |\beta|}{l}= \|\vartheta^*\|^2 \gamma(\gamma-1)<   \|\vartheta_m\|^2 \gamma(\gamma-1).
\end{eqnarray*}

This shows us that all summands of the upper bound are of the form $\exp(\frac12 Tc)$ with the largest $c= \|\vartheta_m\|^2  \gamma(\gamma-1)$. Thus we obtain
\begin{eqnarray*}
 1&\ge& \lim_{T\to\infty} f_m(T) \\
&\ge& \lim_{T\to\infty}  \frac{p_m^\gamma \exp\big(\frac12 T \|\vartheta_m\|^2 \gamma(\gamma-1)\big)}{\sum_k p_k^\gamma \exp\big(\frac12 T \|\vartheta_k\|^2 \gamma(\gamma-1) )\big)+\exp\big(\frac12T \|\vartheta^*\|^2 \gamma(\gamma-1)\big) \sum_{|\beta|=n \atop \beta_i\neq n} C_\beta}=1
 \end{eqnarray*}
which implies the statement for $\gamma\in \mathbb{Q}$. Since the expression is continuous in $\gamma$ we obtain the statement for all $\alpha\in (0,1)$.

Part c) The proof for the case $\alpha<0$ can be done similar. In this case we have to show that $\lim_{T\to\infty} f_1(T)=1$. Note that here $\gamma := \frac{1}{1-\alpha}\in (0,1).$ We start with the similar inequality
\begin{eqnarray*}
f_1(T) &\ge& \frac{\int_{\R^d} p_1 L_T(\vartheta_1,z) (p_1 L_T(\vartheta_1,z))^{\gamma-1}\varphi_{T}(z)dz}{\int_{\R^d} F(T,z)^\gamma\varphi_{T}(z)dz}\\
&=&  \frac{p_1^\gamma \exp\big(\frac12 \|\vartheta_1\|^2 T\gamma (\gamma-1)\big)}{\int_{\R^d} F(T,z)^\gamma\varphi_{T}(z)dz}\end{eqnarray*}
For the denominator we can use the same lines of inequality until \eqref{eq:ineqproof}.
Defining  $\|\vartheta^*\|$ in the same way we obtain $\|\vartheta^*\| > \|\vartheta_1\|$ and since $\gamma\in (0,1)$ that
$$ \|\vartheta^*\|^2 \gamma(\gamma-1)<   \|\vartheta_1\|^2 \gamma(\gamma-1).$$ Looking again for the highest exponents in the denominator we obtain the statement as in part b).

Part d) follows from the representation 
$$ \kappa_i(t,T,y) = \gamma \sum_{k=1}^m \tilde\sigma_i \cdot \vartheta_k f_k(T)$$
in the beginning of part b).

Part e) can be obtained by direct calculation.

\bibliographystyle{apalike}
\bibliography{literature_BSA}

\begin{thebibliography}{}

\bibitem[Baillon et~al., 2018]{baillon2018effect}
Baillon, A., Bleichrodt, H., Keskin, U., l’Haridon, O., and Li, C. (2018).
\newblock The effect of learning on ambiguity attitudes.
\newblock {\em Management Science}, 64(5):2181--2198.

\bibitem[Balter et~al., 2021]{balter2021time}
Balter, A.~G., Mahayni, A., and Schweizer, N. (2021).
\newblock Time-consistency of optimal investment under smooth ambiguity.
\newblock {\em European Journal of Operational Research}, 293(2):643--657.

\bibitem[Bauder et~al., 2021]{bauder2021bayesian}
Bauder, D., Bodnar, T., Parolya, N., and Schmid, W. (2021).
\newblock Bayesian mean--variance analysis: optimal portfolio selection under
  parameter uncertainty.
\newblock {\em Quantitative Finance}, 21(2):221--242.

\bibitem[B{\"a}uerle, 1997]{bauerle1997inequalities}
B{\"a}uerle, N. (1997).
\newblock Inequalities for stochastic models via supermodular orderings.
\newblock {\em Stochastic Models}, 13(1):181--201.

\bibitem[B{\"a}uerle and Grether, 2017]{bauerle2017extremal}
B{\"a}uerle, N. and Grether, S. (2017).
\newblock Extremal behavior of long-term investors with power utility.
\newblock {\em International Journal of Theoretical and Applied Finance},
  20(05):1750029.

\bibitem[B{\"a}uerle and Rieder, 2020]{bauerle2019markov}
B{\"a}uerle, N. and Rieder, U. (2020).
\newblock Markov decision processes under ambiguity.
\newblock {\em Banach Center Publications}.

\bibitem[Bj{\"o}rk et~al., 2010]{bjork2010optimal}
Bj{\"o}rk, T., Davis, M.~H., and Land{\'e}n, C. (2010).
\newblock Optimal investment under partial information.
\newblock {\em Mathematical Methods of Operations Research}, 71(2):371--399.

\bibitem[Branger et~al., 2023]{bmbo2022}
Branger, N., Becker, L., Mahayni, A., and Offermann, S. (2023).
\newblock On the impacts of time inconsistency in optimal asset allocation
  problems.
\newblock {\em Preprint}.

\bibitem[Brennan, 1998]{brennan1998role}
Brennan, M.~J. (1998).
\newblock The role of learning in dynamic portfolio decisions.
\newblock {\em Review of Finance}, 1(3):295--306.

\bibitem[Chen and Epstein, 2002]{chen2002ambiguity}
Chen, Z. and Epstein, L. (2002).
\newblock Ambiguity, risk, and asset returns in continuous time.
\newblock {\em Econometrica}, 70(4):1403--1443.

\bibitem[Desmettre and Steffensen, 2021]{desmettre2021optimal}
Desmettre, S. and Steffensen, M. (2021).
\newblock Optimal investment with uncertain risk aversion.
\newblock {\em Available at SSRN 3805069}.

\bibitem[Ellsberg, 1961]{ellsberg1961risk}
Ellsberg, D. (1961).
\newblock Risk, ambiguity, and the {S}avage axioms.
\newblock {\em The Quarterly Journal of Economics}, 75(4):643--669.

\bibitem[Epstein and Schneider, 2007]{epstein2007learning}
Epstein, L.~G. and Schneider, M. (2007).
\newblock Learning under ambiguity.
\newblock {\em The Review of Economic Studies}, 74(4):1275--1303.

\bibitem[Gennotte, 1986]{gennotte1986optimal}
Gennotte, G. (1986).
\newblock Optimal portfolio choice under incomplete information.
\newblock {\em The Journal of Finance}, 41(3):733--746.

\bibitem[Gilboa and Schmeidler, 2004]{gilboa2004maxmin}
Gilboa, I. and Schmeidler, D. (2004).
\newblock Maxmin expected utility with non-unique prior.
\newblock In {\em Uncertainty in economic theory}, pages 141--151. Routledge.

\bibitem[Gollier, 2011]{gollier2011portfolio}
Gollier, C. (2011).
\newblock Portfolio choices and asset prices: The comparative statics of
  ambiguity aversion.
\newblock {\em The Review of Economic Studies}, 78(4):1329--1344.

\bibitem[Guan and Li, 2022]{guan2022equilibrium}
Guan, G. and Li, B. (2022).
\newblock Equilibrium investment and reinsurance strategies under smooth
  ambiguity with a general second-order distribution.
\newblock {\em Journal of Economic Dynamics and Control}, 143:104515.

\bibitem[Guan et~al., 2023]{glx}
Guan, G., Liang, Z., and Xia, J. (2023).
\newblock Equilibrium portfolio selection for smooth ambiguity preferences.
\newblock {\em arXiv:2302.08181}.

\bibitem[Guidolin and Rinaldi, 2013]{guidolin2013ambiguity}
Guidolin, M. and Rinaldi, F. (2013).
\newblock Ambiguity in asset pricing and portfolio choice: A review of the
  literature.
\newblock {\em Theory and Decision}, 74(2):183--217.

\bibitem[Hansen and Sargent, 2001]{hansen2001robust}
Hansen, L. and Sargent, T.~J. (2001).
\newblock Robust control and model uncertainty.
\newblock {\em American Economic Review}, 91(2):60--66.

\bibitem[Honda, 2003]{honda2003optimal}
Honda, T. (2003).
\newblock Optimal portfolio choice for unobservable and regime-switching mean
  returns.
\newblock {\em Journal of Economic Dynamics and Control}, 28(1):45--78.

\bibitem[Iwaki and Osaki, 2014]{iwaki2014dual}
Iwaki, H. and Osaki, Y. (2014).
\newblock The dual theory of the smooth ambiguity model.
\newblock {\em Economic Theory}, 56(2):275--289.

\bibitem[Ju and Miao, 2012]{ju2012ambiguity}
Ju, N. and Miao, J. (2012).
\newblock Ambiguity, learning, and asset returns.
\newblock {\em Econometrica}, 80(2):559--591.

\bibitem[Kan and Zhou, 2007]{kan2007optimal}
Kan, R. and Zhou, G. (2007).
\newblock Optimal portfolio choice with parameter uncertainty.
\newblock {\em Journal of Financial and Quantitative Analysis}, 42(3):621--656.

\bibitem[Karatzas et~al., 1991]{karatzas1991martingale}
Karatzas, I., Lehoczky, J.~P., Shreve, S.~E., and Xu, G.-L. (1991).
\newblock Martingale and duality methods for utility maximization in an
  incomplete market.
\newblock {\em SIAM Journal on Control and optimization}, 29(3):702--730.

\bibitem[Karatzas and Zhao, 2001]{karatzas2001bayesian}
Karatzas, I. and Zhao, X. (2001).
\newblock Bayesian adaptive portfolio optimization.
\newblock {\em Option pricing, interest rates and risk management}, pages
  632--669.

\bibitem[Klibanoff et~al., 2005]{klibanoff2005smooth}
Klibanoff, P., Marinacci, M., and Mukerji, S. (2005).
\newblock A smooth model of decision making under ambiguity.
\newblock {\em Econometrica}, 73(6):1849--1892.

\bibitem[Lakner, 1995]{lakner1995utility}
Lakner, P. (1995).
\newblock Utility maximization with partial information.
\newblock {\em Stochastic Processes and their Applications}, 56(2):247--273.

\bibitem[Ledoit and Wolf, 2017]{ledoit2017nonlinear}
Ledoit, O. and Wolf, M. (2017).
\newblock Nonlinear shrinkage of the covariance matrix for portfolio selection:
  Markowitz meets goldilocks.
\newblock {\em The Review of Financial Studies}, 30(12):4349--4388.

\bibitem[Lin and Riedel, 2021]{lin2021optimal}
Lin, Q. and Riedel, F. (2021).
\newblock Optimal consumption and portfolio choice with ambiguous interest
  rates and volatility.
\newblock {\em Economic Theory}, 71(3):1189--1202.

\bibitem[Longo and Mainini, 2016]{longo2016learning}
Longo, M. and Mainini, A. (2016).
\newblock Learning and portfolio decisions for {CRRA} investors.
\newblock {\em International Journal of Theoretical and Applied Finance},
  19(03):1650018.

\bibitem[Miao, 2009]{miao2009ambiguity}
Miao, J. (2009).
\newblock Ambiguity, risk and portfolio choice under incomplete information.
\newblock {\em Annals of Economics \& Finance}, 10(2).

\bibitem[Rieder and B{\"a}uerle, 2005]{rieder2005portfolio}
Rieder, U. and B{\"a}uerle, N. (2005).
\newblock Portfolio optimization with unobservable {M}arkov-modulated drift
  process.
\newblock {\em Journal of Applied Probability}, 42(2):362--378.

\bibitem[Rudin, 1991]{rudin1991functional}
Rudin, W. (1991).
\newblock {\em Functional analysis, McGraw Hill}.

\bibitem[Schied, 2007]{schied2007optimal}
Schied, A. (2007).
\newblock Optimal investments for risk-and ambiguity-averse preferences: a
  duality approach.
\newblock {\em Finance and Stochastics}, 11(1):107--129.

\bibitem[Schied et~al., 2009]{schied2009robust}
Schied, A., F{\"o}llmer, H., and Weber, S. (2009).
\newblock Robust preferences and robust portfolio choice.
\newblock {\em Handbook of Numerical Analysis}, 15:29--87.

\bibitem[Sion, 1958]{sion1958general}
Sion, M. (1958).
\newblock On general minimax theorems.
\newblock {\em Pacific Journal of Mathematics}, 8(1):171--176.

\bibitem[Skiadas, 2003]{skiadas2003robust}
Skiadas, C. (2003).
\newblock Robust control and recursive utility.
\newblock {\em Finance and Stochastics}, 7(4):475--489.

\bibitem[Skiadas, 2013]{skiadas2013smooth}
Skiadas, C. (2013).
\newblock Smooth ambiguity aversion toward small risks and continuous-time
  recursive utility.
\newblock {\em Journal of Political Economy}, 121(4):775--792.

\bibitem[Suzuki, 2018]{suzuki2018continuous}
Suzuki, M. (2018).
\newblock Continuous-time smooth ambiguity preferences.
\newblock {\em Journal of Economic Dynamics and Control}, 90:30--44.

\bibitem[Yaari, 1987]{yaari1987dual}
Yaari, M.~E. (1987).
\newblock The dual theory of choice under risk.
\newblock {\em Econometrica: Journal of the Econometric Society}, 55:95--115.

\end{thebibliography}
\end{document}